\def\year{2021}\relax
\documentclass[letterpaper]{article} 
\usepackage{aaai21}  
\usepackage{times}  
\usepackage{helvet} 
\usepackage{courier}  
\usepackage[hyphens]{url}  
\usepackage{graphicx} 
\usepackage{natbib}  
\usepackage{caption} 
\usepackage{amsfonts}
\usepackage{amsmath,amssymb}
\usepackage{algorithm}
\usepackage{algorithmic}
\usepackage{diagbox}
\usepackage{makecell}
\usepackage{siunitx}
\usepackage{xcolor}
\usepackage{amsthm}
\usepackage{subcaption}
\usepackage{multirow}
\usepackage{BOONDOX-cal}
\usepackage{tabularx}
\usepackage{booktabs}
\usepackage{nomencl}
\usepackage{enumitem}
\usepackage{longtable}
\usepackage{thmtools}
\usepackage{thm-restate}
\usepackage[labelfont=bf,format=plain,justification=raggedright,singlelinecheck=false]{caption}
\newcommand\percentage[2][round-precision = 1]{
    \SI[round-mode = places,
        scientific-notation = fixed, fixed-exponent = 0,
        output-decimal-marker={.}, #1]{#2e2}{\percent}%
}

\newtheorem{theorem}{Theorem}[section]

\newtheorem{lemma}[theorem]{Lemma}
\newtheorem{definition}[theorem]{Definition}
\newtheorem{result}{Result}
\newtheorem{proposition}[theorem]{Proposition}

\newtheorem{remark}[theorem]{Remark}

\usepackage{xr}
\title{The Causal Learning of Retail Delinquency}
\author {
	Yiyan Huang,\textsuperscript{\rm 2}\thanks{Co-first authors are in alphabetical order.} 
	Cheuk Hang Leung,\textsuperscript{\rm 2}\footnotemark[1] 
	Xing Yan,\textsuperscript{\rm 3} 
	Qi Wu,\textsuperscript{\rm 2}\thanks{Qi Wu is the corresponding author.}\\
	Nanbo Peng,\textsuperscript{\rm 1} 
	Dongdong Wang,\textsuperscript{\rm 1} 
	Zhixiang Huang\textsuperscript{\rm 1}\\
}
\affiliations {
	\textsuperscript{\rm 1}JD Digits, \textsuperscript{\rm 2}City University of Hong Kong,  \textsuperscript{\rm 3}ISBD, Renmin University of China\\
	yiyhuang3-c@my.cityu.edu.hk, \{chleung87, qiwu55\}@cityu.edu.hk, xingyan@ruc.edu.cn\\ \{pengnanbo, wangdongdong9, huangzhixiang\}@jd.com
}
\begin{document}

\maketitle

\begin{abstract}
This paper focuses on the expected difference in borrower's repayment when there is a change in the lender's credit decisions. Classical estimators overlook the confounding effects and hence the estimation error can be magnificent. As such, we propose another approach to construct the estimators such that the error can be greatly reduced. The proposed estimators are shown to be unbiased, consistent, and robust through a combination of theoretical analysis and numerical testing. Moreover, we compare the power of estimating the causal quantities between the classical estimators and the proposed estimators. The comparison is tested across a wide range of models, including linear regression models, tree-based models, and neural network-based models, under different simulated datasets that exhibit different levels of causality, different degrees of nonlinearity, and different distributional properties. Most importantly, we apply our approaches to a large observational dataset provided by a global technology firm that operates in both the e-commerce and the lending business. We find that the relative reduction of estimation error is strikingly substantial if the causal effects are accounted for correctly.
\end{abstract}
\section{Introduction}\label{sec:Introduction}
A growing number of technology conglomerates provide lending services to shoppers who frequent their e-commerce marketplaces. Technology firms have the information advantage that commercial banks lack. A vast amount of proprietary digital footprints are now at their fingertips. Study shows the information content of proxies for human behavior, lifestyle, and living standard in the non-transnational data are just effective, hence highly valuable for default prediction \cite{berg2020rise}. However, managing retail credit risks in online marketplaces differs in one fundamental way from managing credit-card default risks faced by commercial banks. It comes from the pronounced ``action-response" relationship between the lender's credit decisions and the borrowers' delinquency outcomes. When customers apply for credit to finance their online purchases, they effectively enter into an unsecured loan contract where the counterparty is the platform lender, and they are expected to make installments according to the payment schedules. These loan decisions, especially the loan amount and loan interest, are far more individualized than those made in the traditional credit card business and are frequently adjusted. The e-commerce lenders observe that their credit policies have a systematic causal impact, which alters borrowers' payment behavior, irrespective of age, income, education, occupation, and so forth.

It is conceivable that machine learning (ML) algorithms can effectively tap into the information advantage for accurate estimations of delinquency rate \cite{khandani2010consumer}. However, the existing retail credit risk studies do not recognize the essential of action-response causalities as far as we know. Indeed, obtaining the accurate estimation of delinquency rate does not mean that we can obtain the accurate estimation of the action-response causality. Generally, we face two biases in estimating the action-response causality. The first bias is due to the neglect of the confounding effects. In reality, loan decisions are the results of the lender's decision algorithms that use an overlapping set of borrowers' features with those used for risk assessments, e.g., past shopping and financing records. Thus, ignoring the modelling of the relation between the credit policies and the credit features may cause a big bias in estimating the action-response causality. The second bias comes from the estimation of predictors-response relation. This is a regression bias related to data and ML algorithm, e.g., the sample size, the feature dimensions and the regressor selections.

The goal of this paper is to construct the estimators of the causal parameters which can address both the confounding bias and the regression bias. It is equivalent to finding the score functions with specific conditions (detailed discussions will be presented in the later sections) such that we can recover the estimators from the score functions. To obtain the corresponding estimates of the estimators, we need to estimate the ``counterfactuals'', which are the potential outcomes in delinquent probabilities of borrowers if they were given different amounts of credit lines from the ones they had received. Once the counterfactuals are estimated, we would like to assess both the Average Treatment Effect (ATE) and the Average Treatment Effect on the Treated (ATTE). For new customers, the lender can use the ATE to choose appropriate credit lines for them.  For existing customers, the ATTE allows the lender to gauge the potential changes of risks if their credit lines were changed from the current levels to new ones.

For the rest of the paper, Section \ref{sec:Related Work and Our Contributions} summarizes the related works and our contributions. Section \ref{sec:Background} presents our model setup, and Section \ref{sec:Empirical experiment} presents all the experiments. The paper ends with the conclusion section.
\section{Related Works and Contributions}\label{sec:Related Work and Our Contributions}
\paragraph{Related Works.} In the past credit risk works, the typical supervised learning methods such as direct regression are widely used to estimate the global relationship between the response and the predictors. For example, the regression tree \cite{khandani2010consumer}, Random Forests \cite{malekipirbazari2015risk}, and Recurrent Neural Network \cite{sirignano2016deep} are applied to construct default forecasting models. 

However, when it comes to studying the action-response relationship (e.g., estimating the ATE and ATTE), simply using the typical supervised learning methods to estimate the outcomes for each individual under different interventions and averaging the estimated outcomes for each intervention can produce unsatisfactory results \cite{chernozhukov2018double}, since there is a chance of misspecification of the relationship when confounding effects are present. Some other methods that can account for the confounding effects range from those based on balancing learning and matching \cite{li2017matching,kallus2018balanced,bennett2019policy} to those based on deep learning such as representation learning and adversarial learning \cite{johansson2016learning,yao2018representation,lattimore2016causal,yoon2018ganite}. However, for matching methods such as Propensity Score Matching (PSM) and inverse probability weighting (IPW) (e.g., \cite{hirano2003efficient}), they may amplify the estimation bias if the feature variables are not selected properly or the algorithm is not ideal enough \cite{heckman1998matching}.

To overcome some deficiencies of the above methods, Doubly Robust Estimators (DREs) was proposed \cite{farrell2015robust}. The DREs are recovered from the score functions which incorporate the confounding effects in general \cite{dudik2011doubly}. However, it is not sure if the score functions of the DREs satisfy the \textit{orthogonal condition}, which is defined in \cite{chernozhukov2018double} inspired by \cite{neyman1979c}. Heuristically, those DREs recovered from the score functions which may violate the orthogonal condition (see the detailed discussions in our supplementary) can be sensitive to the nuisance parameters, and hence easily lead to a biased estimation. To solve this problem, some researchers propose the new score functions that satisfy the orthogonal condition \cite{chernozhukov2018double,mackey2018orthogonal,oprescu2019orthogonal}, but they either only consider the binary treatment or derive the theoretical results based on the partially linear model (PLR) setting. To improve that, we not only extend the intervention variable from the binary values to the multiple values, but also consider the fully nonlinear model setting instead of the PLR model setting.


\paragraph{Contributions.}
The contributions of our paper are:
\begin{enumerate}
\item We are the first to show the importance and necessity of considering causality in retail credit risk study using observational records of e-commerce lenders and borrows; 
\item We extend from a partially linear model (e.g., \cite{mackey2018orthogonal, oprescu2019orthogonal}) to a fully nonlinear model. Our estimators recovered from the orthogonal score functions are proved to be regularization unbiased and consistent with an increasing number of population size, thus very suitable for large datasets. Besides, our setup allows the intervention variable to take multiple discrete values, rather than binary values as in \cite{chernozhukov2018double};
\item Our estimators are generic, not only restricted to linear or tree-based methods, but also for any complex methods such as neural network-based models including fully-connected ones (e.g., MLP), convolutional ones (e.g., CNN), and recurrent ones (e.g., GRU);
\item The obtained estimators are robust to model misspecifications when mappings between predictors and outcomes/interventions exhibit different degrees of nonlinearity, and data possess different distributional properties than assumed;
\item We use the parameters in our experiments to control the causality and nonlinearity and show how much error our estimators can correct for compared with direct regression estimators used in the past credit risk works.
\end{enumerate} 
The above points are comprehensively tested through well-designed simulation experiments and verified on a large proprietary real-world dataset via semi-synthetic experiments.
\section{The Model Setup}\label{sec:Background}
\paragraph{The Potential Outcomes.}
Given a probability space $(\Omega,\mathbb{P},\mathcal{F})$, the formulation \eqref{eqt:non-causal model} treats the treatment variable $D$ (or the policy intervention) as part of the explanatory variables $(D,\mathbf{Z})$, where $\mathbf{Z}$ is the feature set. $D$ and $\mathbf{Z}$ are used to regress the response variable (or the outcome) $Y$ such that
\begin{equation}
{\small\begin{aligned}\label{eqt:non-causal model}
	Y&=g(D,\mathbf{Z})+\zeta\quad\text{and}\quad\mathbb{E}\left[\zeta\mid D,\mathbf{Z}\right]=0.
	\end{aligned}}
\end{equation}
Here the outcome $Y$ is a random scalar, the feature set $\mathbf{Z}$ is a random vector, and the intervention $D$ takes discrete values in the set $\{d^{1},\cdots,d^{n}\}$.
$\zeta$ is the scalar noise which summaries all other non-$(D,\mathbf{Z})$-related unknown factors and is assumed to have zero conditional mean given $(D,\mathbf{Z})$. 
The function $g$ is a $\mathbb{P}$-integrable function.
The core of the impact inference is the estimation of the potential outcomes of an individual under the interventions that are different from what we observe. The unobservable potential outcomes are also called the \textit{counterfactuals}. Knowing how to estimate the counterfactuals provides a way to estimate two quantities, both of them are the average ``action-response" relationship between the potential outcome $Y$ and the potential treatment $d^{i}$. Mathematically, they are
\begin{equation*}
{\small\begin{aligned}
&\theta^{i} :=\mathbb{E}\left[g(d^i,\mathbf{Z})\right]\;\;\;\;\text{and}\;\;\;\;\theta^{i\mid j} :=\mathbb{E}\left[g(d^i,\mathbf{Z})\mid D=d^{j}\right].
\end{aligned}}
\end{equation*}

The quantity $\theta^{i}$ means that we want to find the expected outcome when the potential intervention is $d^{i}$. Concurrently, the quantity $\theta^{i\mid j}$ means that given the factual intervention is $d^{j}$, we want to find the expected counterfactual outcome if the intervention $D$ had taken the value $d^{i}$.

\paragraph{Impact Inference without Confounding.}
We begin with the Impact Inference without Confounding (IoC) which accounts for the policy impact but not the confounding effects.
We use $(y_{m},d_{m},\mathbf{z}_{m})$ to represent the observational data associated with the $m^{\textrm{th}}$ customer in the size-$N$ population and use $(y_{m}^{j},d_{m}^{j},\mathbf{z}_{m}^{j})$ to represent the observational data of the $m^{\textrm{th}}$ customer in the sub-population which contains $N_{j}$ customers with observed treatment $d^{j}$. Using sample averaging, we can estimate $\theta^{i}$ as $\frac{1}{N} \overset{N}{\underset{m=1}{\sum}}g(d^i,\mathbf{z}_{m})$.

On the other hand, since
\begin{equation*}
{\small
\begin{aligned}
&\mathbb{E}\left[g(d^i,\mathbf{Z})\mathbf{1}_{\{D=d^{j}\}}\right]=\mathbb{E}\left[\mathbb{E}\left[g(d^i,\mathbf{Z})\mathbf{1}_{\{D=d^{j}\}}\mid D\right]\right]\\
=&\mathbb{E}\left[g(d^i,\mathbf{Z})\mid D=d^{j}\right]\mathbb{E}\left[\mathbf{1}_{\{D=d^{j}\}}\right],
\end{aligned}
}
\end{equation*}
%
we can obtain that $\theta^{i\mid j}=\frac{\mathbb{E}\left[g(d^i,\mathbf{Z})\mathbf{1}_{\{D=d^{j}\}}\right]}{\mathbb{E}\left[\mathbf{1}_{\{D=d^{j}\}}\right]}$. Simultaneously, we use $\frac{1}{N}\overset{N}{\underset{m=1}{\sum}} \mathbf{1}_{\{d_{m}=d^{j}\}} g(d^i,\mathbf{z}_{m})$ to estimate $\mathbb{E}\left[g(d^{i},\mathbf{Z})\mathbf{1}_{\{D=d^{j}\}}\right]$. Consequently, we can estimate $\theta^{i\mid j}$ as $\frac{1}{N_{j}} \overset{N}{\underset{m=1}{\sum}} \mathbf{1}_{\{d_{m}=d^{j}\}}g(d^i,\mathbf{z}_{m})$.
%

Note that $\theta^{i}$ is an average over the whole population $N$, whereas $\theta^{i|j}$ is an average over the sub-population $N_{j}$. To compute the estimates of the two quantities, the key point is to obtain $\hat{g}$ (the estimate of $g$) for every $d^i$ from the observational dataset. In the related literature, specifications of $g$ include both the additive forms \cite{djebbari2008heterogeneous} and the multiplicative forms \cite{hainmueller2019much}. The choices of $\hat{g}$ also include linear models \cite{du2015home,li2017estimation} as well as nonlinear ones. In particular, the predictive advantage of using neural networks to estimate $g$ is demonstrated in \cite{shi2019adapting} where they formulated a similar relationship to estimate the ATE and the ATTE \cite{louizos2017causal,yoon2018ganite}. Other examples can be found in \cite{alaa2018limits,toulis2016long,li2017matching,syrgkanis2019machine}.

Once we obtain $\hat{g}(d^{i},\cdot)$ (the estimate of $g(d^{i},\cdot)$), $\theta^{i}$ and $\theta^{i\mid j}$ can be respectively estimated in the IoC context as
\begin{equation*}
\begin{aligned}
\scalebox{1}{
$\hat{\theta}_{o}^{i}=\frac{1}{N}\overset{N}{\underset{m=1}{\sum}} \hat{g}(d^{i},\mathbf{z}_{m})\;\;\; \text{and}\;\;\;
\hat{\theta}_{o}^{i\mid j}=\frac{1}{N_{j}}\overset{N_{j}}{\underset{m=1}{\sum}} \hat{g}(d^{i},\mathbf{z}_{m}^{j})$}.
\end{aligned}
\end{equation*}
We call them the IoC estimates since we omit the relationship between $D$ and $\mathbf{Z}$, or the so-called confounding effects. Indeed, the confounding effects are obscured from the IoC estimates.
\paragraph{Impact Inference with Confounding.} We then propose the Impact Inference with Confounding (IwC) which can account for both the policy impact and the confounding effects. 
%
Using the IoC estimates to estimate $\theta^i$ and $\theta^{i|j}$ can be misspecified \cite{dudik2011doubly,yuan2019improving} when the confounding effects are present.
The misspecification comes from the fact that, while $\mathbf{Z}$ in \eqref{eqt:non-causal model} affects the outcome $Y$, the intervention $D$ could also be driven by the confounding variable $\mathbf{Z}$. For example, the customer's income level could be a confounding variable in the retail lending context. Customers who receive higher credit lines usually have higher incomes, and higher-income people tend to have lower credit risk. Such examples also widely exist in recommendation systems (e.g., \cite{wang2019doubly,swaminathan2015counterfactual,swaminathan2015self} and references therein). In our paper, we construct better score functions satisfying the orthogonal condition such that the corresponding estimates are assured to be regularization unbiased \cite{chernozhukov2018double}.
 
In order to study the policy impact in the presence of confounding effects, we propose the following formulation for our IwC estimations:
\begin{subequations}
	\begin{align}
	Y&=g(D,\mathbf{U},\mathbf{Z})+\xi,
	\quad\mathbb{E}\left[\xi\mid \mathbf{U},\mathbf{X},\mathbf{Z}\right]=0,\label{eqt:causal model 1}\\
	D&=m(\mathbf{X},\mathbf{Z},\nu),
	\quad\quad\mathbb{E}\left[\nu\mid \mathbf{X},\mathbf{Z}\right]=0,\label{eqt:causal model 2}
	\end{align}
\end{subequations}
where $\mathbf{Z}$ is the confounder variable, $\mathbf{U}$ is the outcome-specific feature set, and $\mathbf{X}$ is the intervention-specific feature set. The map $m$ and the noise term $\nu$ are of the same nature as those of $g$ and $\xi$. We impose no functional restrictions on $g$ and $m$ such that they can be parametric or non-parametric, linear or nonlinear, etc.

If the confounding effects \eqref{eqt:causal model 2} are not recognized but present in the data, the IoC estimates
%
\begin{equation} \label{eqt:IoC estimator for IwC formulation}
{\small
\begin{aligned}
\hat{\theta}_{o}^{i}=\frac{1}{N}\overset{N}{\underset{m=1}{\sum}} \hat{g}(d^{i},\mathbf{u}_{m},\mathbf{z}_{m}),\;\hat{\theta}_{o}^{i\mid j}=\frac{1}{N_{j}}\overset{N_{j}}{\underset{m=1}{\sum}}  \hat{g}(d^{i},\mathbf{u}_{m}^{j},\mathbf{z}_{m}^{j})
\end{aligned}
}
\end{equation}
could be inaccurate. Indeed, we ignore the impacts caused by \eqref{eqt:causal model 2} when \eqref{eqt:IoC estimator for IwC formulation} are used. Thus, the estimated outcome-predictors relation $\hat{g}$ can have opposite outcome-predictors relation of the authentic $g$ w.r.t. the features variable. Furthermore, even if $\hat{g}$ has the similar relation with $g$ w.r.t. the features variable, the \eqref{eqt:IoC estimator for IwC formulation} can be \textit{regularized biased}, meaning that the estimators are sensitive to the estimation $\hat{g}$.
As such, we should construct the estimators which use the information given in \eqref{eqt:causal model 2} and regularized unbiased. PSM, IPW and doubly robust (DRE) approach (e.g., \cite{farrell2015robust}) are methodologies which incorporate the relation \eqref{eqt:causal model 2} through the computation of \textit{propensity score}. However, the DREs can be sensitive w.r.t. the small perturbations on the map $m$ in \eqref{eqt:causal model 2}. Consequently, it may not be suitable for the estimations of ATE and ATTE. To stabilize it, we should build the estimators which can be recovered from the score functions that satisfy the orthogonal condition in the Definition \ref{def:Neyman orthogonal score}. Heuristically, the partial derivative of score functions w.r.t. the nuisance parameters are expected to be $0$. Indeed, the regularization biases of the estimators ATE and ATTE are reduced using a multiplicative term of the propensity score and the residuals between the observed $Y$ and the estimate of $g(d^{i},\cdot,\cdot)$.
\begin{definition}[Orthogonal Condition]\label{def:Neyman orthogonal score}
	Let $W$ be the random elements, $\Theta$ be a convex set which contains the causal parameter $\vartheta$ of dimension $d_{\vartheta}$ ($\theta$ is the true causal parameter we are interested in) and $T$ be a convex set which contains nuisance parameter $\varrho$ ($\rho$ is the true nuisance parameter we are interested). Moreover, we define the Gateaux derivative map $D_{r,j}[\varrho-\rho]:=\partial_{r}\left\{\mathbb{E}[\psi_{j}(W,\theta,\rho+r(\varrho-\rho))]\right\}$. We say that a score function $\psi$ satisfies the  (Neyman) orthogonal condition if for all $r\in[0,1)$, $\varrho\in\mathcal{T}\subset T$ and $j=1,\cdots,d_{\vartheta}$, we have
	\begin{equation}
	\begin{aligned}\label{eqt:Neyman orthogonal score}
	\partial_{\varrho}\mathbb{E}[\psi_{j}(W,\theta,\varrho)]\mid_{\varrho=\rho}[\varrho-\rho]:=D_{0,j}[\varrho-\rho]=0.
	\end{aligned}
	\end{equation}
\end{definition}
To start with our IwC formulation, we let $W=(Y,D,\mathbf{X},\mathbf{U},\mathbf{Z})$. The quantities $\Theta$, $T$, $\vartheta$, $\theta$, $\varrho$ and $\rho$ stated in Definition \ref{def:Neyman orthogonal score} that are needed to check whether a score function of $\theta^{i}$ satisfies the orthogonal condition are defined as follows:
\begin{equation*}
{\small
\begin{aligned}
\Theta=\Theta_{i}&:=\{\vartheta=\mathbb{E}\left[\mathcal{g}(d^{i},\mathbf{U},\mathbf{Z})\right]\mid \textit{$\mathcal{g}$ is $\mathbb{P}$-integrable}\},\\
T=T_{i}&:=\{\varrho=\left(\mathcal{g}(d^{i},\mathbf{U},\mathbf{Z}),a_{i}(\mathbf{X},\mathbf{Z})\right)\mid \textit{$\mathcal{g}$ is $\mathbb{P}$-integrable}\},\\
\theta=\theta^{i}&:=\mathbb{E}\left[g(d^{i},\mathbf{U},\mathbf{Z})\right]\in\Theta_{i},\\
\rho=\rho^{i}&:=\left(g(d^{i},\mathbf{U},\mathbf{Z}),\mathbb{E}\left[\mathbf{1}_{\{D=d^{i}\}}|\mathbf{X},\mathbf{Z}\right]\right)\in T_{i}.
\end{aligned}}
\end{equation*}
Similarly, to check for $\theta^{i\mid j}$, we have
\begin{equation*}
{\small
\begin{aligned}
&\begin{aligned}
\Theta&=\Theta_{i\mid j}\\
:&=\{\vartheta=\mathbb{E}\left[\mathcal{g}(d^{i},\mathbf{U},\mathbf{Z})\mid D=d^{j}\right]\mid \text{$\mathcal{g}$ is $\mathbb{P}$-integrable}\},
\end{aligned}\\
&T=T_{i\mid j}:=\left\{\varrho=(\mathcal{g}(d^{i},\mathbf{U},\mathbf{Z}),m_{j},\right.\\
&\qquad\qquad\qquad\left.a_{j}(\mathbf{X},\mathbf{Z}),a_{i}(\mathbf{X},\mathbf{Z}))\mid \text{$\mathcal{g}$ is $\mathbb{P}$-integrable}\right\},\\
&\theta=\theta^{i\mid j}:=\mathbb{E}\left[g(d^{i},\mathbf{U},\mathbf{Z})\mid D=d^{j}\right]\in\Theta_{i\mid j},\\
&\rho=\rho^{i\mid j}:=\left(g(d^{i},\mathbf{U},\mathbf{Z}),\mathbb{E}\left[\mathbf{1}_{\{D=d^{j}\}}\right],\right.\\
&\qquad\qquad\qquad\left.\mathbb{E}\left[\mathbf{1}_{\{D=d^{j}\}}|\mathbf{X},\mathbf{Z}\right],\mathbb{E}\left[\mathbf{1}_{\{D=d^{i}\}}|\mathbf{X},\mathbf{Z}\right]\right)\in T_{i\mid j}.
\end{aligned}}
\end{equation*}
Here, $\mathcal{g}(d^{i},\mathbf{U},\mathbf{Z})$, $a_{i}(\mathbf{X},\mathbf{Z})$ and $m_{j}$ are the arbitrary nuisance parameters, while $g(d^{i},\mathbf{U},\mathbf{Z})$, $\mathbb{E}\left[\mathbf{1}_{\{D=d^{i}\}}\mid \mathbf{X},\mathbf{Z}\right]$ and $\mathbb{E}\left[\mathbf{1}_{\{D=d^{j}\}}\right]$ are the corresponding true nuisance parameters we are interested in. Our aim is to construct the score functions $\psi$ such that the moments of the Gateaux derivative of $\psi$ w.r.t. $\mathcal{g}(d^{i},\mathbf{U},\mathbf{Z})$, $a_{i}(\mathbf{X},\mathbf{Z})$ and $m_{j}$ evaluating at $g(d^{i},\mathbf{U},\mathbf{Z})$, $\mathbb{E}\left[\mathbf{1}_{\{D=d^{i}\}}\mid \mathbf{X},\mathbf{Z}\right]$ and $\mathbb{E}\left[\mathbf{1}_{\{D=d^{j}\}}\right]$ are $0$, implying the Definition \ref{def:Neyman orthogonal score} holds. 

Before stating the score functions, we introduce some notations to simplify our expression. We define the estimate of $g(d^{i},\mathbf{x},\mathbf{z})$ as $\hat{g}(d^{i},\mathbf{x},\mathbf{z})$. Furthermore, we define $P_{i}(\mathbf{x},\mathbf{z})=\mathbb{E}[\mathbf{1}_{\{D=d^{i}\}}\mid \mathbf{X}=\mathbf{x},\mathbf{Z}=\mathbf{z}]$ and the corresponding estimate as $\hat{P}_{i}(\mathbf{x},\mathbf{z})$ for any $i=1,\cdots,n$. To find $\hat{P}_{i}(\mathbf{x},\mathbf{z})$, we can use any classification methods to obtain it. For example, when we use Logistic regression to estimate $\hat{P}_{i}(\mathbf{x},\mathbf{z})$, it becomes $1/\big[1+\exp\left(-\mathbf{w}_{x}^{T}\mathbf{x}-\mathbf{w}_{z}^{T}\mathbf{z}-\mathbf{w}\right)\big]$.
%
\begin{theorem}\label{thm:desired unbiased estimator}
The score function $\psi^{i}(W,\vartheta,\varrho)$ which can be used to recover an estimate of $\theta^{i}$ and satisfies the Definition \ref{def:Neyman orthogonal score} is
\begin{align}
{\small 
\vartheta-\mathcal{g}(d^{i},\mathbf{U},\mathbf{Z})-\frac{\mathbf{1}_{\{D=d^{i}\}}}{a_{i}(\mathbf{X},\mathbf{Z})}(Y-\mathcal{g}(d^{i},\mathbf{U},\mathbf{Z})), \label{eqt:orthogonal score function expectation}}
\end{align}
while the score function $\psi^{i\mid j}(W,\vartheta,\varrho)$ which can be used to recover an estimate of $\theta^{i\mid j}$ and satisfies the Definition \ref{def:Neyman orthogonal score} is
\begin{align}
&\frac{1}{m_{j}}\left\{\vartheta\mathbf{1}_{\{D=d^{j}\}}-\mathcal{g}(d^{i},\mathbf{U},\mathbf{Z})\mathbf{1}_{\{D=d^{j}\}}\right.\nonumber\\
&\quad\left.-\mathbf{1}_{\{D=d^{i}\}}\frac{a_{j}(\mathbf{X},\mathbf{Z})}{a_{i}(\mathbf{X},\mathbf{Z})}(Y-\mathcal{g}(d^{i},\mathbf{U},\mathbf{Z}))\right\}. \label{eqt:orthogonal score function conditional expectation}
\end{align}
\end{theorem}
We defer the detailed derivations in the supplementary.
Heuristically, we can recover the estimates of $\theta^{i}$ and $\theta^{i\mid j}$ (denoted as $\hat{\theta}_{w}^{i}$ and $\hat{\theta}_{w}^{i\mid j}$) from $\mathbb{E}\left[\psi^{i}(W,\vartheta,\varrho)\mid_{\vartheta=\theta,\varrho=\rho}\right]=0$ and $\mathbb{E}\left[\psi^{i\mid j}(W,\vartheta,\varrho)\mid_{\vartheta=\theta,\varrho=\rho}\right]=0$ respectively, which are:
%
%
\begin{align}
\hat{\theta}_{w}^{i}&=\frac{1}{N}\left\{   
\overset{N}{\underset{m=1}{\sum}}\hat{g}(d^{i},\mathbf{u}_{m},\mathbf{z}_{m}) \right.\nonumber\\
&\qquad\left.+\overset{N_{i}}{\underset{m=1}{\sum}}\frac{(y_{m}^{i}-\hat{g}(d^{i},\mathbf{u}_{m}^{i},\mathbf{z}_{m}^{i}))}{\hat{P}_{i}(\mathbf{x}_{m}^{i},\mathbf{z}_{m}^{i})}\right\},\label{eqt:desired unbiased estimator expectation}\\
\hat{\theta}_{w}^{i\mid j}&=\frac{1}{N_j}\left\{\overset{N_{j}}{\underset{m=1}{\sum}} \hat{g}(d^{i},\mathbf{u}_{m}^{j},\mathbf{z}_{m}^{j})\right.\nonumber\\
&\;\left.+\overset{N_{i}}{\underset{m=1}{\sum}}\frac{\hat{P}_{j}(\mathbf{x}_{m}^{i},\mathbf{z}_{m}^{i})}{\hat{P}_{i}(\mathbf{x}_{m}^{i},\mathbf{z}_{m}^{i})}\left[y_{m}^{i}-\hat{g}(d^{i},\mathbf{u}_{m}^{i},\mathbf{z}_{m}^{i})\right]\right\}.\label{eqt:desired unbiased estimator conditional expectation}
\end{align}

We call $\hat{\theta}_{w}^{i}$ and $\hat{\theta}_{w}^{i\mid j}$ in \eqref{eqt:desired unbiased estimator expectation} and \eqref{eqt:desired unbiased estimator conditional expectation} the IwC estimates. They are regularization unbiased when the residuals between the observed $Y$ and the estimate of $g(d^{i},\cdot,\cdot)$ are used as the regularization term. Besides, they are the consistent estimates (see the Remark \ref{remark:consistency}). Theoritically, we can study the consistency using \textbf{\textit{error decomposition}} (see in the supplementary). We also study the consistency with numerical results in the Section \ref{sec:empirical studies}.

\begin{remark}\label{remark:consistency}
$\hat{\theta}^{i}_{w}$ and $\hat{\theta}^{i\mid j}_{w}$ are the consistent estimates of $\theta^{i}$ and $\theta^{i\mid j}$ if $\hat{P}_{i}$ converges to $P_{i}$ and $\hat{g}$ converges to $g$ in probability (at rate $N^{-\frac{1}{4}}$) when $N$ and $N_{j}$ tend to infinity.
\end{remark}

Whenever the estimates of $\theta^i$ and $\theta^{i|j}$, i.e., $\hat{\theta}^i$ and $\hat{\theta}^{i|j}$ are available, we can estimate the Average Treatment Effect (ATE) and the Average Treatment Effect on the Treated (ATTE) as
\begin{equation}
{\small
\begin{aligned}
\hat{\textrm{ATE}}(i,k)&:=\hat{\theta}^{i,k}=\hat{\theta}^{i}-\hat{\theta}^{k}\\
\hat{\textrm{ATTE}}(i,k|j)&:=\hat{\theta}^{i,k\mid j}=\hat{\theta}^{i\mid j}-\hat{\theta}^{k\mid j}. \label{eqt:desired unbiased estimator ATE and ATTE}
\end{aligned}}
\end{equation}
For IoC formulation, the estimates of ATE and ATTE are denoted as $\hat{\theta}_{o}^{i,k}$ and $\hat{\theta}_{o}^{i,k\mid j}$ which can be computed using \eqref{eqt:IoC estimator for IwC formulation}. For IwC formulation, they are $\hat{\theta}_{w}^{i,k}$ and $\hat{\theta}_{w}^{i,k\mid j}$ computed by $\eqref{eqt:desired unbiased estimator expectation}$ and $\eqref{eqt:desired unbiased estimator conditional expectation}$ respectively.
\section{Experiments}\label{sec:empirical studies}
We now set out experiments to i) estimate the counterfactuals and the treatment effects under different settings and ii) assess the consistency and robustness properties under IwC formulation. 

Our comparisons are made across two aspects: 1) different data properties per \eqref{eqt:simulated data model 1} \& \eqref{eqt:simulated data model 3}, and 2) different choices of the map $\hat{g}$ per \eqref{eqt:causal model 1}. For 1), we generate simulated datasets that possess three main properties: a) different levels of causal effect the intervention $D$
causes on the outcome $Y$, b) different degrees of nonlinearity of this causal impact, and c) different tail heaviness in the distribution of the feature set $(\mathbf{X,U,Z})$. For 2), the various maps under testing include three most commonly used neural network nonlinear models: the Multi-layer Perception Network (MLP), the Convolutional Neural Network (CNN), and the Gated Recurrent Unit (GRU); we also contrast them with widely recognized classic models such as the ordinary least square (OLS) and OLS with LASSO and RIDGE regularization, and the decision-tree based models such as the Random Forest (RF) and one with boosting features, the xgboost (XGB).  

For every set of simulated data and a given choice of $g$, we report estimations of ATE \& ATTE per \eqref{eqt:desired unbiased estimator ATE and ATTE}, using both our IwC estimations
\eqref{eqt:desired unbiased estimator expectation} \eqref{eqt:desired unbiased estimator conditional expectation} and the IoC estimations \eqref{eqt:IoC estimator for IwC formulation}. All results are out-of-sample and we use $70\%$ of data as the training set and the remaining $30\%$ as the testing set. We use grid search to find the optimal hyperparameters of the linear models and the tree-based models. For all neural network-based models, we use the Bayesian optimization to find the optimal hyperparameters. The number of hidden layers ranges from $2$ to $7$ and the number of units for each layer from $50$ to $500$. The batch size is in integer multiples of $32$ and optimized within $[32,3200]$. We search the learning rate between $0.0001$ and $0.1$. The experiments are run on two Ubuntu HP Z4 Workstations each with Intel Core i9 10-Core CPU at 3.3GHz, 128G DIMM-2166 ECC RAM, and two sets of NVIDIA Quadro RTX 5000 GPU. The total computation time of Table \ref{table:light tail and heavy tail with alpha is 0.05 ATE} and Table \ref{table:ATE-alpha_beta real} is 177 hours, including all different sets of $\alpha$ and $\beta$, with each set containing $8$ models; Figure \ref{fig:consistency graph simulated} and Figure \ref{fig:histogram plot real} cost 163 hours in total.\\\\
\textbf{The Data Generating Process.}\label{sec:Models for simulated data} As the ground truth (the factuals and counterfactuals) is unavailable, we construct a data generating process (DGP) for our credit-related dataset similar to many causal learning works:
\begin{subequations}
{\small
\begin{align}
&Y=f(D)q(\mathbf{U},\mathbf{Z})+\xi, \label{eqt:simulated data model 1}\\
&q(\mathbf{U},\mathbf{Z})=\left\{\exp\left(\left|\mathbf{a}_{0}^{T}\mathbf{Z}\right|\right)\right.\nonumber\\
&\qquad\qquad\left.+e_{1}\log\left(e_{2}+k(\mathbf{Z})^{2}+|k(\mathbf{U})|^{\tau}\right)\right\}^{r},\label{eqt:simulated data model 1-q function} \\
&\begin{aligned}
D=\sigma\left(\lambda\left(\mathbf{a}_{1}^{T}\mathbf{X}+\left|\mathbf{a}_{2}^{T}\mathbf{Z}\right|^{\gamma}+\frac{b_{1}\; X_{9}}{1+|X_{3}|}+\frac{b_{2}\; X_{10}}{1+|X_{6}|}\right)+\nu\right),\label{eqt:simulated data model 2}
\end{aligned}\\
&f(D)=\alpha+(1-\alpha)\times\left[\beta D^{m}+(1-\beta)\exp(D^{n})\right],\label{eqt:simulated data model 3}
\end{align}}
\end{subequations}
where $k(\mathbf{Z})$ and $k(\mathbf{U})$ in \eqref{eqt:simulated data model 1-q function} are defined as
\begin{equation}
{\small
\begin{aligned}\label{eqt:simulated data model 1-q function details}
k(\mathbf{Z})&=\log\left(\left|(\mathbf{c}_{1}^{z})^{T}\mathbf{Z}+\sum{_{r=2}^4} (\mathbf{c}_{r}^{z})^{T}\bar{\mathbf{Z}}_{r}\right|\right)\\
&+(\mathbf{c}_{1}^{z})^{T}\log|\mathbf{Z}|+\sum{_{r=2}^4} (\mathbf{c}_{r}^{z})^{T}\log|\bar{\mathbf{Z}}_{r}|,\\
k(\mathbf{U})&=\log\left(\left|(\mathbf{c}_{1}^{u})^{T}\mathbf{U}+\sum{_{r=2}^4} (\mathbf{c}_{r}^{u})^{T}\bar{\mathbf{U}}_{r}\right|\right)\\
&+(\mathbf{c}_{1}^{u})^{T}\log|\mathbf{U}|+\sum{_{r=2}^4}(\mathbf{c}_{r}^{u})^{T}\log|\bar{\mathbf{U}}_{r}|.
\end{aligned}}
\end{equation}
The function $\sigma$ maps the features to the intervention variable $D$ such that $D$ takes five treatment levels $\left\{ d^{1}, d^{2}, d^{3}, d^{4}, d^{5}\right\}$. The confounding features set is a $20$-dimensional random vector $\mathbf{Z}=[Z_{1},\cdots,Z_{20}]^{T}$. Simultaneously, the outcome-specific feature set is a $10$-dimensional random vector $\mathbf{U}=[U_{1},\cdots,U_{10}]^{T}$ and the intervention-specific feature set is a $10$-dimensional random vector $\mathbf{X}=[X_{1},\cdots,X_{10}]^{T}$. All $(\mathbf{Z,U,X})$ are correlated random vectors with the correlation matrix parameterized by $C_{ij}=a+(1-a) \exp(-b |i-j|), a\in [0,1], b\in \mathbb{R}^{+}$. The parameter values of $a,b$ used to generate the correlation matrix, $(\lambda,\gamma, b_{1}, b_{2})$ in \eqref{eqt:simulated data model 2}, $(\mathbf{a}_{0},\tau, r, e_{1}, e_{2})$ in \eqref{eqt:simulated data model 1-q function} , $(\mathbf{a}_{1}, \mathbf{a}_{2})$ in \eqref{eqt:simulated data model 2}, and $(\left\{\mathbf{c}_i^{z}\right\}_{i=1}^4, \left\{\mathbf{c}_i^{u}\right\}_{i=1}^4)$ in \eqref{eqt:simulated data model 1-q function details} are deferred in the supplementary due to space constraints.
The quantity $\bar{\mathbf{Z}}_{r}$ (or $\bar{\mathbf{U}}_{r}$) in \eqref{eqt:simulated data model 1-q function details} is a column vector such that each entry is the product of $r$ elements taken from $\mathbf{Z}$ (or $\mathbf{U}$) without repetition. For example, $\bar{\mathbf{Z}}_{2}=[Z_{1}Z_{2},Z_{1}Z_{3},\cdots,Z_{19}Z_{20}]^{T}$, $\bar{\mathbf{Z}}_{3}=[Z_{1}Z_{2}Z_{3},Z_{1}Z_{2}Z_{4},\cdots,Z_{18}Z_{19}Z_{20}]^{T}$, etc.

Our DGP is reasonable for the following reasons:
1) In many causal works, the researchers never check if their DGP can appropriately fit the real-world data given in their papers (\cite{hainmueller2019much,hill2011bayesian,lim2018forecasting}). Our DGP fits the highly nonlinear and correlated real-world credit data very well. For example, when we fit the real-world dataset (the one used in our semi-synthetic experiment) to \eqref{eqt:simulated data model 2}, the relative mean square error (MSE) we obtained is $\textbf{7.9\%}$ which is very small. 2) Our DGP is generic enough. The functions in our DGP, such as the power/exponential/logarithm/inverse, are all out of lengthy testings of each individual's feature separately. 3) It allows us to control the degrees of causality and nonlinearity. First, the parameter $\alpha \in[0,1]$ in \eqref{eqt:simulated data model 3} controls the amount of the policy impact. The bigger $\alpha$ is, the smaller the impact. In the case $\alpha=1$, $Y$ is no longer a function of the intervention. Second, the parameters $\beta \in[0,1]$ and $n,m\in \mathbb{R}^{+}$ control the degree of nonlinearity for a fixed $\alpha$. For instance, when $m=1, n=2$, the smaller $\beta$ is, the larger the contribution from the term $\exp(D^{n})$, hence the more degree of nonlinearity.\\\\
\textbf{Causality vs. Nonlinearity.} \label{sec:Causality vs. Nonlinearity}
Table \ref{table:light tail and heavy tail with alpha is 0.05 ATE} compares the ATE estimations comprehensively across different choices of $\hat{g}$ and different data properties per \eqref{eqt:simulated data model 1} \& \eqref{eqt:simulated data model 3}.  Each number in column a) ``IoC" and column b) ``IwC" is a weighted average of the ATE estimation in relative errors w.r.t the true ATE, computed from $40,000$ observations and $100$ independent experiments for each individual setting:
\begin{equation}
{\small
\begin{aligned}
\label{eqt:weighted average of relative error general}
\frac{1}{M}\overset{M}{\underset{m=1}{\sum}}\;\underset{\underset{i\neq k}{1\leq i,k\leq n}}{\sum}
\left[\frac{\left|\theta^{i,k;m}\right|}{  \sum{_{\underset{i\neq k}{1\leq i,k\leq n}}} \left|\theta^{i,k;m}\right| }\left|\frac{\hat{\theta}^{i,k;m}}{\theta^{i,k;m}}-1\right|\right],
\end{aligned}}
\end{equation}
where $M$ is the number of experiments conducted, $\theta^{i,k;m}$ is the true value of $\theta^{i,k}$ in the $m^{\textrm{th}}$ experiment, and  $\hat{\theta}^{i,k;m}$ is the estimate of $\theta^{i,k}$ in the $m^{\textrm{th}}$ experiment based on IwC (or IoC). 
The ATTE cases are deferred in the supplementary.

Table \ref{table:light tail and heavy tail with alpha is 0.05 ATE} reports the results when the $\alpha$ and $\beta$ in \eqref{eqt:simulated data model 3} are fixed at different values. When we fix $\alpha=5\%$, the amount of impact of the intervention $D$ on the outcome $Y$ is large. We call it the ``strong causality" case. When the causal effects are strong, whether the data are light-tail distributed or heavy-tail distributed, mostly linear ($\beta=95\%$) or mostly nonlinear ($\beta=5\%$), and whether the choice of $\hat{g}$ is a linear model (OLS, LASSO, RIDGE), a tree-based model (RF, XGB) or a neural-network-based model (GRU, CNN, MLP), the IwC estimates always give superior ATE estimations. Similar observations can be made in the ``strong nonlinearity" case in Table \ref{table:light tail and heavy tail with alpha is 0.05 ATE}. In this case, we fix $\beta=5\%$, meaning that the impact of the intervention $D$ on the outcome $Y$ is very nonlinear, irrespective of whether the amount of the impact is larger or smaller. 

Even if one uses misspecified linear models such as the OLS, LASSO, and RIDGE on heavy-tail distributed highly-nonlinear data, there is at least a $30\%$ reduction in the estimation errors.  When the causal effects and nonlinearity are both strong in the data and one did use the right nonlinear models, e.g., RF, XGB, GRU, CNN, MLP, the error can be reduced by at least $82.7\%$ for the light-tail data and $75.8\%$ for the heavy-tail data.  The significant superiority of IwC vs. IoC ATE estimation across all our settings suggests the importance of considering ``action-response" relationship.\\\\
\textbf{Consistency of the Estimators.} \label{sec:Consistency of the Estimators} We demonstrate the consistency of the estimators through numerical experiments. We repeat the simulated experiment for $100$ times, then compute two quantities. The first quantity requires us to compute the values of $\theta^{i}$ and $\hat{\theta}_{w}^{i}$ based on $100$ experiments and find the corresponding relative errors for each $i$. We then find the mean of all the relative errors. The second quantity is the average of the standard deviation of the difference between $\theta^{i}$ and $\hat{\theta}_{w}^{i}$ of $100$ experiments. Indeed, the first quantity is computed as \eqref{eqt:average of mean relative error and average of standard deviation relative error general-a} and the second quantity is computed as \eqref{eqt:average of mean relative error and average of standard deviation relative error general-b} which are stated as follows:
\begin{subequations}
\begin{equation}\label{eqt:average of mean relative error and average of standard deviation relative error general-a}
\begin{aligned}
\scalebox{0.88}{
$\frac{1}{K}\overset{K}{\underset{i=1}{\sum}}\left|\frac{\overset{M}{\underset{m=1}{\sum}}{\hat{\theta}}_{w}^{i;m}}{\overset{M}{\underset{m=1}{\sum}}\theta^{i;m}}-1\right|$}
\end{aligned}
\end{equation}
\begin{equation}\label{eqt:average of mean relative error and average of standard deviation relative error general-b}
\begin{aligned}
\scalebox{0.85}{
$\frac{1}{K}\overset{K}{\underset{i=1}{\sum}}\sqrt{\frac{1}{M-1}\overset{M}{\underset{m=1}{\sum}}\left([\hat{\theta}_{w}^{i;m}-\theta^{i;m}]-\frac{1}{M}\overset{M}{\underset{s=1}{\sum}}[\hat{\theta}_{w}^{i;s}-\theta^{i;s}]\right)^{2}}$}.
\end{aligned}
\end{equation}
\end{subequations}
%
Here $K$ and $M$ in \eqref{eqt:average of mean relative error and average of standard deviation relative error general-a} and \eqref{eqt:average of mean relative error and average of standard deviation relative error general-b} are the number of estimators and the number of experiments respectively.

To show that the estimators are consistent, we compare the mean of all the relative errors versus the number of observational data and summarize the results in Figure \ref{fig:consistency graph simulated}.
We notice that when the number of observational data increases, the computed values in each plot becomes smaller except using linear regressors in computing the mean relative error between $\theta^{i}$ and $\hat{\theta}_{w}^{i}$. It implies that using linear regressors would cause biased estimations but not the case of using nonlinear regressors. This matches with our expectations since the generated dataset is nonlinear. The consistency analysis of $\hat{\theta}_{w}^{i\mid j}$ is deferred in the supplementary.
\begin{table}
\resizebox{\columnwidth}{!}{%
		\begin{tabular}{ccccccccccccccccccc}
			\\
			& & & & & Light tail & & & &\\
			\cmidrule(r){2-10}
			& \multicolumn{2}{c}{} & \multicolumn{2}{c}{$\underset{\Leftarrow}{\text{Weak causalities}}$} & \multicolumn{1}{c}{} & \multicolumn{2}{c}{$\underset{\Rightarrow}{\text{Weak nonlinearities}}$} & \multicolumn{2}{c}{} \\
			& \multicolumn{3}{c}{$(\alpha,\beta)=(95\%,5\%)$}& \multicolumn{3}{c}{$(\alpha,\beta)=(5\%,5\%)$} & \multicolumn{3}{c}{$(\alpha,\beta)=(5\%,95\%)$}\\
			\cmidrule(r){2-10}
			 \multirow{3}{*}{} & \multicolumn{2}{c}{Weighted avg.} & \multicolumn{1}{c}{Mean Err.} & \multicolumn{2}{c}{Weighted avg.} & \multicolumn{1}{c}{Mean Err.} & \multicolumn{2}{c}{Weighted avg.} & \multicolumn{1}{c}{Mean Err.}\\
			& \multicolumn{2}{c}{ATE} & \multicolumn{1}{c}{reduction} & \multicolumn{2}{c}{ATE} & \multicolumn{1}{c}{reduction} & \multicolumn{2}{c}{ATE} & \multicolumn{1}{c}{reduction}\\
			\midrule
			\multirow{1}{*}{Regressor} & a) \textbf{IoC} & b) \textbf{IwC} & $|$IwC/IoC-1$|$ & a) \textbf{IoC} & b) \textbf{IwC} & $|$IwC/IoC-1$|$ & a) \textbf{IoC} & b) \textbf{IwC} & $|$IwC/IoC-1$|$\\
			\midrule
			OLS&\percentage{0.45378843579525501850113755608618}&	\percentage{0.24501211891999399861141739620507}&	\percentage{0.46007412354919202668313005233358}&	\percentage{0.44870265033558798251078769681044}&	\percentage{0.2146161518872619911046939478183}&	\percentage{0.52169626872774199632942782045575}&	\percentage{0.4247893516927649937287014836329}&	\percentage{0.24970683511008601174196996907995}&	\percentage{0.41216314835808198768773991105263}\\
			LASSO&\percentage{0.45376466703373902200624456781952}&	\percentage{0.17231576689006100622236772323959}&	\percentage{0.62025300908400404775733250062331}&	\percentage{0.44870171103541800494696190071409}&	\percentage{0.20998250751831801230729013241216}&	\percentage{0.53202204860381296924742855480872}&	\percentage{0.42476115372055700536435551839531}&	\percentage{0.18552649473520099743240052703186}&	\percentage{0.563221605577293948918793375924}\\
			RIDGE&\percentage{0.45354394762373201333005567903456}&	\percentage{0.24136839491792599909913974443043}&	\percentage{0.46781696419379897866619444357639}&	\percentage{0.44859614857811602739445788756711}&	\percentage{0.21329934592058699682759481675021}&	\percentage{0.5245181069951949659468937170459}&	\percentage{0.42454277718186300694114265752432}&	\percentage{0.24609041864921399334065199582255}&	\percentage{0.42034011205472499028701349743642}\\
			RF&\percentage{0.66889589667602999956130815917277}&	\percentage{0.10486574718099299363682774810513}&	\percentage{0.84322560849587102627111789843184}&	\percentage{0.67474144687772397155356429720996}&	\percentage{0.089945035483665294884225716032233}&	\percentage{0.86669703499040395833219463384012}&	\percentage{0.67396183573126899890581853469484}&	\percentage{0.10710277363893000446282144366705}&	\percentage{0.84108480931606399710886989851133}\\
			XGB&\percentage{0.8811861302905830140019816099084}&	\percentage{0.12049915620030099416837288117677}&	\percentage{0.86325345797196595398759200179484}&	\percentage{0.87545170054012699800694008445134}&	\percentage{0.10103336611988500415293401601957}&	\percentage{0.88459287239084605225514224002836}&	\percentage{0.88126829301533704530413615430007}&	\percentage{0.1266239970618620036546531082422}&	\percentage{0.85631617741674603827561895741383}\\
			GRU&\percentage{0.507650443555306973841823037219}&	\percentage{0.1702215066094180107025835013701}&	\percentage{0.66468756450347898923070033561089}&	\percentage{0.20534463196282198849118572070438}&	\percentage{0.026447635281974298665508271710678}&	\percentage{0.87120366853922504013496563857188}&	\percentage{0.49520486049449202736738584462728}&	\percentage{0.15878251009696300877216401659098}&	\percentage{0.67935995228639400878023479890544}\\
			CNN&\percentage{0.45040158570583899155792551027844}&	\percentage{0.15074126767600701248817074429098}&	\percentage{0.66531807955388100328519840331865}&	\percentage{0.17681796113132200454209908002667}&	\percentage{0.030573892043141700364250468169303}&	\percentage{0.82708831247955294507789858471369}&	\percentage{0.43917636965661299708330034263781}&	\percentage{0.14589040482609100246591538052598}&	\percentage{0.66780907419913804190514383662958}\\
			MLP&\percentage{0.47083177479931398456969304788799}&	\percentage{0.15043319100870300442096549886628}&	\percentage{0.68049481989013205218697066811728}&	\percentage{0.16101123734757000205419785743288}&	\percentage{0.023521137092784898808828231153711}&	\percentage{0.85391617703048494636419718517573}&	\percentage{0.45281468218747900067455702810548}&	\percentage{0.14340945527709800244586801909463}&	\percentage{0.68329327444881104991480924581992}\\
			\bottomrule
	\end{tabular}
	\label{table:light tail and heavy tail with alpha is 0.05 ATE-sub1}
    }
\resizebox{\columnwidth}{!}{%
		\begin{tabular}{ccccccccccccccccccc}
			\\
			& & & & & Heavy tail & & & &\\
			\cmidrule(r){2-10}
			 & \multicolumn{2}{c}{} & \multicolumn{2}{c}{$\underset{\Leftarrow}{\text{Weak causalities}}$} & \multicolumn{1}{c}{} & \multicolumn{2}{c}{$\underset{\Rightarrow}{\text{Weak nonlinearities}}$} & \multicolumn{2}{c}{}\\
			& \multicolumn{3}{c}{$(\alpha,\beta)=(95\%,5\%)$}& \multicolumn{3}{c}{$(\alpha,\beta)=(5\%,5\%)$}& \multicolumn{3}{c}{$(\alpha,\beta)=(5\%,95\%)$}\\
			\cmidrule(r){2-10}
			 \multirow{3}{*}{} & \multicolumn{2}{c}{Weighted avg.} & \multicolumn{1}{c}{Mean Err.} & \multicolumn{2}{c}{Weighted avg.} & \multicolumn{1}{c}{Mean Err.} & \multicolumn{2}{c}{Weighted avg.} & \multicolumn{1}{c}{Mean Err.}\\
			& \multicolumn{2}{c}{ATE} & \multicolumn{1}{c}{reduction} & \multicolumn{2}{c}{ATE} & \multicolumn{1}{c}{reduction} & \multicolumn{2}{c}{ATE} & \multicolumn{1}{c}{reduction}\\
			\midrule
			\multirow{1}{*}{Regressor} & a) \textbf{IoC} & b) \textbf{IwC} & $|$IwC/IoC-1$|$ & a) \textbf{IoC} & b) \textbf{IwC} & $|$IwC/IoC-1$|$ & a) \textbf{IoC} & b) \textbf{IwC} & $|$IwC/IoC-1$|$ \\
			\midrule
			OLS&	\percentage{0.87579824898634295315957842831267}&	\percentage{0.59945875110420698828050944939605}&	\percentage{0.31552871703268797753949570505938}&	\percentage{0.820860379507617032857069716556}&	\percentage{0.52635638819534702470548381825211}&	\percentage{0.3587747668963189862090246151638}&	\percentage{0.88654680032690402757822312196367}&	\percentage{0.61864461525902503247209551773267}&	\percentage{0.30218617332902497629731897177408}\\
			LASSO&	\percentage{0.87575611391888996681842627367587}&	\percentage{0.49118978236447502494499417480256}&	\percentage{0.43912491781933699064310872017813}&	\percentage{0.79846025193485103965684857030283}&	\percentage{0.47246010337913602050718964164844}&	\percentage{0.40828600768258899167406639207911}&	\percentage{0.88650891918578500483505422380404}&	\percentage{0.5199547136850599660107263844111}&	\percentage{0.41348056129811600056811471404217}\\
			RIDGE&	\percentage{0.87244882039910298665574828191893}&	\percentage{0.59511810918714702633991464608698}&	\percentage{0.31787619482950202343118917269749}&	\percentage{0.81798904124894900036224498762749}&	\percentage{0.52311824972008302836457005469128}&	\percentage{0.36048256964254798973001925332937}&	\percentage{0.88313155925387698541584313716157}&	\percentage{0.61424643603569994976254520224757}&	\percentage{0.30446780029619502361271088375361}\\
			RF&
\percentage{0.67166438579525000740488849260146}&	\percentage{0.13841384193198899565935278133111}&	\percentage{0.79392410129337498414514584510471}&	\percentage{0.67985040176652200560880601187819}&	\percentage{0.10251997209528600218941818411622}&	\percentage{0.8492021600209409548654093669029}&	\percentage{0.68069781311770305443076267692959}&	\percentage{0.15287456712810298942173403702327}&	\percentage{0.77541492835431002905721697970876}\\
			XGB&	\percentage{0.86865675935377495431310990170459}&	\percentage{0.20430692547923298696410654429201}&	\percentage{0.76480131734515599539747654489474}&	\percentage{0.87263698576165205089694154594326}&	\percentage{0.14484507913252700483930368591245}&	\percentage{0.83401450832833501358720695861848}&	\percentage{0.86900521678775299871944071128382}&	\percentage{0.20383578616928199611635363908135}&	\percentage{0.765437787677784031714622869913}\\
			GRU&	\percentage{0.55083601882226196355674119331525}&	\percentage{0.20724814220336498782337741886295}&	\percentage{0.62375709808069301764987812930485}&	\percentage{0.21697864365817298715199967773515}&	\percentage{0.052481876420457998499990992513631}&	\percentage{0.75812423040519294481498491222737}&	\percentage{0.51492474750230299296305247480632}&	\percentage{0.22993091833637999132911033939308}&	\percentage{0.55346694939078899810880329823704}\\
			CNN&	\percentage{0.48263201735741501074983261787565}&	\percentage{0.17764145473603198954748449978069}&	\percentage{0.63193188941611599940273436004645}&	\percentage{0.20845909899752299021535861811572}&	\percentage{0.037983427455799301242223009467125}&	\percentage{0.81778954414337801104295522236498}&	\percentage{0.48746738927339799962723532189557}&	\percentage{0.17868702390907600929814691426145}&	\percentage{0.63343799433348502159191184546216}\\
			MLP&	\percentage{0.48909724121108200112217900823453}&	\percentage{0.19805846786084499444591244809999}&	\percentage{0.59505298502518300107766435758094}&	\percentage{0.18542023797748699576182218606846}&	\percentage{0.030109376254280000029250885518195}&	\percentage{0.83761548047448997067476739175618}&	\percentage{0.4436594283723930276508440329053}&	\percentage{0.19812418539310699960864781132841}&	\percentage{0.55343181566107002566923256381415} \\
			\bottomrule
	\end{tabular}
	\label{table:light tail and heavy tail with alpha is 0.05 ATE-sub2}
     }
     \caption{Comparison of the ATE estimations per \eqref{eqt:weighted average of relative error general} with different \textbf{causalities} and \textbf{nonlinearities} in simulated data with different tail heaviness (light tail \& heavy tail). 
    	$\text{Number of observational data: }N=40000$. $\text{Number of simulated experiments: }M=100$.
     }
     \label{table:light tail and heavy tail with alpha is 0.05 ATE}
\end{table}
\begin{figure}
	\begin{subfigure}{.48\textwidth}
		\centering
		\includegraphics[width=50mm,scale=0.5]{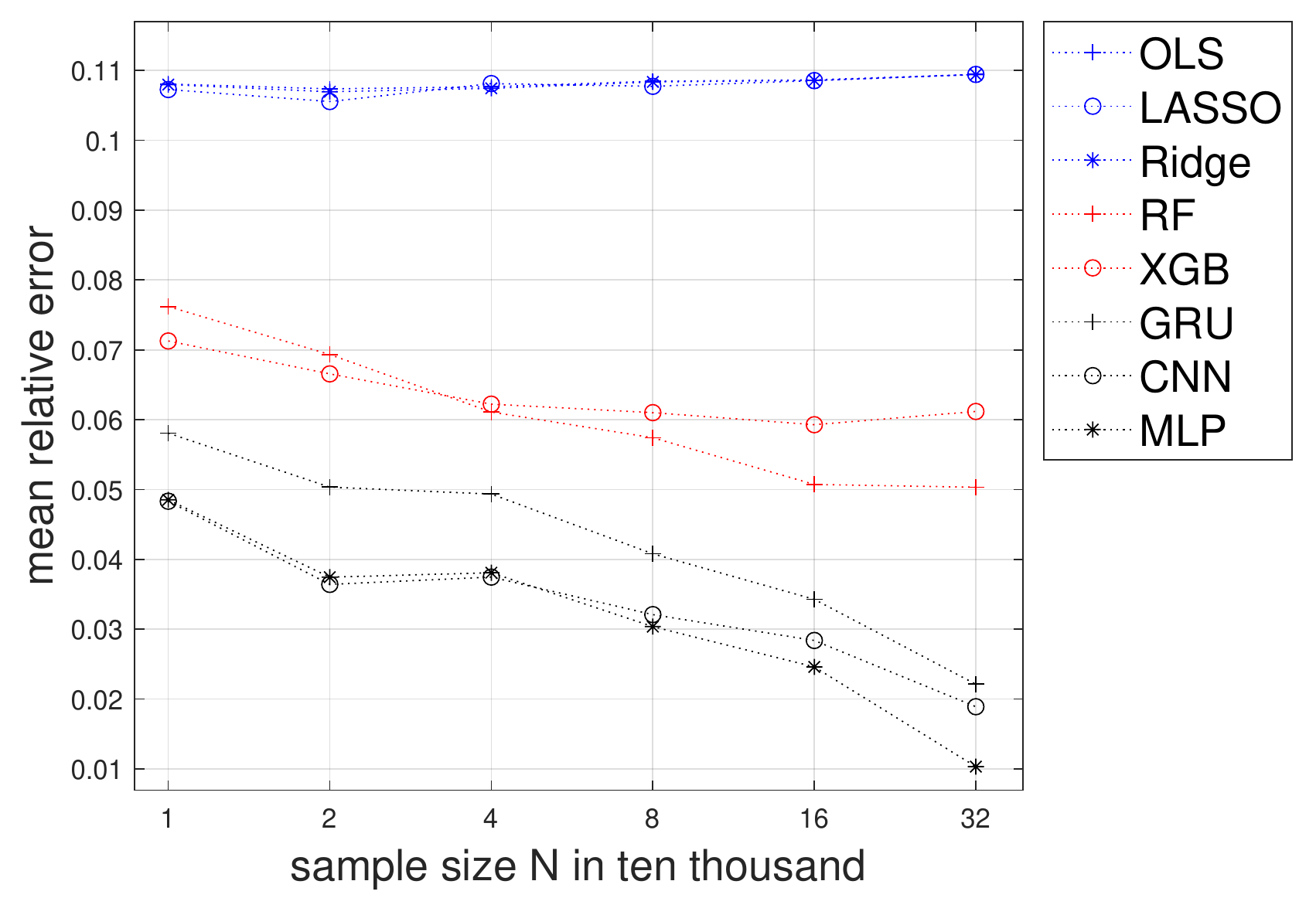}
		\caption{The mean relative error computed by \eqref{eqt:average of mean relative error and average of standard deviation relative error general-a} vs. number of observations.}\label{fig:consistency expectation graph simulated}
	\end{subfigure}
	\quad
	\begin{subfigure}{.48\textwidth}
		\centering
		\includegraphics[width=50mm,scale=0.5]{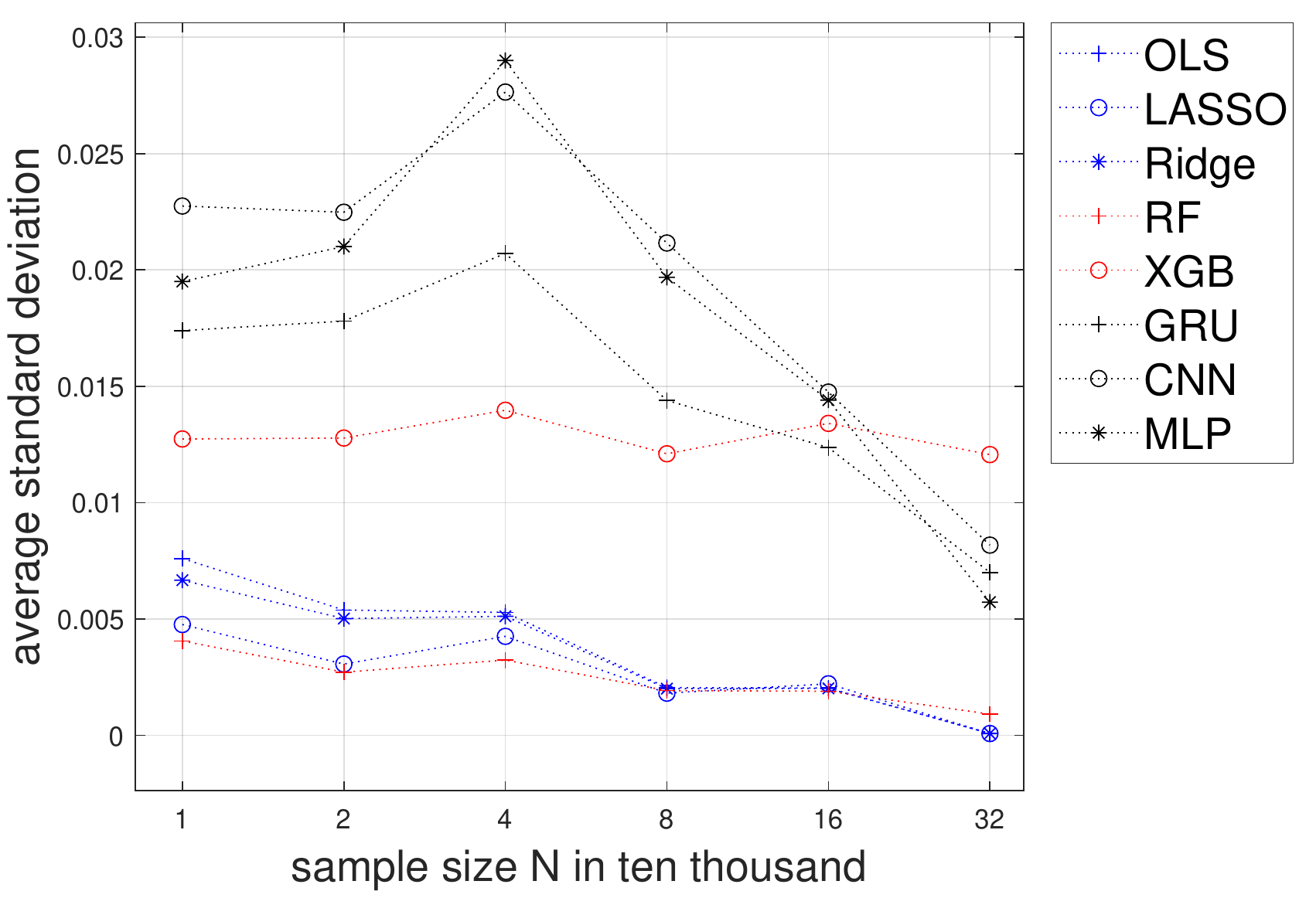}
		\caption{The average standard deviation computed by \eqref{eqt:average of mean relative error and average of standard deviation relative error general-b} vs. number of observations.}\label{fig:consistency expectation_std graph simulated}
	\end{subfigure}
	\caption{Consistency analysis of $\hat{\theta}_{w}^{i}$ in \eqref{eqt:desired unbiased estimator expectation} with $100$ experiments vs. number of observations $N$. $\alpha=0.05$ and $\beta=0.05$ in \eqref{eqt:simulated data model 3}; $N\in\{10000, 20000, 40000, 80000, 160000, 320000\}$.}
	\label{fig:consistency graph simulated}
\end{figure}\\\\
\textbf{Testing on Real-world Data.}\label{sec:Empirical experiment}
We now apply our method to a unique real-world anonymous dataset kindly provided by one of the largest global technology firms that operates in both the e-commerce business and the lending business. The dataset contains observational records of $400,000$ concurrent customers as of writing this paper. The feature dimension of the raw data is $1159$, including 1) the borrowers' shopping, purchasing and payment histories; 2) credit lines, interest charges and financing application decisions set by the lender's decision algorithm;  3) outstanding amounts and delinquency records. After autoencoding, the dimension of the feature set is reduced to $40$, which will be our $(\mathbf{Z,U,X})$. Descriptive statistics of a few representative and important features are stated in Table \ref{table:statistics of variables}.

In the field of credit risk analysis, a default outcome is defined w.r.t. a specific credit event.
The event in this study is the “three-month delinquent” event. It is the borrowers' payment for any of their outstanding loans within the next three months, no matter the borrowers are late on the payment or not.
The intervention variable is the credit line set by the lender's lending algorithm, at the time the customer apply for shopping credit to finance their online purchases.

Different from simulation studies, the true $\theta^{i}$ and $\theta^{i\mid j}$ are unavailable since the counterfactuals can never be observed in the real application. As such, we adopt the semi-synthetic approach to generate the ground truth of counterfactuals, as commonly used in the literature (e.g., \cite{hill2011bayesian,johansson2016learning,louizos2017causal} and references therein). The details of the semi-synthetic setup can be found in the supplementary.

The results reported in Figure \ref{fig:histogram plot real} and Table \ref{table:ATE-alpha_beta real} are strikingly encouraging, given these are from the real-world data. Not only do we see a single-digit percentage estimation error in all settings whenever the IwC estimates are used, we see both significant and robust error reductions compared to the IoC estimates.

\begin{table}	
	\resizebox{\columnwidth}{!}{%
		\centering
		\begin{tabular}{cccccccccccccccc}
			\toprule 
			$\mathbf{\textit{N}}$ & $\textbf{mean}$ & $\textbf{std}$ & $\textbf{min}$ & $\textbf{max}$ & $\mathbf{5\%}$ & $\mathbf{25\%}$ & $\mathbf{50\%}$ & $\mathbf{75\%}$ & $\mathbf{95\%}$ \\
			\midrule 
			{Credit line} &\num[round-precision=1]{4750.7639424999997572740539908409}&	\num[round-precision=1]{2383.1878141451397823402658104897}& \num[round-precision=0]{100.0} & \num[round-precision=0]{50000.0}&	\num[round-precision=0]{600.0} & \num[round-precision=0]{3000.0}&	\num[round-precision=0]{4800.0} & \num[round-precision=0]{6000.0}&	\num[round-precision=0]{9000.0}\\
			{Max payment} & \num[round-precision=1]{506.54839185009899438227876089513}&	\num[round-precision=1]{2694.5432427953501246520318090916}& \num[round-precision=0]{0.0} & \num[round-precision=0]{997500.0}&	\num[round-precision=0]{0.0} & \num[round-precision=0]{0.0}&	\num[round-precision=1]{117.6} & \num[round-precision=1]{319.48}&	\num[round-precision=0]{2499.0}\\
			{Total payments} & \num[round-precision=1]{1134.1232325749899700895184651017}&	\num[round-precision=1]{14760.94171247819940617773681879}&	\num[round-precision=0]{0.0} & \num[round-precision=0]{6759414.0}&	\num[round-precision=0]{0.0} & \num[round-precision=0]{0.0}&	\num[round-precision=1]{198.93} & \num[round-precision=1]{796.84}&	\num[round-precision=1]{4397.8024999999997817212715744972}\\
			{Total order price} & \num[round-precision=1]{3227.4340653750300589308608323336}&	\num[round-precision=1]{40544.43284645889798412099480629}&	\num[round-precision=0]{0.0} & \num[round-precision=0]{24459650.0} & \num[round-precision=1]{55.5} & \num[round-precision=1]{347.38} &	\num[round-precision=1]{1115.7699999999999818101059645414} & \num[round-precision=1]{3268.56} & \num[round-precision=1]{11524.706999999900290276855230331}\\
			{$\#$ of days ordering} & \num[round-precision=1]{2.9694449999999998901500930514885} & \num[round-precision=1]{4.4781264526542496895444855908863} & \num[round-precision=0]{0.0} & \num[round-precision=0]{87.0} &	\num[round-precision=0]{0.0} & \num[round-precision=0]{0.0} &	\num[round-precision=0]{2.0} & \num[round-precision=0]{4.0} &	\num[round-precision=0]{11.0} \\ 
			\bottomrule
	\end{tabular}}
\caption{The mean, standard deviation, minimum, maximum and $5\%$, $25\%$, $50\%$, $75\%$, $95\%$ quantiles of a few selected variables (records within 3 months) in the data.}
\label{table:statistics of variables}
\end{table}
\begin{figure}
	\begin{subfigure}{.49\textwidth}
		\centering
		\includegraphics[width=70mm,scale=0.7]{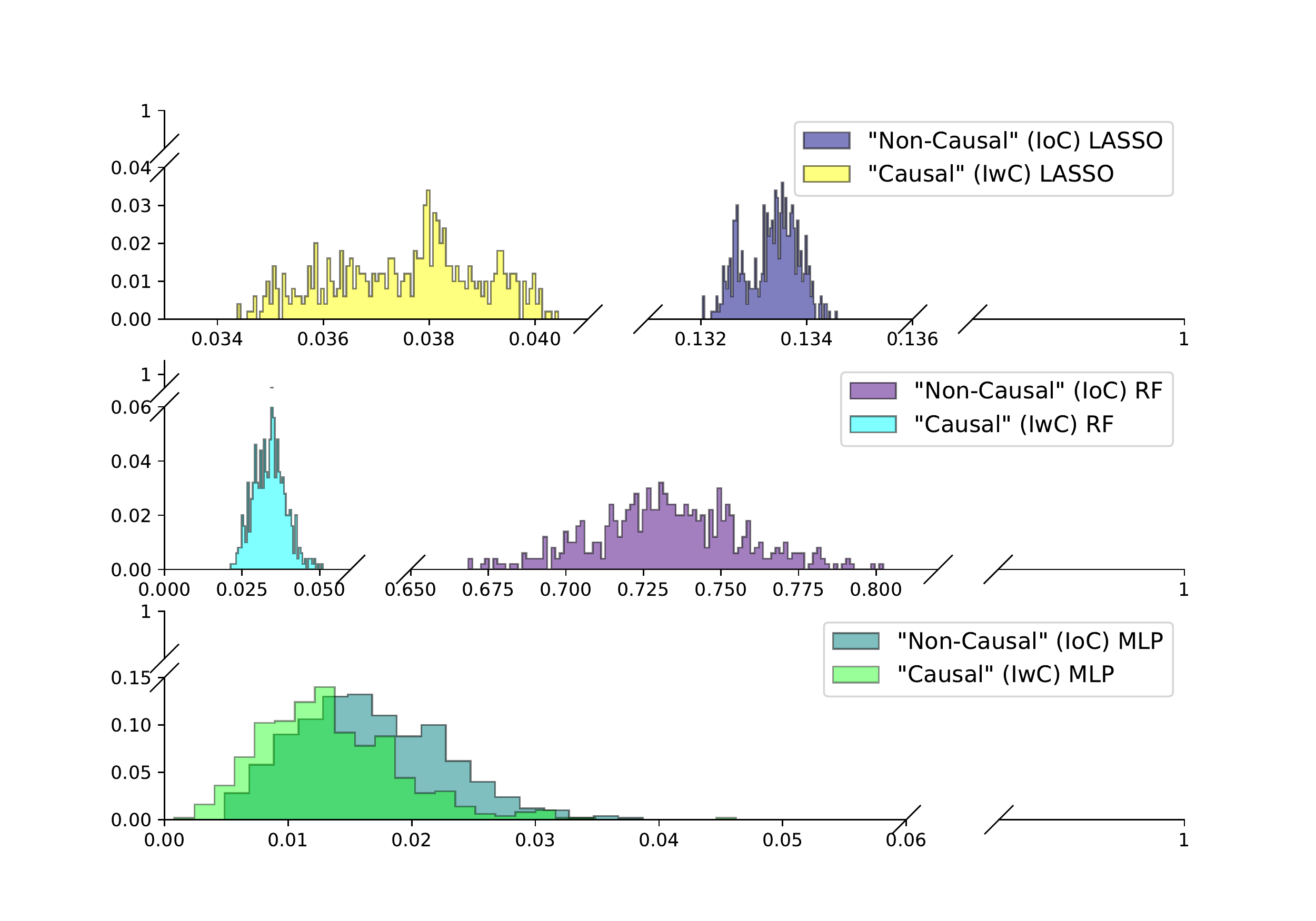}
		\caption{Weighted average of relative err. of estimated ATE using IoC and IwC for different models; $N=160000$, $M=500$, and $\alpha=0.05$ and $\beta=0.05$ in \eqref{eqt:simulated data model 3}.}\label{fig:ATE_3_model_histogram}
	\end{subfigure}
	\quad
	\begin{subfigure}{.49\textwidth}
		\centering
		\includegraphics[width=70mm,scale=0.7]{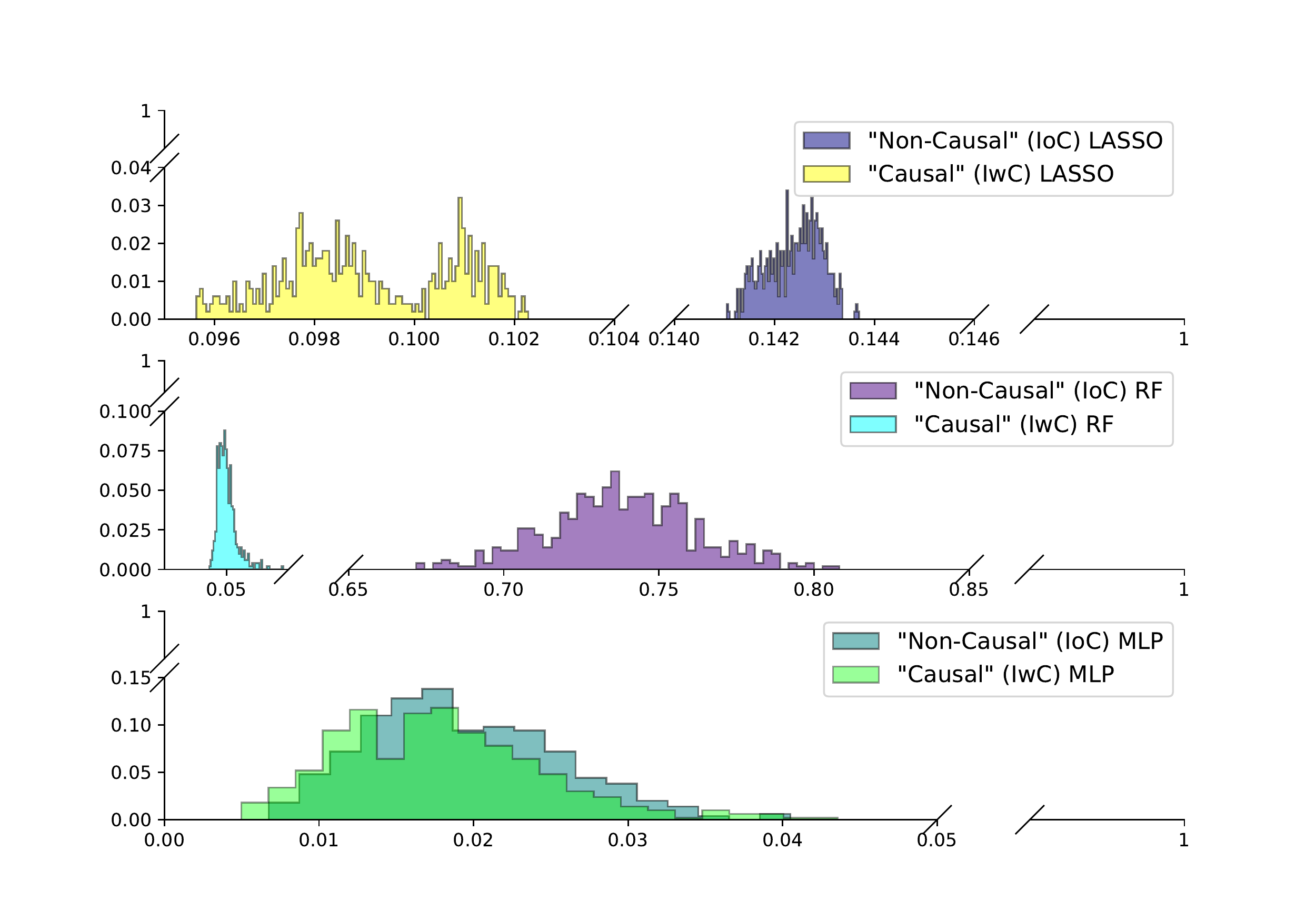}
		\caption{Weighted average of relative err. of estimated ATTE using IoC and IwC for different models; $N=160000$, $M=500$, and $\alpha=0.05$ and $\beta=0.05$ in \eqref{eqt:simulated data model 3}.}\label{fig:ATTE_3_model_histogram}
	\end{subfigure}
	\caption{Frequency histogram of weighted average of relative error of estimated ATE and ATTE for three models: LASSO, random forest (RF) and multilayer perception (MLP) for $500$ experiments.}
	\label{fig:histogram plot real}
\end{figure}
\begin{table}
\resizebox{1.05\linewidth}{!}{%
    \begin{subtable}{\textwidth}
        \centering
\begin{tabular}{cccccccccccccccccc}
			& \multicolumn{6}{c}{$\alpha=5\%$}\\
			& \multicolumn{3}{c}{$\beta=5\%$}& \multicolumn{3}{c}{$\beta=50\%$}\\
			\cmidrule(r){2-7}
			\multirow{3}{*}{} & \multicolumn{2}{c}{Weighted avg.} & \multicolumn{1}{c}{Mean Err.} & \multicolumn{2}{c}{Weighted avg.} & \multicolumn{1}{c}{Mean Err.}\\
			& \multicolumn{2}{c}{ATE} & \multicolumn{1}{c}{reduction} & \multicolumn{2}{c}{ATE} & \multicolumn{1}{c}{reduction}\\
			\midrule
			\multirow{1}{*}{Regressor} & a) \textbf{IoC} & b) \textbf{IwC} & $|$IwC/IoC-1$|$ & a) \textbf{IoC} & b) \textbf{IwC} & $|$IwC/IoC-1$|$ \\
			\midrule
			OLS&\percentage{0.13620436745169600678551091732515}&	\percentage{0.0481091021517250985728431089683}&	\percentage{0.64678737509069594846522477382678}&	\percentage{0.13839795500234999825472925749636}&	\percentage{0.062670263382629098902754094524425}&	\percentage{0.54717348690907396946414564808947}\\
			LASSO&\percentage{0.13620433696970599135589452544082}&	\percentage{0.047295645366292597688850918302705}&	\percentage{0.65275962264834597181817343880539}&	\percentage{0.13839788464693600356092417769105}&	\percentage{0.062148694382107697031791815334145}&	\percentage{0.55094187645531100816498337735538}\\
			RIDGE&\percentage{0.13625649277802701075934521668387}&	\percentage{0.04777657028565240188155982536955}&	\percentage{0.64936298218474997057114705967251}&	\percentage{0.13845748866090099027381654650526}&	\percentage{0.062273637322660196913304275767587}&	\percentage{0.55023279762660204372792804861092}\\
			RF&\percentage{0.73511112016960600268333791973419}&	\percentage{0.029506367768078300972689476111555}&	\percentage{0.95986135026596897112938222562661}&	\percentage{0.73960342712801296904956416256027}&	\percentage{0.036671399265492100882912751558251}&	\percentage{0.95041748331549402717399743778515}\\
			XGB&\percentage{0.88246385208120003973419898102293}&	\percentage{0.058625680937386202840055915430639}&	\percentage{0.93356591230436902772993335020146}&	\percentage{0.8831842470075389472938809376501}&	\percentage{0.071207378397275894466034174001834}&	\percentage{0.91937426574517600030844732827973}\\
			GRU&\percentage{0.040660479526187502663514550249602}&	\percentage{0.024079644116821400517958196019208}&	\percentage{0.40778750281799097665569320270151}&	\percentage{0.073209899482368695489142851329234}&	\percentage{0.032448435101254398393866296146371}&	\percentage{0.5567753086579629862029605646967}\\
			CNN&\percentage{0.032310557834357801765268192184521}&	\percentage{0.021995559053785900677224773858143}&	\percentage{0.31924545634440398833575613934954}&	\percentage{0.045520963964540299417915747426377}&	\percentage{0.033818937719009602160813443560983}&	\percentage{0.25706894640118399530948067877034}\\
			MLP&\percentage{0.031881653343824201130018991534598}&	\percentage{0.020409235825114401680080433720832}&	\percentage{0.35984386992063299715027824277058}&	\percentage{0.036809975519078601269296058262626}&	\percentage{0.03371998381337740213758991103532}&	\percentage{0.083944410777986797067562463325885}\\
			\bottomrule
\end{tabular}
	\label{table:ATE-alpha_beta real-sub1}
    \end{subtable}
    }
\resizebox{1.05\linewidth}{!}{%
    \begin{subtable}[h]{\textwidth}
    \centering
\begin{tabular}{cccccccccccccccccc}
			& \multicolumn{6}{c}{$\beta=5\%$}\\
			& \multicolumn{3}{c}{$\alpha=5\%$}& \multicolumn{3}{c}{$\alpha=50\%$}\\
			\cmidrule(r){2-7}
			\multirow{3}{*}{} & \multicolumn{2}{c}{Weighted avg.} & \multicolumn{1}{c}{Mean Err.} & \multicolumn{2}{c}{Weighted avg.} & \multicolumn{1}{c}{Mean Err.}\\
			& \multicolumn{2}{c}{ATE} & \multicolumn{1}{c}{reduction} & \multicolumn{2}{c}{ATE} & \multicolumn{1}{c}{reduction}\\
			\midrule
			\multirow{1}{*}{Regressor} & a) \textbf{IoC} & b) \textbf{IwC} & $|$IwC/IoC-1$|$ & a) \textbf{IoC} & b) \textbf{IwC} & $|$IwC/IoC-1$|$ \\
			\midrule
			OLS&	\percentage{0.13620436745169600678551091732515}&	\percentage{0.0481091021517250985728431089683}&	\percentage{0.64678737509069594846522477382678}&	\percentage{0.13974563420029201266103768830362}&	\percentage{0.063989097907898598305287407583819}&	\percentage{0.5421030626532099550018983791233}\\
			LASSO&	\percentage{0.13620433696970599135589452544082}&	\percentage{0.047295645366292597688850918302705}&	\percentage{0.65275962264834597181817343880539}&	\percentage{0.13974557163758599687142236689397}&	\percentage{0.061385383136943298676424518589556}&	\percentage{0.5607346807658489806769352981064}\\
			RIDGE&	\percentage{0.13625649277802701075934521668387}&	\percentage{0.04777657028565240188155982536955}&	\percentage{0.64936298218474997057114705967251}&	\percentage{0.13979891524026200055885738038342}&	\percentage{0.063564481952010198906322102629929}&	\percentage{0.54531491290353495404730210793787}\\
			RF&
\percentage{0.73511112016960600268333791973419}&	\percentage{0.029506367768078300972689476111555}&	\percentage{0.95986135026596897112938222562661}&	\percentage{0.73512538865047094560623008874245}&	\percentage{0.039700388159261697229585053037226}&	\percentage{0.94599507951678396544537008594489}\\
			XGB&	\percentage{0.88246385208120003973419898102293}&	\percentage{0.058625680937386202840055915430639}&	\percentage{0.93356591230436902772993335020146}&	\percentage{0.88429452898837301866308280295925}&	\percentage{0.075342931143127706006268340388488}&	\percentage{0.91479882700470904755007950370782}\\
			GRU&	\percentage{0.040660479526187502663514550249602}&	\percentage{0.024079644116821400517958196019208}&	\percentage{0.40778750281799097665569320270151}&	\percentage{0.063793216871499094922093320292333}&	\percentage{0.029842160300511198778306010126471}&	\percentage{0.53220480540081005482733189637656}\\
			CNN&	\percentage{0.032310557834357801765268192184521}&	\percentage{0.021995559053785900677224773858143}&	\percentage{0.31924545634440398833575613934954}&	\percentage{0.047935133163496401742520447442075}&	\percentage{0.031026474344767098617614564659561}&	\percentage{0.3527404161172860153072861066903}\\
			MLP&	\percentage{0.031881653343824201130018991534598}&	\percentage{0.020409235825114401680080433720832}&	\percentage{0.35984386992063299715027824277058}&	\percentage{0.037150202455738398632512087260693}&	\percentage{0.027669049428001701168566839328378}&	\percentage{0.25521134209249102253735941303603}\\
			\bottomrule
\end{tabular}
	\label{table:ATE-alpha_beta real-sub2}
     \end{subtable}
     }
	\caption{Comparison of the ATE estimations per \eqref{eqt:weighted average of relative error general} under different \textbf{nonlinearities} and \textbf{causalities} in semi-synthetic data. The upper table is presented when $\alpha$ is fixed at $5\%$ while the lower table is presented when $\beta$ is fixed at $5\%$. $\text{Number of observational data: }N=80000$. $\text{Number of experiments: }M=100$.}
	\label{table:ATE-alpha_beta real}
\end{table}

\section{Conclusion}\label{sec:conclusion}
This paper presents the first retail credit risk study that estimates the action-response causality from observational data of both a) the e-commerce lender's credit decision records and b) their borrowers' purchase, borrowing, and payment records. Our study shows that the confounding effects between the lender's credit decisions and the borrowers' credit risks, if overlooked, would result in significant biases in risk assessment. Using our IwC formulations, the biases can be reduced to a few percentages. The larger the dataset is, the higher the error reduction. The nature of the current study is about state estimation. Our future study will be towards the decision making problem built on the current framework.

\section{Acknowledgments}\label{sec:Acknowledgments}
Qi WU acknowledges the support from the JD Digits - CityU Joint Laboratory in Financial Technology and Engineering, the Laboratory for AI Powered Financial Technologies, and the GRF support from the Hong Kong Research Grants Council under GRF 14206117 and 11219420. 
\bibliography{aaai21}
\onecolumn
\section{Appendix}\label{Appendix}
The section ``\textbf{Notations}" summarizes the notations which appear in the main paper and Appendix. In the section ``\textbf{Examples}", we give a tangible example which demonstrates the difference in computing IoC estimation and IwC estimation. In the section ``\textbf{Details of models for simulated data}", we give the detailed descriptions about the models which are used for simulated data generating process and semi-synthetic data generating process. In the section ``\textbf{Tables and Figures}", we summarize all the tables and figures which are not included in the main paper. In the section ``\textbf{Proofs of Theorem}", we provide detailed proofs of Theorems in the main paper. In the section ``\textbf{Score functions of IoC and Doubly Robust (DR) estimators}", we state out the score functions which can be used to recover the IoC and doubly robust (DR) estimators, and check if the corresponding score functions satisfy the moment condition and the orthogonal condition. In the section ``\textbf{Consistency}", we undergo error decomposition between the `true' quantity and the corresponding IwC estimator, which allows us to study the consistency of the IwC estimators.
\subsection{Notations}\label{sec:Appendix-Notations}
In this section, we provide detailed descriptions about the notations which appear in the main paper and the Appendix.

\renewcommand*{\arraystretch}{1.37}
\begin{longtable}{@{}l @{\hspace{4mm}} l }
	$Y,y$:&outcome (response) variable, observational outcome (response)\\
	$D,d$:&intervention variable, an observation intervention of $D$\\
	$\mathbf{Z},\mathbf{z}$:&a multidimensional confounder affecting $Y$ and $D$, an observation of $\mathbf{Z}$\\
	$\mathbf{X},\mathbf{x}$:&a multidimensional random variable affecting $D$ only, an observation of $\mathbf{X}$\\
	$\mathbf{U},\mathbf{u}$:&a multidimensional random variable affecting $Y$ only, an observation of $\mathbf{U}$\\
	$\zeta,\;\xi,\;\nu$:&noises/disturbances\\
	$\mathbb{E},\mathbb{V}$:&Expectation of random variables, Variance of random variables\\
	$(\Omega,\mathbb{P},\mathcal{F})$:&Filtered probability space such that $\Omega$ contains outcomes of $D$, $\mathbf{U}$, $\mathbf{X}$ and $\mathbf{Z}$\\
	$P_{i}(\mathbf{x},\mathbf{z}),\hat{P}_{i}(\mathbf{x},\mathbf{z})$:&$\mathbb{E}[\mathbf{1}_{\{D=d^{i}\}}\mid \mathbf{X}=\mathbf{x},\mathbf{Z}=\mathbf{z}]$, estimate of $P_{i}(\mathbf{x},\mathbf{z})$\\
	$\mathbf{P}(\mathbf{x},\mathbf{z})$:&$(P_{1}(\mathbf{x},\mathbf{z}),\cdots,P_{n}(\mathbf{x},\mathbf{z}))$\\
	$\mathcal{g},\;g$:&$\mathbb{P}$-integrable function which describe the relationship between $Y$ and $Y$'s predictor\\
	&variables such as $D$, $\mathbf{U}$ and $\mathbf{Z}$\\
	$\mathcal{m},\;m$:&$\mathbb{P}$-integrable function which describe the relationship between $D$ and $D$'s predictor\\
	&variables such as $\mathbf{X}$ and $\mathbf{Z}$\\
	$\theta^{i},\;\theta^{i\mid j}$:&$\mathbb{E}\left[g(d^{i},\mathbf{U},\mathbf{Z})\right]$ for each $i$, $\mathbb{E}\left[g(d^{i},\mathbf{U},\mathbf{Z})\mid D=d^{j}\right]$ for each $i$, $j$ and $i\neq j$\\
	$\theta^{i,k}$:&average treatment effect (ATE) $=\theta^{i}-\theta^{k}$\\
	$\theta^{i,k\mid j}$:&average treatment effect on the treated (ATTE) $=\theta^{i\mid j}-\theta^{k\mid j}$\\
	$\hat{\theta}_{o}^{i},\;\hat{\theta}_{o}^{i\mid j}\;\hat{\theta}_{o}^{i,k},\;\hat{\theta}_{o}^{i,k\mid j}$:&Impact inference without Confounding (IoC) estimates of $\theta^{i},\;\theta^{i\mid j}\;\theta^{i,k},\;\theta^{i,k\mid j}$\\
	$\hat{\theta}_{w}^{i},\;\hat{\theta}_{w}^{i\mid j}\;\hat{\theta}_{w}^{i,k},\;\hat{\theta}_{w}^{i,k\mid j}$:&Impact inference with Confounding (IwC) estimates of $\theta^{i},\;\theta^{i\mid j}\;\theta^{i,k},\;\theta^{i,k\mid j}$\\
	$\psi$:&score function\\
	$\varrho,\;\rho$:&nuisance parameters,``true" nuisance parameters\\
	&(The ``true" nuisance parameters are $g(d^{i},\cdot,\cdot)$, $\mathbb{E}\left[\mathbf{1}_{\{D=d^{i}\}}\mid\mathbf{X},\mathbf{Z}\right]$ and $\mathbb{E}\left[\mathbf{1}_{\{D=d^{i}\}}\right]$.)
\end{longtable}

\clearpage
\subsection{Examples}\label{sec:Appendix-Examples}
\textbf{A Tangible Example}\quad
We now work through an example to illustrate the difference between the estimations from causal impact inference and the estimations from non-causal impact inference. In this example, the outcome $Y$ is the delinquency probability of an customer indicating how likely he will be late for the next payment. The intervention $D$ is the credit limit set by the e-commerce lender. The common factor $\mathbf{Z}$ that confounds both $Y$ and $D$ is the monthly income. The intervention-specific variable $\mathbf{X}$ is his age, and the outcome-specific variable $\mathbf{U}$ is his monthly expenditure. Generally, when the individual's income $\mathbf{Z}$ is high and the expenditure $\mathbf{U}$ is low, he is less likely to be late for payment. Meanwhile, the lender tends to increase the credit lines $D$ as customers' higher income $\mathbf{Z}$ grow and decreases their credit lines as they age.
Suppose there are two types of credit lines $D=\$1000$ and $D=\$2000$. Let's assume the relationship between $Y$ and $(D,\mathbf{U},\mathbf{Z})$ as well as the relationship between $D$ and $(\mathbf{X},\mathbf{Z})$ are
\begin{equation}
\begin{gathered}\label{eqt:ATE and ATTE example}
Y(\mathbf{U},\mathbf{Z};D)=g(D,\mathbf{U},\mathbf{Z})+\xi,\quad\text{ where }\\
g(D,\mathbf{U},\mathbf{Z})=\begin{cases}
\frac{\mathbf{-Z}+\mathbf{U}+10000}{100000},&\; D=1000\\
\frac{\mathbf{-Z}+\mathbf{U}+60000}{100000},&\; D=2000
\end{cases}\quad\text{ and }\quad
D=\begin{cases}
1000,&\; \frac{\mathbf{Z}-10\mathbf{X}}{5000}<1\\
2000,&\; \frac{\mathbf{Z}-10\mathbf{X}}{5000}\geq 1
\end{cases}\text{\quad respectively.}
\end{gathered}
\end{equation}
Here, $\xi$ is the normal distributed noise such that it follows $N(0,0.001)$. Suppose the training set contains records of $10$ customers whose $(\mathbf{z}_{m},\mathbf{x}_{m},\mathbf{u}_{m},d_{m},y_{m})_{m=1}^{10}$ are:
\begin{equation*}
\begin{aligned}
&(1000,21,500,1000,0.095),\;(2000,22,1000,1000,0.090),\\
&(3000,23,1500,1000,0.087),\;(4000,24,2000,1000,0.800),\\
&(5000,25,2500,1000,0.075),\;(6000,26,3000,2000,0.569),\;(7000,27,3500,2000,0.566),\\
&(8000,28,4000,2000,0.562),\;(9000,29,4500,2000,0.553),\;(10000,30,5000,2000,0.551).
\end{aligned}
\end{equation*}
We use these ten observations to train the model $Y(\mathbf{U},\mathbf{Z};D)$ in \eqref{eqt:ATE and ATTE example} using linear regression and linear logistic regression and we can obtain the trained model $\hat{Y}(\mathbf{U},\mathbf{Z};D)$ and $\hat{P}_{2000}(\mathbf{X},\mathbf{Z})=\hat{\mathbb{E}}[\mathbf{1}_{\{D=2000\}}\mid \mathbf{X},\mathbf{Z}]$ such that
\begin{align}\label{eqt:ATE and ATTE example trained function}
\hat{Y}(\mathbf{U},\mathbf{Z};D)&=\begin{cases}
\num{-4.046756} \times 10^{-6} \mathbf{Z}\num{-2.023378}\times 10^{-6}\mathbf{U}+\num{0.100655}, & D=1000\\
\num{-4.218229}\times 10^{-6}\mathbf{Z}\num{-2.10911}\times 10^{-6} \mathbf{U}+\num{0.10217}, & D=2000
\end{cases},
\\
\hat{P}_{2000}(\mathbf{X},\mathbf{Z})&=1/\big[1+\exp\left(0.0031\mathbf{Z}-0.6648\mathbf{X}-0.0332\right)\big]\label{eqt:estimated probability example trained function}.
\end{align}
Suppose another four individuals whose $(\mathbf{z}_{m},\mathbf{x}_{m},\mathbf{u}_{m},d_{m},y_{m})_{m=1}^{4}$ are
\begin{equation*}
\begin{aligned}
&(3000,20,500,1000,0.075),\;(4000,22,1000,1000,0.071),\\
&(8000,24,1500,2000,0.533),\;(8000,26,3000,2000,0.541).
\end{aligned}
\end{equation*}
From \eqref{eqt:ATE and ATTE example}, the true potential outcomes $(Y(\mathbf{u},\mathbf{z};1000),Y(\mathbf{u},\mathbf{z};2000))$ of the four individuals are
\begin{equation*}
\begin{aligned}
(0.075,0.575),\;(0.071,0.571),\;(0.037,0.533),\;(0.040,0.541).
\end{aligned}
\end{equation*}
From \eqref{eqt:ATE and ATTE example trained function}, the estimations of the potential outcomes $(\hat{Y}(\mathbf{u},\mathbf{z};1000),\hat{Y}(\mathbf{u},\mathbf{z};2000))$ of the four individuals are
\begin{equation*}
\begin{aligned}
(\num[round-precision=4]{0.08750311},\num[round-precision=4]{0.08846092}),\;(\num[round-precision=4]{0.08244466},\num[round-precision=4]{0.08318813}),\;(\num[round-precision=4]{0.06524595},\num[round-precision=4]{0.06526066}),\;(\num[round-precision=4]{0.05816412},\num[round-precision=4]{0.05787875}).
\end{aligned}
\end{equation*}
From \eqref{eqt:estimated probability example trained function}, the estimated probabilities are
\begin{equation*}
\begin{aligned}
&(0.9818,0.0182),\;(0.9005,0.0995),\;(0.0001,0.9999),\;(0.0000,1.0000).
\end{aligned}
\end{equation*}
We then compute the IoC estimations and IwC estimations of $\theta^{1}$, $\theta^{2}$ and $\theta^{2\mid 1}$. Using \eqref{eqt:IoC estimator for IwC formulation} in the main paper, the IoC estimations $\hat{\theta}_{o}^{1}$, $\hat{\theta}_{o}^{2}$ and $\hat{\theta}_{o}^{2\mid 1}$ of $\theta^{1}$, $\theta^{2}$ and $\theta^{2\mid 1}$ are computed such that
\begin{equation*}
\begin{aligned}
\hat{\theta}^{1}_{o}&=\frac{(0.0875+0.0824+0.0652+0.0582)}{4}=0.0733\\
\hat{\theta}^{2}_{o}&=\frac{(0.0885+0.0832+0.0653+0.0579)}{4}=0.0737\\
\hat{\theta}_{o}^{1\mid 2}&=\frac{(0.0652+0.0582)}{2}=0.0617.
\end{aligned}
\end{equation*}
On the other hand, we compute the IwC estimations $\hat{\theta}_{w}^{1}$, $\hat{\theta}_{w}^{2}$ and $\hat{\theta}_{w}^{2\mid 1}$ using \eqref{eqt:desired unbiased estimator expectation} and \eqref{eqt:desired unbiased estimator conditional expectation} in the main paper, which give
\begin{equation*}
\begin{aligned}
\hat{\theta}^{1}_{w}&=\frac{[(0.0875+0.0824+0.0652+0.0582)+\frac{(0.075-0.0875)}{0.9818}+\frac{(0.071-0.0824)}{0.9005}]}{4}=0.0670\\
\hat{\theta}^{2}_{w}&=\frac{[(0.0885+0.0832+0.0653+0.0579)+\frac{(0.533-0.0653)}{0.9999}+\frac{(0.541-0.0579)}{1.0000}]}{4}=0.3114\\
\hat{\theta}_{w}^{1\mid 2}&=\frac{[(0.0652+0.0582)+0.0182\cdot\frac{(0.075-0.0875)}{0.9818}+0.0995\cdot\frac{(0.071-0.0824)}{0.9005}]}{2}=0.0610.
\end{aligned}
\end{equation*}
Meanwhile, we compute the true ATE and true ATTE, which are
\begin{equation*}
\begin{aligned}
\text{true ATE:}&=\theta^{1}-\theta^{2}\\
&=\frac{[(0.075-0.575)+(0.071-0.571)+(0.037-0.533)+(0.040-0.541)]}{4}=-0.4993\\
\text{true ATTE:}&=\theta^{1\mid 2}-\theta^{2\mid 2}=\frac{(0.037+0.040)}{2}-\frac{(0.533+0.541)}{2}=-0.4985.
\end{aligned}
\end{equation*}
Next, we compute the IoC ATE and the IoC ATTE, which are
\begin{equation*}
\begin{aligned}
\text{IoC ATE:}&=\hat{\theta}^{1}_{o}-\hat{\theta}^{2}_{o}=0.0733-0.0737=-0.0004\\
\text{IoC ATTE:}&=\hat{\theta}_{o}^{1\mid 2}-\theta^{2\mid 2}=0.0617-\frac{0.533+0.541}{2}=-0.4753.
\end{aligned}
\end{equation*}
Last but not least, we compute the IwC ATE and the IwC ATTE, which are
\begin{equation*}
\begin{aligned}
\text{IwC ATE:}&=\hat{\theta}^{1}_{w}-\hat{\theta}^{2}_{w}=0.0670-0.3114=-0.2444\\
\text{IwC ATTE:}&=\hat{\theta}_{w}^{1\mid 2}-\theta^{2\mid 2}=0.0610-\frac{0.533+0.541}{2}=-0.4760.
\end{aligned}
\end{equation*}
From the above calculations, we notice that the IwC estimations give better estimations than the IoC estimations.
\subsection{Details of models for simulated data}\label{sec:Appendix-Details of Models for simulated data} 
In the main paper, we give out the explicit form of $q(\mathbf{U},\mathbf{Z})$ and the $D$ construction, which are
\begin{equation*}
\begin{aligned}
q(\mathbf{U},\mathbf{Z})&=\left\{e_{1}\log\left(e_{2}+k(\mathbf{Z})^{2}+|k(\mathbf{U})|^{\tau}\right)+\exp\left(\left|\mathbf{a}_{0}^{T}\mathbf{Z}\right|\right)\right\}^{r},\\
D&=\sigma\left(n(\mathbf{X},\mathbf{Z},\nu)\right)=\sigma\left(\lambda\left(\mathbf{a}_{1}^{T}\mathbf{X}+\left|\mathbf{a}_{2}^{T}\mathbf{Z}\right|^{\gamma}+\frac{b_{1}\; X_{9}}{1+|X_{3}|}+\frac{b_{2}\; X_{10}}{1+|X_{6}|}\right)+\nu\right).
\end{aligned}
\end{equation*}
These forms are chosen for the simulated experiment and the empirical experiment. Differences come from the choices of $k(\mathbf{Z})$ and $k(\mathbf{U})$. In this section, we give out the numerical values of the coefficients which appear in the terms $k(\mathbf{Z})$ and $k(\mathbf{U})$. Besides, we determine the credit limit $D$ from the function $n(\mathbf{X},\mathbf{Z},\nu)$. To assign the credit limit of each individual in a population, we need to divide the whole population into 5 categories based on the values of $n(\mathbf{X},\mathbf{Z},\nu)$ of the whole population: 1) $<=20$ percentile of $n(\mathbf{X},\mathbf{Z},\nu)$ of the whole population, 2) in between 20 percentile and 40 percentile of $n(\mathbf{X},\mathbf{Z},\nu)$ of the whole population, 3) in between 40 percentile and 60 percentile of $n(\mathbf{X},\mathbf{Z},\nu)$ of the whole population, 4) in between 60 and 80 percentile of $n(\mathbf{X},\mathbf{Z},\nu)$ of the whole population and 5) $>=80$ percentile of $n(\mathbf{X},\mathbf{Z},\nu)$ of the whole population. Under each category, the consumers are assigned to have the same credit limit which equals to the median value of $n(\mathbf{X},\mathbf{Z},\nu)$ of that category.
\paragraph{Coefficients of $k(\mathbf{Z})$ and $k(\mathbf{U})$ which are used in the simulated experiment} In the simulated experiment, the quantities $k(\mathbf{Z})$ and $k(\mathbf{U})$ are given as
\begin{equation}
\begin{aligned}
k(\mathbf{Z})&=\log\left(\left|(\mathbf{c}_{1}^{z})^{T}\mathbf{Z}+\overset{4}{\underset{r=2}{\sum}}(\mathbf{c}_{r}^{z})^{T}\bar{\mathbf{Z}}_{r}\right|\right)+(\mathbf{c}_{1}^{z})^{T}\log|\mathbf{Z}|+\overset{4}{\underset{r=2}{\sum}}(\mathbf{c}_{r}^{z})^{T}\log|\bar{\mathbf{Z}}_{r}|,\\
k(\mathbf{U})&=\log\left(\left|(\mathbf{c}_{1}^{u})^{T}\mathbf{U}+\overset{4}{\underset{r=2}{\sum}}(\mathbf{c}_{r}^{u})^{T}\bar{\mathbf{U}}_{r}\right|\right)+(\mathbf{c}_{1}^{u})^{T}\log|\mathbf{U}|+\overset{4}{\underset{r=2}{\sum}}(\mathbf{c}_{r}^{u})^{T}\log|\bar{\mathbf{U}}_{r}|
\end{aligned}
\end{equation}
where $\bar{\mathbf{Z}}_{r}$ (or $\bar{\mathbf{U}}_{r}$) is a column vector such that each entry is the product of $r$ elements taken from $\mathbf{Z}$ (or $\mathbf{U}$) without replacement. For example, $\bar{\mathbf{Z}}_{2}=[Z_{1}Z_{2},Z_{1}Z_{3},\cdots,Z_{19}Z_{20}]^{T}$, $\bar{\mathbf{Z}}_{3}=[Z_{1}Z_{2}Z_{3},Z_{1}Z_{2}Z_{4},\cdots,Z_{18}Z_{19}Z_{20}]^{T}$ and $\bar{\mathbf{Z}}_{4}=[Z_{1}Z_{2}Z_{3}Z_{4},Z_{1}Z_{2}Z_{3}Z_{5},\cdots,Z_{17}Z_{18}Z_{19}Z_{20}]^{T}$.

We give detailed descriptions about the quantities $\mathbf{c}_{1}^{z}$, $\mathbf{c}_{2}^{z}$, $\mathbf{c}_{3}^{z}$, $\mathbf{c}_{4}^{u}$, $\mathbf{c}_{1}^{u}$, $\mathbf{c}_{2}^{u}$, $\mathbf{c}_{3}^{u}$ and $\mathbf{c}_{4}^{u}$. As stated in the main paper, $\mathbf{Z}$ is a $20\times 1$ column vector and $\mathbf{U}$ is a $10\times 1$ column vector. The dimensions of the quantities $\mathbf{c}_{1}^{z}$, $\mathbf{c}_{2}^{z}$, $\mathbf{c}_{3}^{z}$, $\mathbf{c}_{4}^{u}$, $\mathbf{c}_{1}^{u}$, $\mathbf{c}_{2}^{u}$, $\mathbf{c}_{3}^{u}$ and $\mathbf{c}_{4}^{u}$ are $20\times 1$, $190\times 1$, $1140\times 1$, $4845\times 1$, $10\times 1$, $45\times 1$, $120\times 1$ and $210\times 1$ respectively such that the values are generated from the normal distribution with mean $0$ and variance $1$.

Finally, we present the coefficients $\mathbf{a}_{0}$, $e_{1}$, $e_{2}$, $\tau$, $r$ presented in $q(\mathbf{U},\mathbf{Z})$ and $\lambda$, $\gamma$, $b_{1}$, $b_{2}$, $\mathbf{a}_{1}$, $\mathbf{a}_{2}$ presented in $n(\mathbf{X},\mathbf{Z},\nu)$, which are
\begin{equation*}
\begin{aligned}
e_{1}&=0.1,\; e_{2}=0,\; \tau=1.5,\; r=0.5,\; \lambda=1,\; \gamma=2,\;  b_{1}=1.775,\; b_{2}=-1.354,\\
\mathbf{a}_{0} &= \left[
{0.15}, {0.15}, {0.15}, \num{0}, {0.15}, \num{0}, \num{0},\num{0},\num{0},\num{0},\num{0},\num{0},\num{0},\num{0},\num{0},\num{0},\num{0},\num{0},\num{0},\num{0}\right]^{T},\\ 
\mathbf{a}_{1} &= \left[
\num{0.079702899999999993196908576464921}, \num{-1.7610222300000000217323758988641}, \num{-0.91346963000000003241041213186691}, \num{0.68418343999999997606664692284539}, \num{0.64784691999999999278969653460081}, \num{0.67517954000000002245940322609385}, \num{1.4604472699999999640851910953643}, \num{-0.77882898000000000404696720579523}, \num{-0.21162480000000000179838366420881}, \num{1.0102116699999998949976998119382}
\right]^{T},\\
\mathbf{a}_{2} &= \left[
\num{1.1537280100000000260251908912323}, \num{1.2269035200000000251918663707329}, \num{0.68461399999999994481214571351302}, \num{1.4707825199999999821187657289556}, \num{-0.42191908000000000189544380191364}, \num{-0.11914524999999999421707030933248},\num{-0.24666432999999998720852545375237}, \num{0.010599449999999999844080278421643}, \num{-0.70721460000000002654729769346886}, \num{-0.47233888000000001650136027819826}, \right.\\
&\qquad\qquad \left. \num{0.41497956000000002507732688172837}, \num{0.94239603000000005117442469781963}, \num{2.2511959500000000566899416298838}, \num{-0.10741729999999999334292510866362}, \num{0.41419468999999997649297256430145}, \num{-0.46030412999999997802902385046764}, \num{1.5260769899999999665851646568626}, \num{0.56253337000000003254029934396385}, \num{-0.82265637000000002565514023444848}, \num{-0.9987326100000000206335926122847}
\right]^{T}.
\end{aligned}
\end{equation*}
\paragraph{Coefficients of $k(\mathbf{Z})$ and $k(\mathbf{U})$ which are used in the empirical experiment} In the empirical experiment, the quantities $k(\mathbf{Z})$ and $k(\mathbf{U})$ are given as
\begin{equation}
\begin{aligned}
k(\mathbf{Z})&=(\mathbf{c}_{1}^{z})^{T}\mathbf{Z}\qquad\text{and}\qquad k(\mathbf{U})=(\mathbf{c}_{1}^{u})^{T}\mathbf{U}.
\end{aligned}
\end{equation}

We give detailed descriptions about the quantities $\mathbf{c}_{1}^{z}$ and $\mathbf{c}_{1}^{u}$. $\mathbf{Z}$ is a $20\times 1$ column vector and $\mathbf{U}$ is a $10\times 1$ column vector such that
\begin{equation*}
\begin{aligned}
\mathbf{c}_{1}^{z} &= \left[
\num{-0.09656428735342949}, \num{0.9734025822882408}, \num{-0.8582891237280186}, \num{-0.30860415187183593}, \num{-0.28146793638372014}, \num{-2.1733362770260594},\num{1.07992405679886}, \num{-0.48654644176178413}, \num{0.9186030042716846}, \num{-0.3898601365276411}, \right.\\
&\qquad\qquad \left. \num{0.5508196301315733}, \num{1.214349568204593}, \num{0.5911358538192414}, \num{-0.039596488856784254}, \num{-0.2689273104251347}, \num{-0.49321543083682934}, \num{0.20597029681886586}, \num{-0.4562984857686687}, \num{1.0593842547558803}, \num{-0.2403438672537869}
\right]^{T},\\
\mathbf{c}_{1}^{u} &= \left[
\num{0.40284599}, \num{-1.72689173}, \num{-1.06178813}, \num{-1.34073716}, \num{1.48927656}, \num{-1.10148294}, \num{-0.31908929}, \num{-1.93599287}, \num{0.23803084}, \num{-0.00819786}
\right]^{T}.
\end{aligned}
\end{equation*}

Finally, we present the coefficients $\mathbf{a}_{0}$, $e_{1}$, $e_{2}$, $\tau$, $r$ presented in $q(\mathbf{U},\mathbf{Z})$ and $\lambda$, $\gamma$, $b_{1}$, $b_{2}$, $\mathbf{a}_{1}$, $\mathbf{a}_{2}$ presented in $n(\mathbf{X},\mathbf{Z},\nu)$, which are
\begin{equation*}
\begin{aligned}
e_{1}&=1,\; e_{2}=1,\; \tau=1.5,\; r=0.5,\; \lambda=0.0507,\; \gamma=0.5896,\;  b_{1}=0,\; b_{2}=0,\\
\mathbf{a}_{0} &= \left[
\num{-0.09656428735342949}, \num{0.9734025822882408}, \num{-0.8582891237280186}, \num{-0.30860415187183593}, \num{-0.28146793638372014}, \num{-2.1733362770260594}, \num{1.07992405679886}, \num{-0.48654644176178413}, \num{0.9186030042716846}, \num{-0.3898601365276411},\right.\\
&\qquad\qquad \left.\num{0.5508196301315733}, \num{1.214349568204593}, \num{0.5911358538192414}, \num{-0.039596488856784254}, \num{-0.2689273104251347}, \num{-0.49321543083682934}, \num{0.20597029681886586}, \num{-0.4562984857686687}, \num{1.0593842547558803}, \num{-0.2403438672537869}
\right]^{T},\\ 
\mathbf{a}_{1} &= \left[
\num{0.079702899999999993196908576464921}, \num{-1.7610222300000000217323758988641}, \num{-0.91346963000000003241041213186691}, \num{0.68418343999999997606664692284539}, \num{0.64784691999999999278969653460081}, \num{0.67517954000000002245940322609385}, \num{1.4604472699999999640851910953643}, \num{-0.77882898000000000404696720579523}, \num{-0.21162480000000000179838366420881}, \num{1.0102116699999998949976998119382}
\right]^{T},\\
\mathbf{a}_{2} &= \left[
\num{1.1537280100000000260251908912323}, \num{1.2269035200000000251918663707329}, \num{0.68461399999999994481214571351302}, \num{1.4707825199999999821187657289556}, \num{-0.42191908000000000189544380191364}, \num{-0.11914524999999999421707030933248},\num{-0.24666432999999998720852545375237}, \num{0.010599449999999999844080278421643}, \num{-0.70721460000000002654729769346886}, \num{-0.47233888000000001650136027819826}, \right.\\
&\qquad\qquad \left. \num{0.41497956000000002507732688172837}, \num{0.94239603000000005117442469781963}, \num{2.2511959500000000566899416298838}, \num{-0.10741729999999999334292510866362}, \num{0.41419468999999997649297256430145}, \num{-0.46030412999999997802902385046764}, \num{1.5260769899999999665851646568626}, \num{0.56253337000000003254029934396385}, \num{-0.82265637000000002565514023444848}, \num{-0.9987326100000000206335926122847}
\right]^{T}.
\end{aligned}
\end{equation*}
\clearpage
\subsection{Tables and Figures}\label{sec:Appendix-Tables and Figures}
In Table \ref{table:statistics of variables-full}, we give a full summary of the empirical data which is used to generate semi-synthetic data to study the difference between the IoC estimations and the IwC estimations.
Table \ref{table:light tail and heavy tail with alpha is 0.05 ATTE} and Table \ref{table:light tail and heavy tail with beta is 0.05 ATTE} summarize the ATTE estimations comprehensively across different choices of regressors $\hat{g}$ and different data properties per \eqref{eqt:simulated data model 1} \& \eqref{eqt:simulated data model 3} in the main paper based on simulated data. Table \ref{table:light tail and heavy tail with alpha is 0.05 ATTE} reports results based on the simulated data which examines strong causality while Table \ref{table:light tail and heavy tail with beta is 0.05 ATTE} reports results based on the simulated data which examines strong nonlinearity. In these comparisons, the regressors include the ordinary least square (OLS), LASSO, RIDGE, random forest (RF), xgboost (XGB), gated recurrent unit (GRU), convolutional neural network (CNN) and multi-layer perception (MLP). Moreover, the data properties cover both the light-tail case and the heavy-tail case. 
%
%
For the light-tail case, we assume that $\mathbf{U}$, $\mathbf{X}$ and $\mathbf{Z}$ follow the multivariate normal distributions with mean $0$ and different covariance matrices $[C_{ij}]$ which are $0.8+0.2\exp(-0.2|i-j|)$, $0.2+0.8\exp(-2|i-j|)$ and $0.5+0.5\exp(-0.5|i-j|)$ respectively. For the heavy-tail case, we assume that $\mathbf{U}$, $\mathbf{X}$ and $\mathbf{Z}$ follow the multivariate student-$t$ distributions with mean $0$ and different covariance matrices $[C_{ij}]$ which are $0.8+0.2\exp(-0.2|i-j|)$, $0.5+0.5\exp(-0.5|i-j|)$ and $0.5+0.5\exp(-0.5|i-j|)$ respectively. The corresponding degree of freedoms in generating the student-$t$ distribution of $\mathbf{U}$, $\mathbf{X}$ and $\mathbf{Z}$ are $10$, $10$ and $5$ respectively.
In Table \ref{table:ATTE-alpha_beta real}, we summarize the weighted average of ATTE in relative errors w.r.t the true ATTE using the semi-synthetic dataset when $\alpha$ is fixed at $5\%$ and $\beta$ is fixed at $5\%$.
Figure \ref{fig:consistency graph conditional expectation simulated} is a consistency analysis of $\hat{\theta}_{w}^{i\mid j}$. Figure \ref{fig:ATE_3_model_histogram real} is a histogram plot of weighted average of relative error of estimated ATE while Figure \ref{fig:ATTE_3_model_histogram real} is a histogram plot of weighted average of relative error of estimated ATTE based on the semi-synthetic data. The experiments are conducted for $500$ times and eight different models are used to train the function $g$, including OLS, LASSO, RIDGE, XGB, RF, GRU, CNN and MLP.

Finally, we state out the metrics generating the tables and figures, which are \eqref{eqt:consistency_theta_ij_mean}, \eqref{eqt:consistency_theta_ij_std} and \eqref{eqt:weighted average of relative error ATTE}. Indeed, \eqref{eqt:consistency_theta_ij_mean} and \eqref{eqt:consistency_theta_ij_std} are the mean and standard deviation of estimated $\theta^{i|j}$ respectively. Besides, \eqref{eqt:weighted average of relative error ATTE} is the weighted average of estimated $\theta^{i,k|j}$ in relative errors with true $\theta^{i,k|j}$.
%
\begin{subequations}
{\small
	\begin{equation}\label{eqt:consistency_theta_ij_mean}
	\begin{aligned}
	\frac{1}{n(n-1)}\overset{}{\underset{\underset{i \neq j}{1\leq i,j\leq n}}{\sum}}\left|\frac{\overset{M}{\underset{m=1}{\sum}}{\hat{\theta}}_{w}^{i|j;m}}{\overset{M}{\underset{m=1}{\sum}}\theta^{i|j;m}}-1\right|
	\end{aligned}
	\end{equation}
	\begin{equation}\label{eqt:consistency_theta_ij_std}
	\begin{aligned}
	\frac{1}{n(n-1)}\overset{}{\underset{\underset{i \neq j}{1\leq i,j\leq n}}{\sum}}\sqrt{\frac{1}{M-1}\overset{M}{\underset{m=1}{\sum}}\left([\hat{\theta}_{w}^{i|j;m}-\theta^{i|j;m}]-\frac{1}{M}\overset{M}{\underset{s=1}{\sum}}[\hat{\theta}_{w}^{i|j;s}-\theta^{i|j;s}]\right)^{2}}
	\end{aligned}
	\end{equation}
\begin{equation}
\begin{aligned}
\label{eqt:weighted average of relative error ATTE}
\frac{1}{M}\overset{M}{\underset{m=1}{\sum}}\;\underset{\underset{i\neq k \neq j}{1\leq i,k,j\leq n}}{\sum}
\left[ 
\frac{\left|\theta^{i,k|j;m}\right|}{  \sum{_{\underset{i\neq k \neq j}{1\leq i,k,j\leq n}}} \left|\theta^{i,k|j;m}\right| }\left|\frac{\hat{\theta}^{i,k|j;m}}{\theta^{i,k|j;m}}-1\right|
\right]
\end{aligned}
\end{equation}
}
\end{subequations}
\begin{flushleft}
Here, $M$ is the number of experiments conducted, $n$ is the number of treatments, $\theta^{i|j;m}$ and $\theta^{i,k|j;m}$ are the true value in the $m^{\textrm{th}}$ experiment, and  $\hat{\theta}^{i|j;m}$ and $\hat{\theta}^{i,k|j;m}$ are their estimate based on IwC (or IoC). 
\end{flushleft}

\begin{table}
	\caption{\small Full statistics of the mean, standard deviation, skewness, excess kurtosis, minimum, maximum and $5\%$, $25\%$, $50\%$, $75\%$, $95\%$ quantiles of the credit loan data.}
	\label{table:statistics of variables-full}
	\resizebox{\linewidth}{!}{%
		\centering
		\begin{tabular}{cccccccccccccccc}
			\toprule
			$\mathbf{\textit{N}}$ & $\textbf{mean}$ & $\textbf{std}$ & $\textbf{skewness}$ & $\textbf{kurtosis}$ & $\textbf{min}$ & $\textbf{max}$ & $\mathbf{5\%}$ & $\mathbf{25\%}$ & $\mathbf{50\%}$ & $\mathbf{75\%}$ & $\mathbf{95\%}$\\
			\midrule
			\textbf{credit limit} & \num{4750.7639424999997572740539908409}&	\num{2383.1878141451397823402658104897}&	\num{0.62026806528459299538980076249572}&	\num{3.9016062833883600191597906814422}&	\num[round-precision=0]{100.0}&	\num[round-precision=0]{50000.0}&	\num[round-precision=0]{600.0}&	\num[round-precision=0]{3000.0}&	\num[round-precision=0]{4800.0}&	\num[round-precision=0]{6000.0}&	\num[round-precision=0]{9000.0}\\
			$x_{1}$&\num{0.43570695474969101113060787611175}&	\num{0.16260669562456100956104876331665}&	\num{0.48414282203062297682549797173124}&	\num{0.15818788694934099448019537703658}&	\num{1.2130701e-28} &	\num[round-precision=0]{1}&	\num{0.19493635649999999093218150392204}&	\num{0.32123119999999999452100496455387}&	\num{0.41935048999999902097357562524849}&	\num{0.53493368749999903322134287009249}&	\num{0.73419440999999896479977223862079}\\
			$x_{2}$&\num{0.35543009479415299534110772583517}&	\num{0.15629726006588501308058880567842}&	\num{1.095672504357960042042918757943}&	\num{1.4842769348800499695784083087347}&	\num[round-precision=0]{0.0}&	\num[round-precision=0]{1.0}&	\num{0.15456493199999998822136149101425}&	\num{0.25091799500000000477939465781674}&	\num{0.32525062500000001541877736599417}&	\num{0.42589326249999998053041849743749}&	\num{0.67029882199999901715159467130434}\\
			$x_{3}$&\num{0.47785278011265497344339792107348}&	\num{0.18151822651955901166687112890941}&	\num{0.10108142973965300592809057889099}&	\num{-0.523694845042341983187839105085}&	\num[round-precision=0]{0.0}&	\num[round-precision=0]{1.0}&	\num{0.19036235749999999566917097126861}&	\num{0.34386252499999997445101485027408}&	\num{0.47158515499999997810931517960853}&	\num{0.60895599499999997217258851378574}&	\num{0.78394242849999895561552420986118}\\
			$x_{4}$&\num{0.53611164701833502466143954734434}&	\num{0.18202545588009200194790082605323}&	\num{-0.3107180063338320086252508644975}&	\num{-0.4005292024821500151965381064656}&	\num[round-precision=0]{0.0}&	\num[round-precision=0]{1.0}&	\num{0.21281249950000000192851246083592}&	\num{0.4109937449999990222160306529986}&	\num{0.55211674999999904578373843833106}&	\num{0.67176779999999902681651064995094}&	\num{0.80900629299999904375795267696958}\\
			$x_{5}$&\num{0.50170605579501204029924110727734}&	\num{0.16705136057589400899736631345149}&	\num{-0.043520554933936403396721459557739}&	\num{-0.32693368062887800951088479450846}&	\num[round-precision=0]{0.0}&	\num[round-precision=0]{1.0}&	\num{0.22705369499999999982620124683308}&	\num{0.38321779499999897300455131698982}&	\num{0.50315642999999998785654042876558}&	\num{0.62153442249999901747514741146006}&	\num{0.77216439149999904767440739306039}\\
			$x_{6}$&\num{0.45140319917442800479179254580231}&	\num{0.14962560866752799837442466923676}&	\num{0.40420170738818800293401523049397}&	\num{0.16647395268423900716037167057948}&	\num[round-precision=0]{0.0}&	\num[round-precision=0]{1.0}&	\num{0.23052196949999900765782001599291}&	\num{0.34590243249999902408475804804766}&	\num{0.43732346499999902311728305903671}&	\num{0.54576334499999901161970683460822}&	\num{0.7190600114999989855846251884941}\\
			$x_{7}$&\num{0.50474042459551304862230836079107}&	\num{0.15720757227686699142310544630163}&	\num{-0.091050462700349402145327815105702}&	\num{-0.0060015150473331502212204213719815}&	\num[round-precision=0]{0.0}&	\num[round-precision=0]{1.0}&	\num{0.23735045350000000263790411736409}&	\num{0.40047061249999899024132332669978}&	\num{0.51039389999999995595203472475987}&	\num{0.61074233749999995435331356929964}&	\num{0.75636610500000001078291234080098}\\
			$x_{8}$&\num{0.50088014814718295752982157864608}&	\num{0.1747915991476670094506573605031}&	\num{0.17047273018469499827887148057926}&	\num{-0.37866385804333901887730462476611}&	\num[round-precision=0]{0.0}&	\num[round-precision=0]{1.0}&	\num{0.23267932499999999240714032566757}&	\num{0.37401206749999998946165646884765}&	\num{0.48881664499999999407009454444051}&	\num{0.62324797499999995409325492801145}&	\num{0.80375959299999999441155296153738}\\
			$x_{9}$&\num{0.53369370557267703514270351661253}&	\num{0.15213521764842799499639625082636}&	\num{-0.33375358434002100338489071873482}&	\num{0.056857297653664297387710035991404}&	\num{3.8693236e-32}&	\num[round-precision=0]{1.0}&	\num{0.26563953000000001258840143236739}&	\num{0.43687785000000001200959331981721}&	\num{0.54284749999999903913305843161652}&	\num{0.64103876250000002556816980359145}&	\num{0.76647010500000001265163973585004}\\
			$x_{10}$&\num{0.55033424627606297718784844619222}&	\num{0.16609046242251698788372493709176}&	\num{-0.31659403085895998497534264970454}&	\num{0.062514676108409800159293467913812}&	\num{5.7953964e-32}&	\num[round-precision=0]{1.0}&	\num{0.26970070050000000128065380522457}&	\num{0.43934854000000000961989599090884}&	\num{0.55969199499999999769528358228854}&	\num{0.66968852499999997895940850867191}&	\num{0.80561105249999998001442236272851}\\
			$z_{1}$&\num{0.52257147512797896826697297001374}&	\num{0.17721332796673200382642221484275}&	\num{-0.23599684689741801113527230882028}&	\num{-0.44564375901910902033620232032263}&	\num{8.0244887e-35}&	\num[round-precision=0]{1.0}&	\num{0.22108350199999998730859829265682}&	\num{0.39577681749999898963565669873788}&	\num{0.53269085000000004903597528027603}&	\num{0.65767335999999998463749761867803}&	\num{0.79177621999999903135147860666621}\\
			$z_{2}$&\num{0.49544448028780502246704031676927}&	\num{0.16422003275190399929073237217381}&	\num{0.063587401559871997713813129848859}&	\num{-0.13786734788979498889105457237747}&	\num{3.9549168E-27}&	\num[round-precision=0]{1.0}&	\num{0.22899428599999999134340100681584}&	\num{0.38137290749999902539357776731777}&	\num{0.49403034000000001224606194227817}&	\num{0.60740780999999999245630988298217}&	\num{0.76842611499999902147095554028056}\\
			$z_{3}$&\num{0.42081834238723397900727718479175}&	\num{0.1578412291587029947503850735302}&	\num{0.43799065384525098787449337578437}&	\num{0.21611596989006298663582583685638}&	\num[round-precision=0]{0.0}&	\num[round-precision=0]{1.0}&	\num{0.18271210749999999833903530088719}&	\num{0.30596782999999999619333834743884}&	\num{0.41124303499999997901781512155139}&	\num{0.52230288999999996413237113301875}&	\num{0.69273910499999902246059946264722}\\
			$z_{4}$&\num{0.60215352213873696296531079497072}&	\num{0.23369188456080799976000150763866}&	\num{-0.55752657291486895640986176658771}&	\num{-1.3358899773025700863371412197012}&	\num[round-precision=0]{0.0}&	\num[round-precision=0]{1.0}&	\num{0.23397520750000000422375023845234}&	\num{0.32901011000000002226784090453293}&	\num{0.72543975000000004982325663149823}&	\num{0.78791525250000005176076456336887}&	\num{0.85652687749999900646002970461268}\\
			$z_{5}$&\num{0.46049214853104297739960770741163}&	\num{0.15584976727897101311803851331206}&	\num{0.10911590830869899448885007586796}&	\num{-0.093198429139645697083516040493123}&	\num{4.538858e-24}&	\num[round-precision=0]{1.0}&	\num{0.20920313249999999971606712279026}&	\num{0.35361616499999998231018594196939}&	\num{0.45588944999999997387263306336536}&	\num{0.56436625499999903965431258257013}&	\num{0.72511571050000001026347717925091}\\
			$z_{6}$&\num{0.43969427677691397260417716097436}&	\num{0.16440243120383699060305104922008}&	\num{0.41339095770966799436862970651418}&	\num{-0.10576388887713900066600558602659}&	\num{1.8881841e-24}&	\num[round-precision=0]{1.0}&	\num{0.19660353499999999593583765999938}&	\num{0.31920618000000000602511818215135}&	\num{0.42442524500000000653443521514419}&	\num{0.54654076499999904026338981566369}&	\num{0.73579605949999904446201526297955}\\
			$z_{7}$&\num{0.49473457478292698352717593479611}&	\num{0.14657029476432401060570498430025}&	\num{-0.056402870274454702803623717954906}&	\num{0.14925560425746900450327814269258}&	\num[round-precision=0]{0.0}&	\num[round-precision=0]{1.0}&	\num{0.25318797799999998021647229506925}&	\num{0.39742486500000001603538635208679}&	\num{0.49673074000000000394194898944988}&	\num{0.59270089999999997498747461577295}&	\num{0.73218357999999994500939237696002}\\
			$z_{8}$&\num{0.5025217206501519928707466533524}&	\num{0.15242671620093101170212435135909}&	\num{-0.055247290212409497223156051859405}&	\num{0.0038270918330835998477645887305698}&	\num[round-precision=0]{0.0}&	\num[round-precision=0]{1.0}&	\num{0.24712508749999900770077942979697}&	\num{0.40024379249999997298914422572125}&	\num{0.50608722999999999903053549132892}&	\num{0.60625983250000003987878471889417}&	\num{0.74685117999999994786719526018715}\\
			$z_{9}$&\num{0.3957697698219569959476871190418}&	\num{0.21994008254957500558823824121646}&	\num{0.83323880940651096072002701475867}&	\num{-0.29539809236904701617021373749594}&	\num[round-precision=0]{0.0}&	\num[round-precision=0]{1.0}&	\num{0.1384606359999999980470164473445}&	\num{0.23005235250000000135628397401888}&	\num{0.32397428000000000336555672220129}&	\num{0.53529714249999904573940057161963}&	\num{0.83549927000000001608981392564601}\\
			$z_{10}$&\num{0.45004021925189602315597880988207}&	\num{0.16073503844106901139454635085713}&	\num{0.34399244647800897300982114757062}&	\num{-0.051249082788190403192629673867486}&	\num[round-precision=0]{0.0}&	\num[round-precision=0]{1.0}&	\num{0.20647964850000000147645096149063}&	\num{0.3348794600000000176720504896366}&	\num{0.43721425000000002647482233442133}&	\num{0.5553383800000000203311856239452}&	\num{0.73338417699999902676921692545875}\\
			$z_{11}$&\num{0.4738105494626789759848861649516}&	\num{0.15739544907715699584294100077386}&	\num{0.075221884343309106113473205823539}&	\num{-0.094171916218853901625607250025496}&	\num[round-precision=0]{0.0}&	\num[round-precision=0]{1.0}&	\num{0.21830151899999999942636463856616}&	\num{0.36536141249999998237285581126343}&	\num{0.47173505999999998383742649821215}&	\num{0.58057833500000000093166363512864}&	\num{0.73571329999999901527729662120692}\\
			$z_{12}$&\num{0.43270457922947802575919240553048}&	\num{0.19165442777016800102174443054537}&	\num{0.40674507874703702547947159473551}&	\num{-0.35803334086705701500008558468835}&	\num{2.38649119999999e-36}&	\num[round-precision=0]{1.0}&	\num{0.15171232500000000897877328043251}&	\num{0.28750504249999997430364828687743}&	\num{0.41281026500000000956092094384076}&	\num{0.56092629999999998879900431347778}&	\num{0.78162224499999899407498560321983}\\
			$z_{13}$&\num{0.47875460459198199147934360553336}&	\num{0.15289956292205700294495329671918}&	\num{0.12280680914790699509442362113987}&	\num{-0.14660971782884599234897393671417}&	\num{7.38224e-22} & \num[round-precision=0]{1.0}&	\num{0.2347282369999990070397188901552}&	\num{0.37164021250000001117541614803486}&	\num{0.47420434499999897237643153857789}&	\num{0.58212991250000001297593144045095}&	\num{0.7376164399999990406442407220311}\\
			$z_{14}$&\num{0.48108439087995197303015970646811}&	\num{0.16592951219894300218804517044191}&	\num{0.13591346183899399857608614183846}&	\num{-0.28203423585331899303696445713285}&	\num[round-precision=0]{0.0}&	\num[round-precision=0]{1.0}&	\num{0.2149128040000000128273427435488}&	\num{0.36334900999999997228684378569596}&	\num{0.47667106999999997452022171273711}&	\num{0.59386567999999995137727637484204}&	\num{0.76474757799999903973287018743576}\\
			$z_{15}$&\num{0.61751902549990800306289884247235}&	\num{0.16422133176865200132610311811732}&	\num{-0.8983303774306600208987561018148}&	\num{0.46727243838010001919514024848468}&	\num[round-precision=0]{0.0}&	\num[round-precision=0]{1.0}&	\num{0.29262764000000002218015993094014}&	\num{0.52489649000000004885890803052462}&	\num{0.65912774999999901304192917450564}&	\num{0.73403599249999995635107552516274}&	\num{0.82165986999999995887833392771427}\\
			$z_{16}$&\num{0.42804997678293099649948771912022}&	\num{0.1605016594930120110706894820396}&	\num{0.4999040438768649896594808978989}&	\num{0.17476580864575300133800794810668}&	\num[round-precision=0]{0.0}&	\num[round-precision=0]{1.0}&	\num{0.18955781799999998904482367834134}&	\num{0.31271312499999998069100115571928}&	\num{0.41366296999999901951028391522414}&	\num{0.52840188749999994488604215803207}&	\num{0.71575301699999904769811109872535}\\
			$z_{17}$&\num{0.3987101479013169802456673096458}&	\num{0.16243450897199399052261981069023}&	\num{0.95540487221567804709820848074742}&	\num{0.96968799694445495518380084831733}&	\num[round-precision=0]{0.0}&	\num[round-precision=0]{1.0}&	\num{0.19242433899999999979968379193451}&	\num{0.28079596000000001110308289753448}&	\num{0.36640267500000001099635937862331}&	\num{0.48795141250000001464925958316599}&	\num{0.71295349349999903587615790456766}\\
			$z_{18}$&\num{0.54023826390634499539089574682293}&	\num{0.16241724812996899407480100308021}&	\num{-0.23512252487459500693844915986119}&	\num{-0.14704782638330099464596401048766}&	\num[round-precision=0]{0.0}&	\num[round-precision=0]{1.0}&	\num{0.26187253449999997600983192569402}&	\num{0.43184286249999997986392941129452}&	\num{0.54654639999999998778434928681236}&	\num{0.6565980750000000032429170460091}&	\num{0.79539305999999998508798171314993}\\
			$z_{19}$&\num{0.39067114075234599113173317164183}&	\num{0.29975076620611501576618707076705}&	\num{0.6200070015180400018195427946921}&	\num{-1.4602845869620999508242675801739}&	\num[round-precision=0]{0.0}&	\num[round-precision=0]{1.0}&	\num{0.1096017955000000015930439190015}&	\num{0.15517379249999999069054013034474}&	\num{0.20397438000000001068201527232304}&	\num{0.77674286999999997416210817391402}&	\num{0.8476692399999989913084164072643}\\
			$z_{20}$&\num{0.47555568995074998239758201634686}&	\num{0.16270320618060399309179331339692}&	\num{0.25997636041649102578432461996272}&	\num{-0.14331864542812100538071717892308}&	\num[round-precision=0]{0.0}&	\num[round-precision=0]{1.0}&	\num{0.22471115199999999734181699295732}&	\num{0.35874118500000001796834681044857}&	\num{0.46612009500000001205322064379288}&	\num{0.58452594999999996083772657584632}&	\num{0.75620601699999900890958315358148}\\
			$u_{1}$&\num{0.5416302888928450354200094807311}&	\num{0.15347979082156901209366139937629}&	\num{-0.29616396125543797968049375413102}&	\num{-0.014261478634370600496183989491783}&	\num{1.7548352e-21}&	\num[round-precision=0]{1.0}&	\num{0.27228436000000000305476532957982}&	\num{0.44213928000000002338509830224211}&	\num{0.55110886999999997293286924104905}&	\num{0.65008674499999996587717987495125}&	\num{0.77683331149999901299452176317573}\\
			$u_{2}$&\num{0.53354396653203495226591712707886}&	\num{0.15479612595389100060394582669687}&	\num{-0.15730993842807200477906803826045}&	\num{-0.16546010827887600025398739944649}&	\num[round-precision=0]{0.0}&	\num[round-precision=0]{1.0}&	\num{0.27033662899999999496003511012532}&	\num{0.42887788999999998384993205036153}&	\num{0.53907640000000001112567815653165}&	\num{0.64296192249999994938747249761946}&	\num{0.77787400249999905099684838205576}\\
			$u_{3}$&\num{0.64702875165567397264254623223678}&	\num{0.176165853476506012453839389309}&	\num{-1.3130496310451000852026481879875}&	\num{1.4942259356184799390376838346128}&	\num[round-precision=0]{0.0}&	\num[round-precision=0]{1.0}&	\num{0.26484521750000000750446815800387}&	\num{0.57723338749999997254747086117277}&	\num{0.69676154999999995176551692566136}&	\num{0.76685917000000003440618456806988}&	\num{0.84474429650000004698995326180011}\\
			$u_{4}$&\num{0.3660526820366720168742347141233}&	\num{0.18139608542170199134702102128358}&	\num{1.3622537023041398906997301310184}&	\num{1.2091368563188000795349807958701}&	\num[round-precision=0]{0.0}&	\num[round-precision=0]{1.0}&	\num{0.1859362850000000066952310362467}&	\num{0.24000061749999901383745282146265}&	\num{0.30329105499999997608284729722072}&	\num{0.43162791249999998832720393693307}&	\num{0.7788593249999999912347448116634}\\
			$u_{5}$&\num{0.5100078515037459458980606541445}&	\num{0.166766379178065998756252952262}&	\num{-0.033484070243063303318198364877389}&	\num{-0.29206833475751697859479349972389}&	\num[round-precision=0]{0.0}&	\num[round-precision=0]{1.0}&	\num{0.23458681300000000513605868945888}&	\num{0.39361126000000001834422391766566}&	\num{0.5110950800000000349143647326855}&	\num{0.62740306249999999721467247582041}&	\num{0.78294500499999997167321907909354}\\
			$u_{6}$&\num{0.47930065777430602702580131335708}&	\num{0.1632342244330400038432316023318}&	\num{0.14453498514930399054811971382151}&	\num{-0.1041447003404910032209329528996}&	\num[round-precision=0]{0.0}&	\num[round-precision=0]{1.0}&	\num{0.22099607299999998732964456849004}&	\num{0.36634233249999997861223732797953}&	\num{0.47168179999999998441850834751676}&	\num{0.58933703500000000907732555788243}&	\num{0.75707182699999997499418213919853}\\
			$u_{7}$&\num{0.53373614055177898141124614994624}&	\num{0.1588095296565009972855619935217}&	\num{-0.14956586603835800297623848109652}&	\num{-0.046468873680104202938956348134525}&	\num[round-precision=0]{0.0}&	\num[round-precision=0]{1.0}&	\num{0.26342732949999997371648419175472}&	\num{0.42807118499999902150321418048406}&	\num{0.53898960500000003825959993264405}&	\num{0.644229902500000006781988304283}&	\num{0.78523108999999902124500295030884}\\
			$u_{8}$&\num{0.52755795749514300396754151734058}&	\num{0.15704480428721198648212009629788}&	\num{-0.2788050279363429861945178345195}&	\num{-0.021159269219764999248045711510713}&	\num[round-precision=0]{0.0}&	\num[round-precision=0]{1.0}&	\num{0.25164249699999902132319107295189}&	\num{0.42668361750000000087723606156942}&	\num{0.53567776499999997241729943198152}&	\num{0.63805237749999998975170001358492}&	\num{0.77083172999999904906331948950537}\\
			$u_{9}$&\num{0.47666843037424200257845541273127}&	\num{0.17216006908560899102411667627166}&	\num{-0.052469123850884299975483315847669}&	\num{-0.21978703637430199724356327806163}&	\num[round-precision=0]{0.0}&	\num[round-precision=0]{1.0}&	\num{0.18185577150000001300789165270544}&	\num{0.35899995750000002203705662395805}&	\num{0.48187050500000000452516246696177}&	\num{0.59719080000000002161897327823681}&	\num{0.74909239499999902278659646981396}\\
			$u_{10}$&\num{0.47691393844119500666778321829042}&	\num{0.15499893133164699410286857528263}&	\num{0.18330661948177701114914839308767}&	\num{-0.14905969851978501106160024392011}&	\num[round-precision=0]{0.0}&	\num[round-precision=0]{1.0}&	\num{0.23111855800000000171756653344346}&	\num{0.36665408499999901925647805001063}&	\num{0.47186651499999998637946418966749}&	\num{0.58073036249999998581472482328536}&	\num{0.74140070399999902228671544435201}\\
			\bottomrule
	\end{tabular}}
\end{table}
\begin{table}
	\vspace{-0.2cm}
	\caption{\small Comparison of the ATTE estimations per \eqref{eqt:weighted average of relative error ATTE} with strong \textbf{causalities} in the simulated data. Number of observations $N=40000$, Number of experiments $M=100$.}
	\label{table:light tail and heavy tail with alpha is 0.05 ATTE}
	\resizebox{\linewidth}{!}{%
		\begin{tabular}{cccccccccccccccccc}
			\toprule
			& \multicolumn{6}{c}{Light tail} &\multicolumn{6}{c}{Heavy tail}\\
			& \multicolumn{3}{c}{$\beta=5\%$}& \multicolumn{3}{c}{$\beta=95\%$}& \multicolumn{3}{c}{$\beta=5\%$}& \multicolumn{3}{c}{$\beta=95\%$}\\
			\cmidrule(r){2-13}
			\multirow{3}{*}{} & \multicolumn{2}{c}{Weighted avg.} & \multicolumn{1}{c}{Mean Err.} & \multicolumn{2}{c}{Weighted avg.} & \multicolumn{1}{c}{Mean Err.} & \multicolumn{2}{c}{Weighted avg.} & \multicolumn{1}{c}{Mean Err.} & \multicolumn{2}{c}{Weighted avg.} & \multicolumn{1}{c}{Mean Err.}\\
			& \multicolumn{2}{c}{ATTE} & \multicolumn{1}{c}{reduction} & \multicolumn{2}{c}{ATTE} & \multicolumn{1}{c}{reduction} & \multicolumn{2}{c}{ATTE} & \multicolumn{1}{c}{reduction} & \multicolumn{2}{c}{ATTE} & \multicolumn{1}{c}{reduction}\\
			\midrule
			\multirow{1}{*}{Regressor} & a) \textbf{IoC} & b) \textbf{IwC} & $|\text{IwC/IoC-1}|$ & a) \textbf{IoC} & b) \textbf{IwC} & $|\text{IwC/IoC-1}|$ & a) \textbf{IoC} & b) \textbf{IwC} & $|\text{IwC/IoC-1}|$ & a) \textbf{IoC} & b) \textbf{IwC} & $|\text{IwC/IoC-1}|$ \\
			\midrule
			OLS&\percentage{0.42674529192685201062218425249739}&	\percentage{0.28517766200313598989879437795025}&	\percentage{0.33173800063383301894859300773533}&	\percentage{0.40521731595072002507507136215281}&	\percentage{0.36990566188331802655042679361941}&	\percentage{0.087142510148052299001619758200832}&	\percentage{0.80920386674208599764313021296402}&	\percentage{0.72051789311527303372173491879948}&	\percentage{0.10959657667463799934903789790042}&	\percentage{0.87980580176382594537187742389506}&	\percentage{0.89685311105929998998220753492205}&	\percentage{-0.019376218321472099548330447760236}\\
			LASSO&\percentage{0.42674409987771200203354737823247}&	\percentage{0.28012748383329799661822789857979}&	\percentage{0.34357034130390601278648432526097}&	\percentage{0.40519013667116399179946029107668}&	\percentage{0.28786279034706802049115026420623}&	\percentage{0.28956121017159297537091333651915}&	\percentage{0.78534617001380002676569347386248}&	\percentage{0.64740568422277000593112461501732}&	\percentage{0.17564290889532899186953329717653}&	\percentage{0.87976566684164403397971909726039}&	\percentage{0.76934524400909698460537811115501}&	\percentage{0.12551117529848099252376414369792}\\
			RIDGE&\percentage{0.42661013232655597970932603857364}&	\percentage{0.28374824160602302169920108099177}&	\percentage{0.33487692835944299352135544722842}&	\percentage{0.40498285802560102242253492477175}&	\percentage{0.36293649293271601496968514766195}&	\percentage{0.10382257979478799747585782142778}&	\percentage{0.80614569101689903529717184937908}&	\percentage{0.71716563499051000363238017598633}&	\percentage{0.11037714028359499329745574414119}&	\percentage{0.87618722686951400202559625540744}&	\percentage{0.89037095375468200408874963613926}&	\percentage{-0.016188009195072201179543824878238}\\
			RF&\percentage{0.68345792776398195389475631600362}&	\percentage{0.11773571967467100218840414527222}&	\percentage{0.82773523447176999834340449524461}&	\percentage{0.68553522594745597285026406098041}&	\percentage{0.12946029509341000496291940180527}&	\percentage{0.81115442329825204748061651116586}&	\percentage{0.68794518025506701253135588558507}&	\percentage{0.12803817678082698994224131183728}&	\percentage{0.8138831691017079927874533495924}&	\percentage{0.69115299890936798821172715179273}&	\percentage{0.15294296268404600214196875640482}&	\percentage{0.77871330526614301881238588975975}\\
			XGB&\percentage{0.8759918682640980280851294992317}&	\percentage{0.12492188354303500330555465325233}&	\percentage{0.8573937863252240054734443219786}&	\percentage{0.87713875477585001583236135047628}&	\percentage{0.15585892461826500454691313279909}&	\percentage{0.82230984120853900520842216792516}&	\percentage{0.87135296222672597998837318300502}&	\percentage{0.18335608369034400721986344251491}&	\percentage{0.78957312175563998390970255059074}&	\percentage{0.86592616631998997611674440122442}&	\percentage{0.23746065435509500218813627725467}&	\percentage{0.72577263098047295475367945982725}\\
			GRU&\percentage{0.26619181769250999103348931384971}&	\percentage{0.029582030962104599591810938363778}&	\percentage{0.88886949561959704979585694673005}&	\percentage{0.54437728318369105195984047895763}&	\percentage{0.18733994097087799945988706440403}&	\percentage{0.65586377911426496645219685888151}&	\percentage{0.26398674468970700912606730526022}&	\percentage{0.065243362434137905370157284323795}&	\percentage{0.75285364228864704028154619663837}&	\percentage{0.50634208881740294661000234555104}&	\percentage{0.27008188576188102425490455971158}&	\percentage{0.46660194416648897997035305706959}\\
			CNN&\percentage{0.20641751107793798736267376625619}&	\percentage{0.037910250312056699428797656992174}&	\percentage{0.81634188827253695297514468620648}&	\percentage{0.46731069644453898659364199374977}&	\percentage{0.17907686276276399417461959728826}&	\percentage{0.61679271601261598245713457799866}&	\percentage{0.21344352957694098793517412104848}&	\percentage{0.046880572881321899703088718069921}&	\percentage{0.78036076814208299445851935161045}&	\percentage{0.47780552538747900204185725669959}&	\percentage{0.22172333628744700373403020421392}&	\percentage{0.53595485086188798806716704348219}\\
			MLP&\percentage{0.20197828037252299471226990590367}&	\percentage{0.032135918035045001051042135031821}&	\percentage{0.84089418933672199774065347810392}&	\percentage{0.47707919530695702681200032202469}&	\percentage{0.17507425151130501239649106537399}&	\percentage{0.63302895361290900666517700301483}&	\percentage{0.18709951751898298843990176010266}&	\percentage{0.040577208859634297799523494632012}&	\percentage{0.78312499466751794852115153844352}&	\percentage{0.43306140512329699410187799912819}&	\percentage{0.23960107021189300247066000792984}&	\percentage{0.44672725997442297440898073546123}\\			
			\bottomrule
	\end{tabular}}
\end{table}
\vspace{-0.1em}
\begin{table}
	\caption{\small Comparison of the ATTE estimations per \eqref{eqt:weighted average of relative error ATTE} with strong \textbf{nonlinearities} in the simulated data. Number of observations $N=40000$, Number of experiments $M=100$.}
	\label{table:light tail and heavy tail with beta is 0.05 ATTE}
	\resizebox{\linewidth}{!}{%
		\begin{tabular}{cccccccccccccccccc}
			\toprule
			& \multicolumn{6}{c}{Light tail} &\multicolumn{6}{c}{Heavy tail}\\
			& \multicolumn{3}{c}{$\alpha=5\%$}& \multicolumn{3}{c}{$\alpha=95\%$}& \multicolumn{3}{c}{$\alpha=5\%$}& \multicolumn{3}{c}{$\alpha=95\%$}\\
			\cmidrule(r){2-13}
			\multirow{3}{*}{} & \multicolumn{2}{c}{Weighted avg.} & \multicolumn{1}{c}{Mean Err.} & \multicolumn{2}{c}{Weighted avg.} & \multicolumn{1}{c}{Mean Err.} & \multicolumn{2}{c}{Weighted avg.} & \multicolumn{1}{c}{Mean Err.} & \multicolumn{2}{c}{Weighted avg.} & \multicolumn{1}{c}{Mean Err.}\\
			& \multicolumn{2}{c}{ATTE} & \multicolumn{1}{c}{reduction} & \multicolumn{2}{c}{ATTE} & \multicolumn{1}{c}{reduction} & \multicolumn{2}{c}{ATTE} & \multicolumn{1}{c}{reduction} & \multicolumn{2}{c}{ATTE} & \multicolumn{1}{c}{reduction}\\
			\midrule
			\multirow{1}{*}{Regressor} & a) \textbf{IoC} & b) \textbf{IwC} & $|\text{IwC/IoC-1}|$ & a) \textbf{IoC} & b) \textbf{IwC} & $|\text{IwC/IoC-1}|$ & a) \textbf{IoC} & b) \textbf{IwC} & $|\text{IwC/IoC-1}|$ & a) \textbf{IoC} & b) \textbf{IwC} & $|\text{IwC/IoC-1}|$ \\
			\midrule
			OLS&\percentage{0.42674529199999999828918362254626}&	\percentage{0.28517766199999999843228692952835}&	\percentage{0.33173800100000000457001192444295}&	\percentage{0.43142600100000000340472183779639}&	\percentage{0.36046136400000000632459773441951}&	\percentage{0.1644885490000000116683764872505}&	\percentage{0.80920386700000002111465846610372}&	\percentage{0.7205178930000000203648369279108}&	\percentage{0.10959657700000000057904969708034}&	\percentage{0.86775510099999997315478594828164}&	\percentage{0.86295715200000000510982545165461}&	\percentage{0.0055291500000000000369926311805102}\\
			LASSO&\percentage{0.42674410000000001508269065197965}&	\percentage{0.28012748399999998216713947840617}&	\percentage{0.34357034100000000131913679979334}&	\percentage{0.43140165699999999393554617199698}&	\percentage{0.26493598499999998496789999080647}&	\percentage{0.38587165699999997903191228942887}&	\percentage{0.78534616999999995456960277806502}&	\percentage{0.64740568399999998128180322964909}&	\percentage{0.17564290900000001371950020256918}&	\percentage{0.86771023499999999639697989550768}&	\percentage{0.72120191300000002776471319521079}&	\percentage{0.16884475500000001302325358665257}\\
			RIDGE&\percentage{0.42661013199999997524969330697786}&	\percentage{0.28374824199999998386445554388047}&	\percentage{0.33487692800000001813742755985004}&	\percentage{0.43118683499999999053642341095838}&	\percentage{0.35343982099999998736095108142763}&	\percentage{0.18030934100000001185648557111563}&	\percentage{0.80614569099999999757244495413033}&	\percentage{0.71716563499999996800227108906256}&	\percentage{0.11037713999999999858747656844571}&	\percentage{0.8641886060000000258085606219538}&	\percentage{0.85631285300000004045983814648935}&	\percentage{0.0091134660000000006219167403287429}\\
			RF&\percentage{0.68345792800000004785232476933743}&	\percentage{0.11773572000000000203900896167397}&	\percentage{0.82773523400000004190957270111539}&	\percentage{0.67645121100000005220920229476178}&	\percentage{0.13279306399999998844485560312023}&	\percentage{0.80369158699999998507479403997422}&	\percentage{0.68794518000000004498417638387764}&	\percentage{0.1280381770000000030140796525302}&	\percentage{0.81388316900000001741233290886157}&	\percentage{0.67799109199999996189234252597089}&	\percentage{0.14662494900000000441764314018656}&	\percentage{0.78373617200000000870829808263807}\\
			XGB&\percentage{0.87599186799999995134413666164619}&	\percentage{0.12492188399999999692813190677043}&	\percentage{0.85739378600000004659875685320003}&	\percentage{0.87653581300000005249728474154836}&	\percentage{0.15594535500000000771692043599614}&	\percentage{0.82208900900000003719725327755441}&	\percentage{0.87135296200000000865770743985195}&	\percentage{0.18335608400000000250962273184996}&	\percentage{0.78957312199999996060739704262232}&	\percentage{0.86394394900000004433593403518898}&	\percentage{0.24816846400000000527619192780548}&	\percentage{0.71274934599999995032959532181849}\\
			GRU&\percentage{0.26619181800000002446893176966114}&	\percentage{0.029582031000000001685901196424311}&	\percentage{0.88886949599999998117283439569292}&	\percentage{0.5748850540000000064466689764231}&	\percentage{0.20740460399999999241416048789688}&	\percentage{0.63922422000000000963382262852974}&	\percentage{0.26398674500000002263178089378926}&	\percentage{0.065243361999999999190080757216492}&	\percentage{0.75285364200000004508694928517798}&	\percentage{0.54997794600000005083018095319858}&	\percentage{0.24785633800000000936947230911755}&	\percentage{0.54933404200000002237658236481366}\\
			CNN&\percentage{0.20641751099999999796708038957149}&	\percentage{0.03791024999999999950395235259748}&	\percentage{0.81634188799999996000167357124155}&	\percentage{0.50095247399999998094699549255893}&	\percentage{0.18998984699999998948527490938432}&	\percentage{0.62074277200000005372260147851193}&	\percentage{0.21344352999999999243030401885335}&	\percentage{0.046880573000000001770981583604225}&	\percentage{0.78036076799999998243606569303665}&	\percentage{0.48066603800000001767145363373857}&	\percentage{0.22371931000000000455685267297667}&	\percentage{0.53456393400000001836502860896871}\\
			MLP&\percentage{0.20197828000000001003044758363103}&	\percentage{0.032135917999999999428162311687629}&	\percentage{0.84089418900000001411143557561445}&	\percentage{0.51581325200000005537503966479562}&	\percentage{0.18760613000000000982758763257152}&	\percentage{0.63629059700000001331687826677808}&	\percentage{0.18709951799999999244583648305706}&	\percentage{0.040577209000000002969699153254624}&	\percentage{0.78312499499999999041932596810511}&	\percentage{0.47911406899999997577310750784818}&	\percentage{0.24977098999999999828780516963889}&	\percentage{0.47868157999999999541174133810273}\\
			\bottomrule
	\end{tabular}}
\end{table}

\begin{table}
	\caption{\small Comparison of the ATTE estimations per \eqref{eqt:weighted average of relative error ATTE} under different \textbf{nonlinearities} and \textbf{causalities} in the semi-synthetic data. Number of observations $N=80000$, Number of experiments $M=100$.}
	\label{table:ATTE-alpha_beta real}
	\resizebox{\linewidth}{!}{%
		\begin{tabular}{cccccccccccccccccc}
			\toprule			
			& \multicolumn{6}{c}{$\alpha=5\%$} &\multicolumn{6}{c}{$\beta=5\%$}\\
			& \multicolumn{3}{c}{$\beta=5\%$}& \multicolumn{3}{c}{$\beta=50\%$}& \multicolumn{3}{c}{$\alpha=5\%$}& \multicolumn{3}{c}{$\alpha=50\%$}\\
			\cmidrule(r){2-13}
			\multirow{3}{*}{} & \multicolumn{2}{c}{Weighted avg.} & \multicolumn{1}{c}{Mean Error} & \multicolumn{2}{c}{Weighted avg.} & \multicolumn{1}{c}{Mean Error} & \multicolumn{2}{c}{Weighted avg.} & \multicolumn{1}{c}{Mean Error} & \multicolumn{2}{c}{Weighted avg.} & \multicolumn{1}{c}{Mean Error}\\
			& \multicolumn{2}{c}{ATTE} & \multicolumn{1}{c}{reduction} & \multicolumn{2}{c}{ATTE} & \multicolumn{1}{c}{reduction} & \multicolumn{2}{c}{ATTE} & \multicolumn{1}{c}{reduction} & \multicolumn{2}{c}{ATTE} & \multicolumn{1}{c}{reduction}\\
			\midrule
			\multirow{1}{*}{Regressor} & a) \textbf{IoC} & b) \textbf{IwC} & $|\text{IwC/IoC-1}|$ & a) \textbf{IoC} & b) \textbf{IwC} & $|\text{IwC/IoC-1}|$ & a) \textbf{IoC} & b) \textbf{IwC} & $|\text{IwC/IoC-1}|$ & a) \textbf{IoC} & b) \textbf{IwC} & $|\text{IwC/IoC-1}|$ \\
			\midrule
			OLS&\percentage{0.13620436745169600678551091732515}&	\percentage{0.0481091021517250985728431089683}&	\percentage{0.64678737509069594846522477382678}&	\percentage{0.13974563420029201266103768830362}&	\percentage{0.063989097907898598305287407583819}&	\percentage{0.5421030626532099550018983791233}&	\percentage{0.14307271938313098624817598647496}&	\percentage{0.11322826865307600296972623254987}&	\percentage{0.20859637573627801132580827925267}&	\percentage{0.14872100773236499038354452295607}&	\percentage{0.1420532638278750037219566593194}&	\percentage{0.044833907503429602792177632863968}\\
			LASSO&\percentage{0.13620433696970599135589452544082}&	\percentage{0.047295645366292597688850918302705}&	\percentage{0.65275962264834597181817343880539}&	\percentage{0.13974557163758599687142236689397}&	\percentage{0.061385383136943298676424518589556}&	\percentage{0.5607346807658489806769352981064}&	\percentage{0.14307270316422499023545356067189}&	\percentage{0.11206466934141999514107368440818}&	\percentage{0.21672920925533001312857095399522}&	\percentage{0.14872088083507201128341534968058}&	\percentage{0.13735870326836099875045249518735}&	\percentage{0.076399342869083702867172291917086}\\
			RIDGE&\percentage{0.13625649277802701075934521668387}&	\percentage{0.04777657028565240188155982536955}&	\percentage{0.64936298218474997057114705967251}&	\percentage{0.13979891524026200055885738038342}&	\percentage{0.063564481952010198906322102629929}&	\percentage{0.54531491290353495404730210793787}&	\percentage{0.14310061317859801288676635522279}&	\percentage{0.11256949749142200478679853858921}&	\percentage{0.21335419191440599395725996600959}&	\percentage{0.14882738875424400548652670295269}&	\percentage{0.14124388756982500514602918428864}&	\percentage{0.0509550106865139032508160710222}\\
			RF&\percentage{0.73511112016960600268333791973419}&	\percentage{0.029506367768078300972689476111555}&	\percentage{0.95986135026596897112938222562661}&	\percentage{0.73512538865047094560623008874245}&	\percentage{0.039700388159261697229585053037226}&	\percentage{0.94599507951678396544537008594489}&	\percentage{0.73978570936865695184536662054597}&	\percentage{0.054136766839205097134879451914458}&	\percentage{0.92682101566221597455097480633412}&	\percentage{0.73924689402427601336853513203096}&	\percentage{0.066075897766637706198800117363135}&	\percentage{0.91061728050430201797382778750034}\\
			XGB&\percentage{0.88246385208120003973419898102293}&	\percentage{0.058625680937386202840055915430639}&	\percentage{0.93356591230436902772993335020146}&	\percentage{0.88429452898837301866308280295925}&	\percentage{0.075342931143127706006268340388488}&	\percentage{0.91479882700470904755007950370782}&	\percentage{0.88072530008878502361113760343869}&	\percentage{0.093936758100667802628969127454184}&	\percentage{0.89334159233168597413765610326664}&	\percentage{0.88121154813217295043159538181499}&	\percentage{0.12082815223163000106776365782935}&	\percentage{0.86288405719632299550880816241261}\\
			GRU&\percentage{0.040660479526187502663514550249602}&	\percentage{0.024079644116821400517958196019208}&	\percentage{0.40778750281799097665569320270151}&	\percentage{0.063793216871499094922093320292333}&	\percentage{0.029842160300511198778306010126471}&	\percentage{0.53220480540081005482733189637656}&	\percentage{0.046313970056905996719454066123944}&	\percentage{0.03077032459353110094291139375855}&	\percentage{0.33561462004393799452017788098601}&	\percentage{0.071561866324126893745827260318038}&	\percentage{0.039189139713903399508954095153968}&	\percentage{0.45237398454082899945660756202415}\\
			CNN&\percentage{0.032310557834357801765268192184521}&	\percentage{0.021995559053785900677224773858143}&	\percentage{0.31924545634440398833575613934954}&	\percentage{0.047935133163496401742520447442075}&	\percentage{0.031026474344767098617614564659561}&	\percentage{0.3527404161172860153072861066903}&	\percentage{0.035820512932048798049322613223922}&	\percentage{0.030023586787429900474277744137908}&	\percentage{0.16183258334730099559095606309711}&	\percentage{0.052236857163566401129894956056887}&	\percentage{0.042724677307197901632296321849935}&	\percentage{0.18209709337189800004708217784355}\\
			MLP&\percentage{0.031881653343824201130018991534598}&	\percentage{0.020409235825114401680080433720832}&	\percentage{0.35984386992063299715027824277058}&	\percentage{0.037150202455738398632512087260693}&	\percentage{0.027669049428001701168566839328378}&	\percentage{0.25521134209249102253735941303603}&	\percentage{0.035328213488848701773203941911561}&	\percentage{0.027025215514531998606262419571067}&	\percentage{0.23502456406230198826534660838661}&	\percentage{0.041555560683160103208066971092194}&	\percentage{0.037051645474192797247159347762135}&	\percentage{0.10838297293850200631037949960955}\\
			\bottomrule
	\end{tabular}}
\end{table}

\begin{figure}[h]
	\begin{subfigure}{.5\textwidth}
		\centering
		\includegraphics[width=1\textwidth]{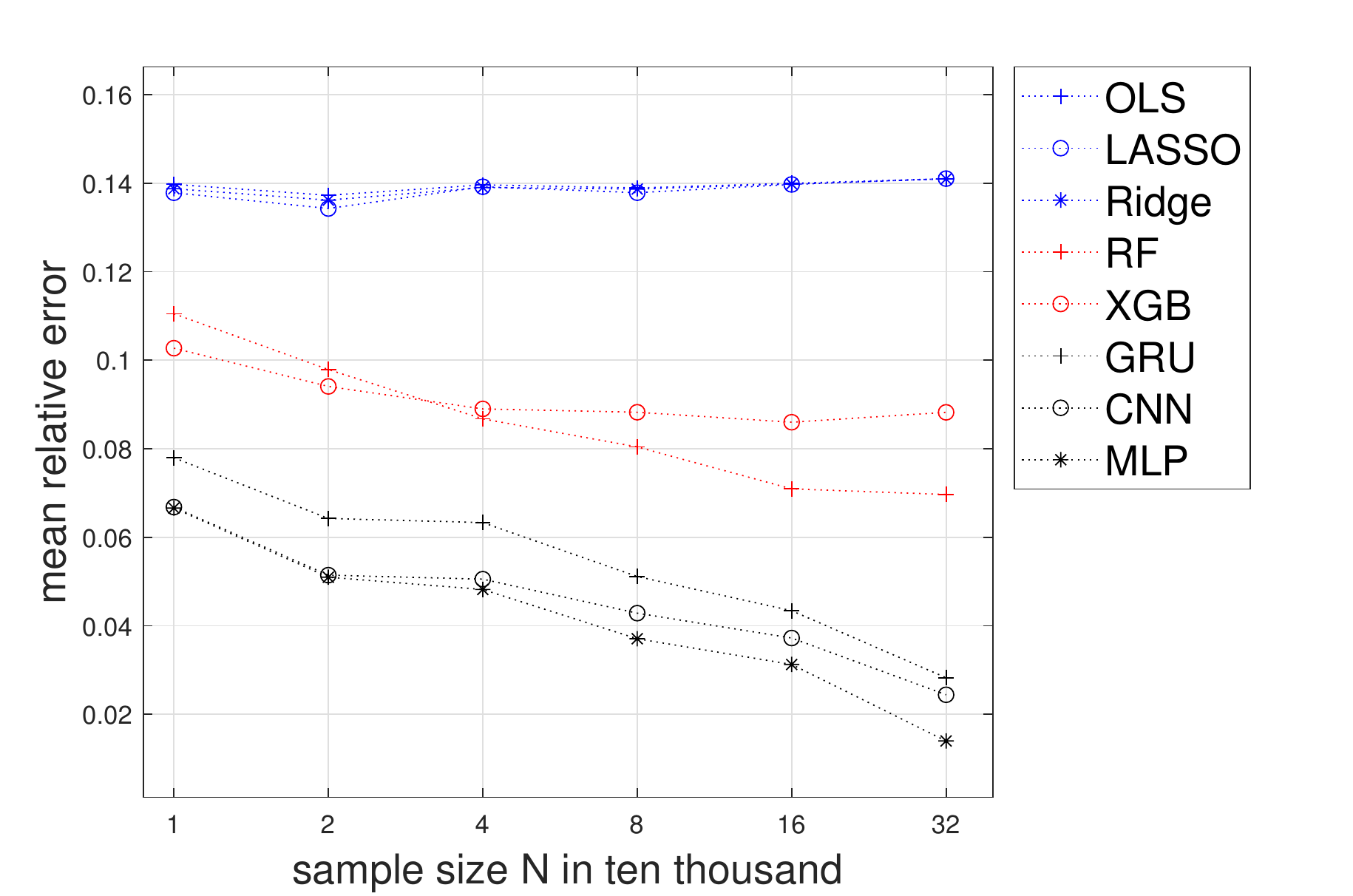}
		\caption{The mean relative error computed by \eqref{eqt:consistency_theta_ij_mean} vs. number of observations.}\label{fig:consistency conditional expectation graph simulated}
	\end{subfigure}
	\qquad
	\begin{subfigure}{.5\textwidth}
		\centering
		\includegraphics[width=1\textwidth]{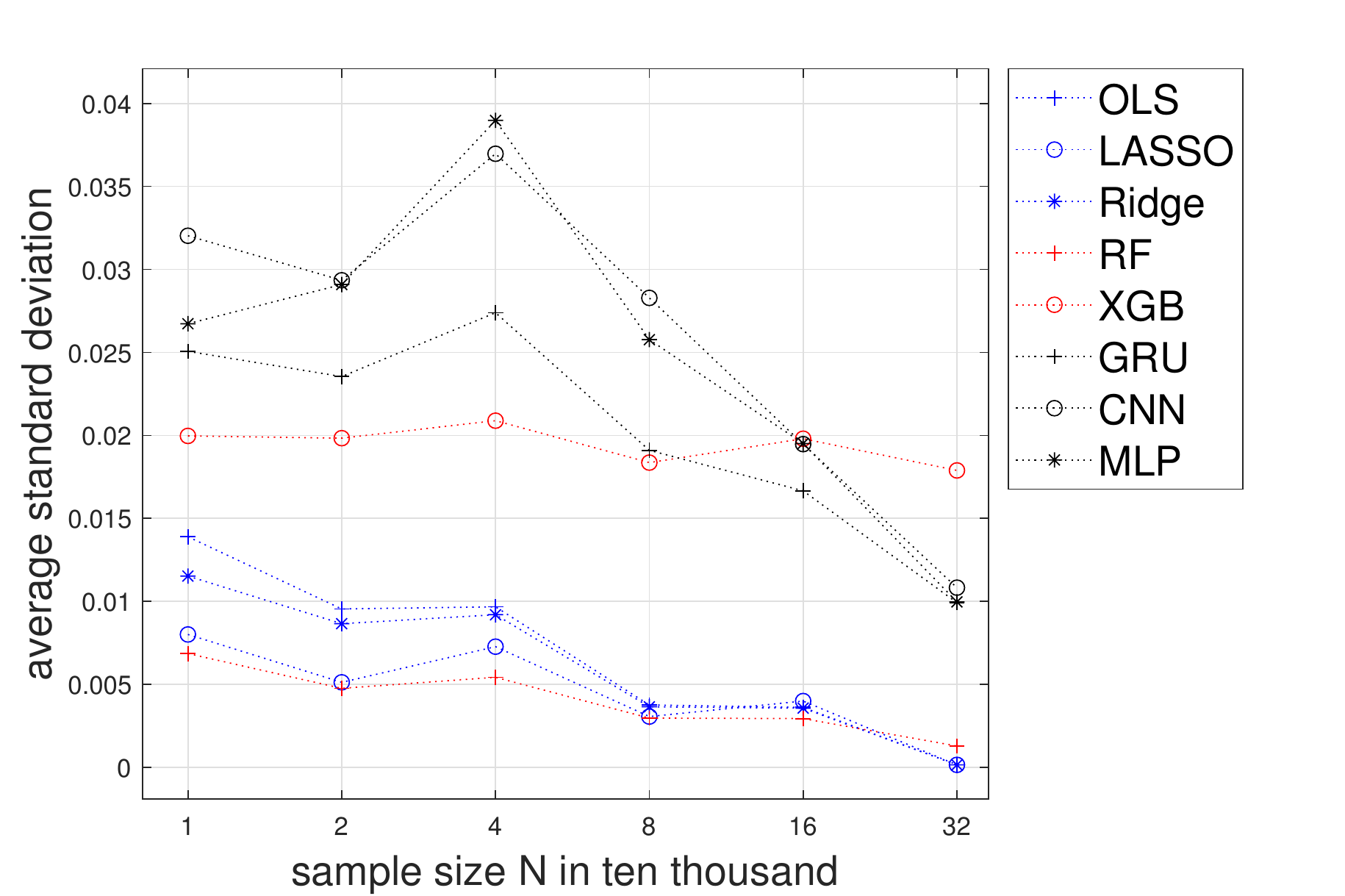}
		\caption{The average standard deviation computed by \eqref{eqt:consistency_theta_ij_std} vs. number of observations.}\label{fig:consistency conditional expectation_std graph simulated}
	\end{subfigure}
	\caption{Consistency analysis of $\hat{\theta}_{w}^{i\mid j}$ with $100$ experiments vs. number of observations $N$. $\alpha=0.05$ and $\beta=0.05$ in \eqref{eqt:simulated data model 3} of the main paper; $N\in\{10000, 20000, 40000, 80000, 160000, 320000\}$.}
	\label{fig:consistency graph conditional expectation simulated}	
\end{figure}

\begin{figure}
	\centering
	\includegraphics[width=1\textwidth]{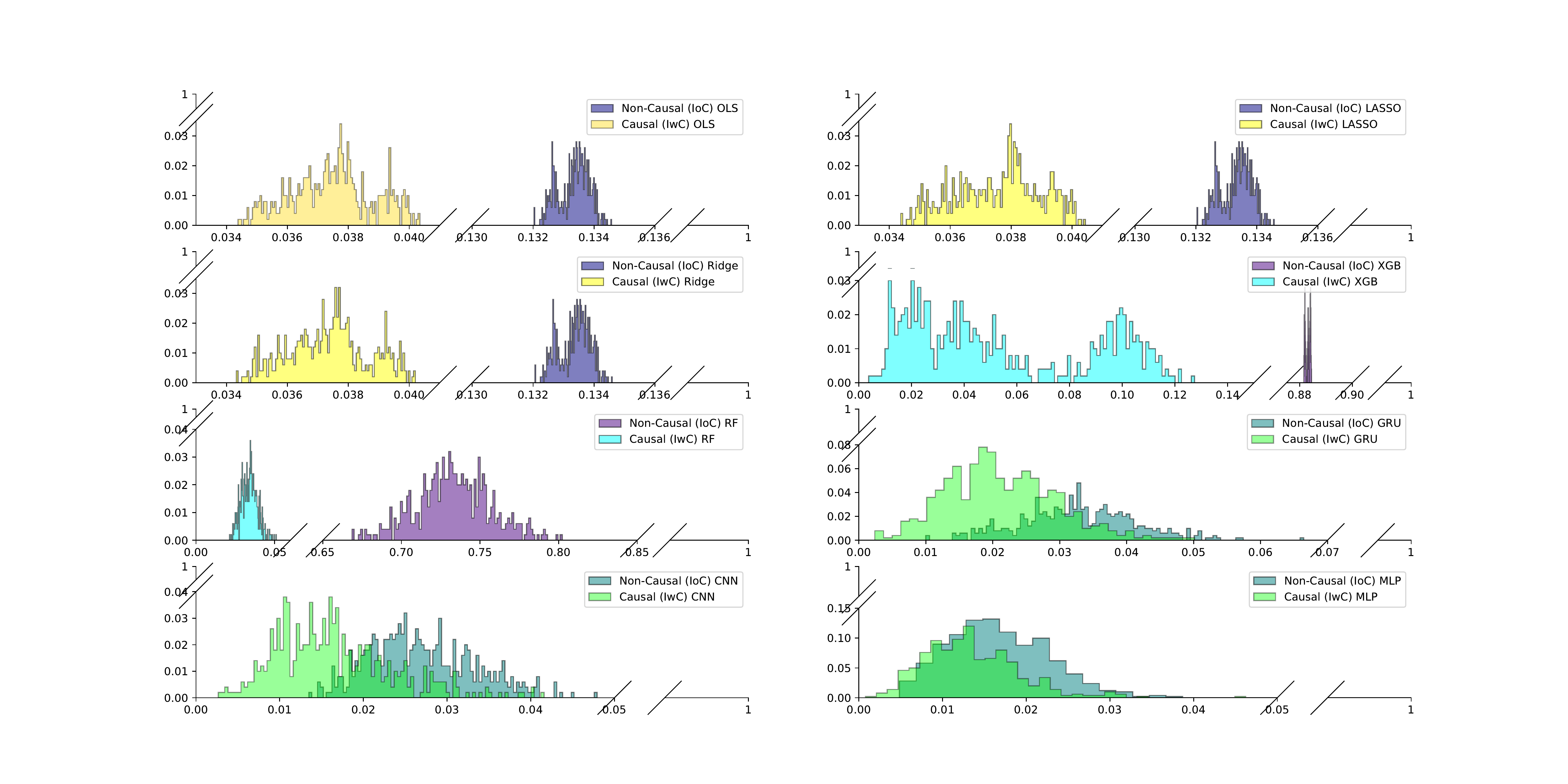}
	\caption{Frequency histogram of weighted average of relative error of estimated ATE for eight different models: ordinary least square(OLS), LASSO, RIDGE, random forest (RF), GRU, CNN, and multilayer perception (MLP) based on the semi-synthetic data with $N=160000$, $M=500$, and $\alpha=0.05$ and $\beta=0.05$ in \eqref{eqt:simulated data model 3} of the main paper.}\label{fig:ATE_3_model_histogram real}
\end{figure}

\begin{figure}
	\centering
	\includegraphics[width=1\textwidth]{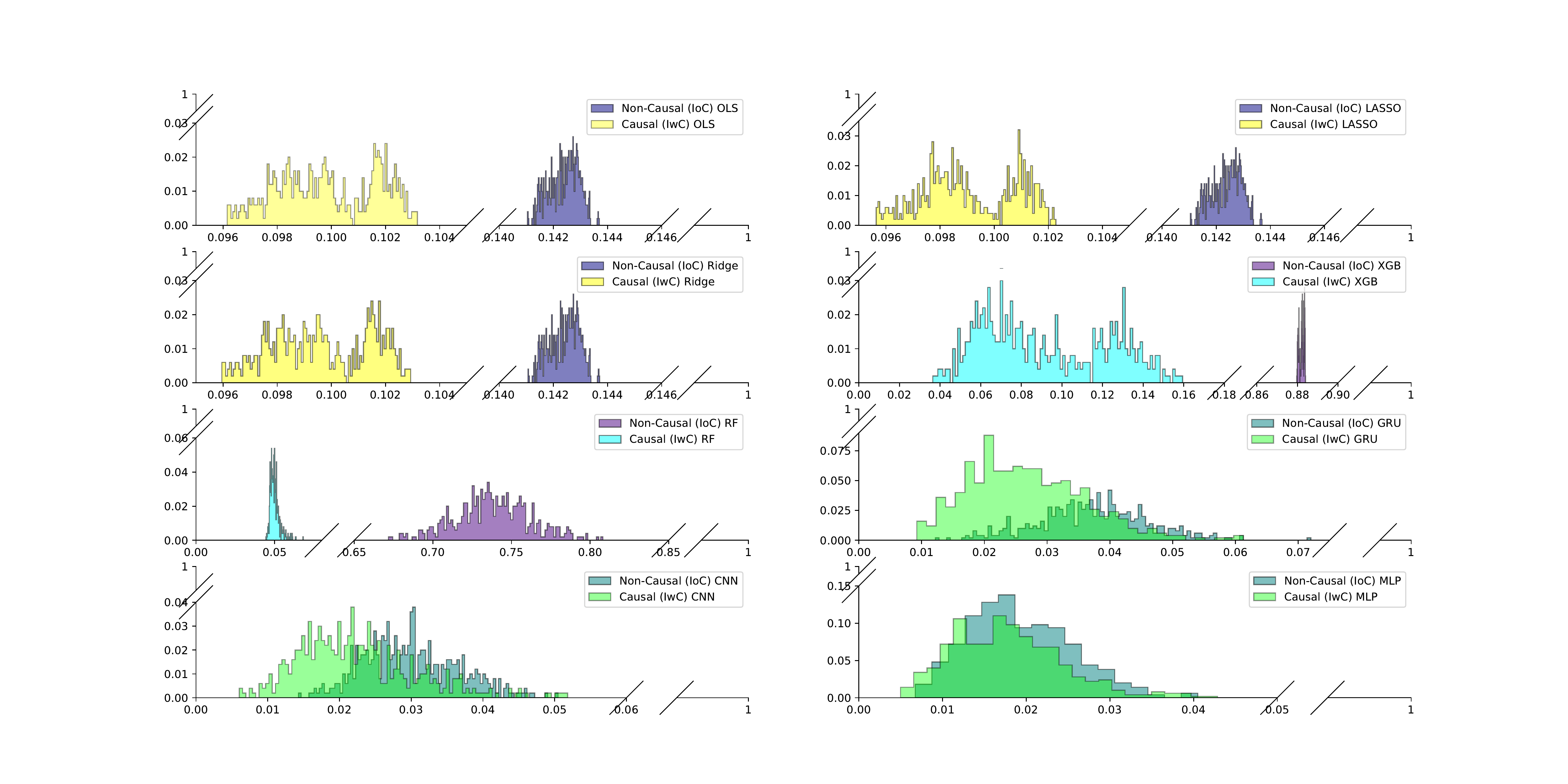}
	\caption{Frequency histogram of weighted average of relative error of estimated ATTE for eight different models: ordinary least square(OLS), LASSO, RIDGE, random forest (RF), GRU, CNN, and multilayer perception (MLP) based on the semi-synthetic data with $N=160000$, $M=500$, and $\alpha=0.05$ and $\beta=0.05$ in \eqref{eqt:simulated data model 3} of the main paper.}\label{fig:ATTE_3_model_histogram real}
\end{figure}

\clearpage
\subsection{Proofs of Theorem}\label{sec:Appendix-Proofs of Theorems and Corollary}
In this section, we provide detailed proofs of the Theorem stated in the main paper. Our goal is to construct score functions which satisfy the \textit{moment condition} and the \textit{orthogonal condition} concurrently such that we can recover the desired estimators. The two conditions are stated as follows:

\begin{definition}[Moment Condition]\label{def:moment condition}
	Let $W$ be the random elements, $\Theta$ be a convex set which contains the causal parameter $\vartheta$ of dimension $d_{\vartheta}$ ($\theta$ is the true causal parameter we are interested in) and $T$ be a convex set which contains nuisance parameter $\varrho$ ($\rho$ is the true nuisance parameter we are interested in). We say that a score function $\psi=[\psi_{1},\cdots,\psi_{d_{\vartheta}}]^{T}$ satisfies moment condition if
	\begin{equation}
	\begin{aligned}\label{eqt:moment condition}
	\mathbb{E}[\psi_{j}(W,\theta,\rho)]=0,\quad j=1,\cdots,d_{\vartheta}.
	\end{aligned}
	\end{equation}
\end{definition}

\begin{definition}[Orthogonal Condition]\label{def:Neyman orthogonal score restated in appendix}
	Using the same notations defined in \eqref{def:moment condition}. Moreover, we define the Gateaux derivative map $D_{r,j}[\varrho-\rho]:=\partial_{r}\left\{\mathbb{E}[\psi_{j}(W,\theta,\rho+r(\varrho-\rho))]\right\}$. We say that a score function $\psi$ satisfies (Neyman) orthogonal condition if for all $r\in[0,1)$, $\varrho\in\mathcal{T}\subset T$ and $j=1,\cdots,d_{\vartheta}$, we have
	\begin{equation}
	\begin{aligned}\label{eqt:Neyman orthogonal score restated in appendix}
	\partial_{\varrho}\mathbb{E}[\psi_{j}(W,\theta,\varrho)]\mid_{\varrho=\rho}[\varrho-\rho]:=D_{0,j}[\varrho-\rho]=0.
	\end{aligned}
	\end{equation}
\end{definition}

The orthogonal score functions which can be used to recover the estimators (hence, the estimates) of $\theta^{i}$, $\theta^{i\mid j}$, $\theta^{i,k}$ and $\theta^{i,k\mid j}$ are summarized in Theorem 3.2 in the main paper. Theorem 3.2 is restated here as the Theorem \ref{thm:Appendix-desired unbiased estimator}.
\begin{theorem}
	\label{thm:Appendix-desired unbiased estimator}
	Let $W=(Y,D,\mathbf{U},\mathbf{X},\mathbf{Z})$ and $\Theta_{i}$, $T_{i}$, $\Theta_{i,k}$, $T_{i,k}$, $\Theta_{i\mid j}$, $T_{i\mid j}$, $\Theta_{i,k\mid j}$ and $\Theta_{i,k\mid j}$ be convex set (see the definitions of these quantities below). 
	
	\textbf{(a) Score function of $\theta^{i}$:} The score function $\psi(W,\vartheta,\varrho)$ such that $\mathbb{E}[\psi(W,\theta^{i},\rho^{i})]=0$ and $\partial_{\varrho^{i}}\mathbb{E}[\psi(W,\theta^{i},\varrho^{i})]\mid_{\varrho^{i}=\rho^{i}}=0$ is
	\begin{equation}
	\begin{aligned}\label{eqt:orthogonal score function expectation restated in appendix}
	\vartheta-\mathcal{g}(d^{i},\mathbf{U},\mathbf{Z})-\frac{\mathbf{1}_{\{D=d^{i}\}}}{a_{i}(\mathbf{X},\mathbf{Z})}(Y-\mathcal{g}(d^{i},\mathbf{U},\mathbf{Z})),
	\end{aligned}
	\end{equation}
	where
	\begin{equation*}
	\begin{aligned}
	\Theta=\Theta_{i}&:=\{\vartheta^{i}=\mathbb{E}\left[\mathcal{g}(d^{i},\mathbf{U},\mathbf{Z})\right]\mid \textit{$\mathcal{g}$ is $\mathbb{P}$-integrable}\},\\
	T=T_{i}&:=\{\varrho^{i}=\left(\mathcal{g}(d^{i},\mathbf{U},\mathbf{Z}),a_{i}(\mathbf{X},\mathbf{Z})\right)\mid \textit{$\mathcal{g}$ is $\mathbb{P}$-integrable}\},\\
	\theta=\theta^{i}&:=\mathbb{E}\left[g(d^{i},\mathbf{U},\mathbf{Z})\right]\in\Theta_{i},\\
	\rho=\rho^{i}&:=\left(g(d^{i},\mathbf{U},\mathbf{Z}),\mathbb{E}\left[\mathbf{1}_{\{D=d^{i}\}}|\mathbf{X},\mathbf{Z}\right]\right)\in T_{i}.
	\end{aligned}
	\end{equation*}
	\textbf{(b) Score function of $\theta^{i\mid j}$:} Suppose $i\neq j$. The score function $\psi(W,\vartheta,\varrho)$ such that $\mathbb{E}[\psi(W,\theta^{i\mid j},\rho^{i\mid j})]=0$ and $\partial_{\varrho^{i\mid j}}\mathbb{E}[\psi(W,\theta^{i\mid j},\varrho^{i\mid j})]\mid_{\varrho^{i\mid j}=\rho^{i\mid j}}=0$ is
	\begin{equation}
	\begin{aligned}\label{eqt:orthogonal score function conditional expectation restated in appendix}
	\frac{1}{m_{j}}&\left\{\vartheta\mathbf{1}_{\{D=d^{j}\}}-\mathcal{g}(d^{i},\mathbf{U},\mathbf{Z})\mathbf{1}_{\{D=d^{j}\}}-\mathbf{1}_{\{D=d^{i}\}}\frac{a_{j}(\mathbf{X},\mathbf{Z})}{a_{i}(\mathbf{X},\mathbf{Z})}(Y-\mathcal{g}(d^{i},\mathbf{U},\mathbf{Z}))\right\},
	\end{aligned}
	\end{equation}
	where 
	\begin{equation*}
	\begin{aligned}
	&\Theta=\Theta_{i\mid j}:=\{\vartheta^{i\mid j}=\mathbb{E}\left[\mathcal{g}(d^{i},\mathbf{U},\mathbf{Z})\mid D=d^{j}\right]\mid \text{$\mathcal{g}$ is $\mathbb{P}$-integrable}\},\\
	&T=T_{i\mid j}:=\left\{\varrho^{i\mid j}=(g(d^{i},\mathbf{U},\mathbf{Z}),m_{j},a_{j}(\mathbf{X},\mathbf{Z}),a_{i}(\mathbf{X},\mathbf{Z}))\mid \text{$\mathcal{g}$ is $\mathbb{P}$-integrable}\right\},\\
	&\theta=\theta^{i\mid j}:=\mathbb{E}\left[g(d^{i},\mathbf{U},\mathbf{Z})\mid D=d^{j}\right]\in\Theta_{i\mid j},\\
	&\rho=\rho^{i\mid j}:=\left(g(d^{i},\mathbf{U},\mathbf{Z}),\mathbb{E}\left[\mathbf{1}_{\{D=d^{j}\}}\right],\mathbb{E}\left[\mathbf{1}_{\{D=d^{j}\}}|\mathbf{X},\mathbf{Z}\right],\mathbb{E}\left[\mathbf{1}_{\{D=d^{i}\}}|\mathbf{X},\mathbf{Z}\right]\right)\in T_{i\mid j}.
	\end{aligned}
	\end{equation*}
\end{theorem}

The proof of Theorem \ref{thm:Appendix-desired unbiased estimator} makes use of the following result from \cite{chernozhukov2018double2}:
\begin{result}\label{result:riesz representer}
	Let $X$ be the (predictor) variables, $Y$ be the response variable and $\mathcal{g}$ is a functional. Let $\vartheta=\mathbb{E}\left[m(X,\mathcal{g})\right]$ and $\theta=\mathbb{E}\left[m(X,g)\right]$. If $\mathcal{g}\mapsto \mathbb{E}\left[m(X,\mathcal{g})\right]$ is linear in $\mathcal{g}$, then there exists $\alpha(X)$ such that $\mathbb{E}\left[m(X,\mathcal{g})\right]=\mathbb{E}\left[\alpha(X)\mathcal{g}(X)\right]$. Furthermore, denoting $W=(Y,X)$. The score function associated with $\theta$ which is defined as
	\begin{equation}
	\begin{aligned}\label{eqt:riesz representer}
	\psi(W,\vartheta,\varrho):=\vartheta-m(X,\mathcal{g})-\alpha(X)[Y-\mathcal{g}(X)],\;\text{where $\varrho:=(\mathcal{g},\alpha)$}
	\end{aligned}
	\end{equation}
	is an orthogonal score function.
\end{result}

\begin{proof}[Proof of Theorem \ref{thm:Appendix-desired unbiased estimator}]
	Let $X=(D,\mathbf{U},\mathbf{X},\mathbf{Z})$ such that $W=(Y,X)=(Y,D,\mathbf{U},\mathbf{X},\mathbf{Z})$.
	We only need to prove (a) and (b) in the main paper. The proof is divided into two parts.
	\underline{\textbf{\textit{Proof of (a):}}}
	
	From the \eqref{eqt:causal model 1} in the main paper, let $X=(D,\mathbf{U},\mathbf{X},\mathbf{Z})$. We also assume that $m(X,\mathcal{g}):=\mathcal{g}(d^{i},\mathbf{U},\mathbf{Z})$. The true causal parameter is $\theta^{i}:=\mathbb{E}\left[m(X,g)\right]=\mathbb{E}\left[g(d^{i},\mathbf{U},\mathbf{Z})\right]$. Therefore, we know that
	\begin{equation*}
	\begin{aligned}
	\text{$\Theta_{i}$ is a convex set such that it equals $\{\vartheta^{i}=\mathbb{E}\left[\mathcal{g}(d^{i},\mathbf{U},\mathbf{Z})\right]\mid \mathcal{g}\; is\;\mathbb{P}\text{-}integrable\}$}
	\end{aligned}
	\end{equation*} 
	Using Result \ref{result:riesz representer}, we show that $\mathcal{g}\mapsto \mathbb{E}\left[m(X,\mathcal{g})\right]:=\mathbb{E}\left[\mathcal{g}(d^{i},\mathbf{U},\mathbf{Z})\right]$ is linear. Indeed, for any $\mathcal{g}_{1}$ and $\mathcal{g}_{2}$ as well as $a\in\mathbb{R}$, we have
	\begin{equation*}
	\begin{aligned}
	\mathbb{E}\left[m(X,\mathcal{g}_{1}+\mathcal{g}_{2})\right]:&=\mathbb{E}\left[(\mathcal{g}_{1}+\mathcal{g}_{2})(d^{i},\mathbf{U},\mathbf{Z})\right]=\mathbb{E}\left[\mathcal{g}_{1}(d^{i},\mathbf{U},\mathbf{Z})\right]+\mathbb{E}\left[\mathcal{g}_{2}(d^{i},\mathbf{U},\mathbf{Z})\right]\\
	&=\mathbb{E}\left[m(X,\mathcal{g}_{1})\right]+\mathbb{E}\left[m(X,\mathcal{g}_{2})\right]\\
	\mathbb{E}\left[m(X,(a \mathcal{g})\right]:&=\mathbb{E}\left[(a \mathcal{g})(d^{i},\mathbf{U},\mathbf{Z})\right]=a\mathbb{E}\left[\mathcal{g}(d^{i},\mathbf{U},\mathbf{Z})\right]=a\mathbb{E}\left[m(X,\mathcal{g})\right].
	\end{aligned}
	\end{equation*}
	Consequently, $\mathcal{g}\mapsto \mathbb{E}\left[m(X,\mathcal{g})\right]:=\mathbb{E}\left[\mathcal{g}(d^{i},\mathbf{U},\mathbf{Z})\right]$ is linear in $\mathcal{g}$. Making use of Result \ref{result:riesz representer}, we would like to find $\alpha(X)$ such that $\mathbb{E}\left[m(X,\mathcal{g})\right]=\mathbb{E}\left[\alpha(X)\mathcal{g}(X)\right]$, i.e., $\mathbb{E}\left[\mathcal{g}(d^{i},\mathbf{U},\mathbf{Z})\right]=\mathbb{E}\left[\alpha(D,\mathbf{U},\mathbf{X},\mathbf{Z})\mathcal{g}(D,\mathbf{U},\mathbf{Z})\right]$. We guess $\alpha(X)=\alpha(D,\mathbf{U},\mathbf{X},\mathbf{Z})=\frac{\mathbf{1}_{\{D=d^{i}\}}}{\mathbb{E}\left[\mathbf{1}_{\{D=d^{i}\}}\mid \mathbf{X},\mathbf{Z}\right]}$. In fact, we have
	\begin{equation*}
	\begin{aligned}
	\mathbb{E}\left[\alpha(D,\mathbf{U},\mathbf{X},\mathbf{Z})\mathcal{g}(D,\mathbf{U},\mathbf{Z})\right]&=\mathbb{E}\left[\frac{\mathbf{1}_{\{D=d^{i}\}}}{\mathbb{E}\left[\mathbf{1}_{\{D=d^{i}\}}\mid \mathbf{X},\mathbf{Z}\right]}\mathcal{g}(D,\mathbf{U},\mathbf{Z})\right]\\
	&=\mathbb{E}\left[\frac{\mathcal{g}(d^{i},\mathbf{U},\mathbf{Z})\mathbb{E}\left[\mathbf{1}_{\{D=d^{i}\}}\mid \mathbf{U},\mathbf{X},\mathbf{Z}\right]}{\mathbb{E}\left[\mathbf{1}_{\{D=d^{i}\}}\mid \mathbf{X},\mathbf{Z}\right]}\right]\\
	&=\mathbb{E}\left[\mathcal{g}(d^{i},\mathbf{U},\mathbf{Z})\right]=\mathbb{E}\left[m(X,\mathcal{g})\right].
	\end{aligned}
	\end{equation*}
	The second last equality follows since we assume that $\mathbf{U}$ is independent of $(\mathbf{X},\mathbf{Z})$.
	
	The proof is as follows: $\mathbb{E}\left[\mathbf{1}_{\{D=d^{i}\}}\mid \mathbf{U}=\mathbf{u},\mathbf{X}=\mathbf{x},\mathbf{Z}=\mathbf{z}\right]=\frac{h_{D,\mathbf{U},\mathbf{X},\mathbf{Z}}(d^{i},\mathbf{u},\mathbf{x},\mathbf{z})}{h_{\mathbf{U},\mathbf{X},\mathbf{Z}}(\mathbf{u},\mathbf{x},\mathbf{z})}$ and $\mathbb{E}\left[\mathbf{1}_{\{D=d^{i}\}}\mid \mathbf{X}=\mathbf{x},\mathbf{Z}=\mathbf{z}\right]=\frac{h_{D,\mathbf{X},\mathbf{Z}}(d^{i},\mathbf{x},\mathbf{z})}{h_{\mathbf{X},\mathbf{Z}}(\mathbf{x},\mathbf{z})}$, where $h_{D,\mathbf{U},\mathbf{X},\mathbf{Z}}(\cdot,\cdot,\cdot,\cdot)$ is the joint density function of $D$, $\mathbf{U}$, $\mathbf{X}$ and $\mathbf{Z}$; $h_{\mathbf{U},\mathbf{X},\mathbf{Z}}(\cdot,\cdot,\cdot)$ is the joint density function of $\mathbf{U}$, $\mathbf{X}$ and $\mathbf{Z}$; $h_{D,\mathbf{X},\mathbf{Z}}(\cdot,\cdot,\cdot)$ is the joint density function of $D$, $\mathbf{X}$ and $\mathbf{Z}$; $h_{\mathbf{X},\mathbf{Z}}(\cdot,\cdot)$ is the joint density function of $\mathbf{X}$ and $\mathbf{Z}$. 
	
	Now, we write $\mathbb{E}\left[\mathbf{1}_{\{D=d^{i}\}}\mid \mathbf{U}=\mathbf{u},\mathbf{X}=\mathbf{x},\mathbf{Z}=\mathbf{z}\right]$ as
	\begin{equation*}
	\begin{aligned}
	\mathbb{E}\left[\mathbf{1}_{\{D=d^{i}\}}\mid \mathbf{U}=\mathbf{u},\mathbf{X}=\mathbf{x},\mathbf{Z}=\mathbf{z}\right]&=\frac{h_{D,\mathbf{U},\mathbf{X},\mathbf{Z}}(d^{i},\mathbf{u},\mathbf{x},\mathbf{z})}{h_{\mathbf{U},\mathbf{X},\mathbf{Z}}(\mathbf{u},\mathbf{x},\mathbf{z})}\\
	&=\frac{h_{\mathbf{U}\mid D,\mathbf{X},\mathbf{Z}}(\mathbf{u}\mid d^{i},\mathbf{x},\mathbf{z})}{h_{\mathbf{U}\mid \mathbf{X},\mathbf{Z}}(\mathbf{u}\mid\mathbf{x},\mathbf{z})}\times\frac{h_{D,\mathbf{X},\mathbf{Z}}(d^{i},\mathbf{x},\mathbf{z})}{h_{\mathbf{X},\mathbf{Z}}(\mathbf{x},\mathbf{z})},
	\end{aligned}
	\end{equation*}
	where $h_{\mathbf{U}\mid D,\mathbf{X},\mathbf{Z}}(\cdot\mid \cdot,\cdot,\cdot)$ is the conditional density function of $\mathbf{U}$ given $d$, $\mathbf{X}$ and $\mathbf{Z}$ and $h_{\mathbf{U}\mid \mathbf{X},\mathbf{Z}}(\cdot\mid \cdot,\cdot)$ is the conditional density function of $\mathbf{U}$ given $\mathbf{X}$ and $\mathbf{Z}$. Since $\mathbf{U}$ is independent of $(\mathbf{X},\mathbf{Z})$, we know that $\frac{h_{\mathbf{U}\mid D,\mathbf{X},\mathbf{Z}}(\mathbf{u}\mid d^{i},\mathbf{x},\mathbf{z})}{h_{\mathbf{U}\mid \mathbf{X},\mathbf{Z}}(\mathbf{u}\mid\mathbf{x},\mathbf{z})}=1$. Therefore, we have $\mathbb{E}\left[\mathbf{1}_{\{D=d^{i}\}}\mid \mathbf{U}=\mathbf{u},\mathbf{X}=\mathbf{x},\mathbf{Z}=\mathbf{z}\right]=\mathbb{E}\left[\mathbf{1}_{\{D=d^{i}\}}\mid \mathbf{X}=\mathbf{x},\mathbf{Z}=\mathbf{z}\right]$, implying that $\mathbb{E}\left[\mathbf{1}_{\{D=d^{i}\}}\mid \mathbf{U},\mathbf{X},\mathbf{Z}\right]=\mathbb{E}\left[\mathbf{1}_{\{D=d^{i}\}}\mid \mathbf{X},\mathbf{Z}\right]$ a.s..
	
	As a consequence, we can construct the score function $\psi(W,\vartheta,\varrho)$ of $\theta^{i}$ as
	\begin{equation*}
	\begin{aligned}
	\vartheta-\mathcal{g}(d^{i},\mathbf{U},\mathbf{Z})-\frac{\mathbf{1}_{\{D=d^{i}\}}}{a_{i}(\mathbf{X},\mathbf{Z})}[Y-\mathcal{g}(d^{i},\mathbf{U},\mathbf{Z})],
	\end{aligned}
	\end{equation*}
	where $\varrho:=\left(\mathcal{g}(d^{i},\mathbf{U},\mathbf{Z}),a_{i}(\mathbf{X},\mathbf{Z})\right)$, $\rho:=\left(g(d^{i},\mathbf{U},\mathbf{Z}),\mathbb{E}\left[\mathbf{1}_{\{D=d^{i}\}}\mid \mathbf{X},\mathbf{Z}\right]\right)$ and $T_{i}:=\{\varrho^{i}\mid \mathcal{g} \textit{ is } \mathbb{P}\textit{-integrable}\}$.
	
	We then show that the score function is orthogonal through checking the Definition \ref{def:moment condition} and the Definition \ref{def:Neyman orthogonal score restated in appendix}.
	Note that 
	\begin{equation*}
	\begin{aligned}
	\mathbb{E}\left[\psi(W,\theta^{i},\rho)\right]&=\mathbb{E}\left[\theta^{i}-g(d^{i},\mathbf{U},\mathbf{Z})-\frac{\mathbf{1}_{\{D=d^{i}\}}}{\mathbb{E}\left[\mathbf{1}_{\{D=d^{i}\}}\mid \mathbf{X},\mathbf{Z}\right]}[Y-g(d^{i},\mathbf{U},\mathbf{Z})]\right]\\
	&=-\mathbb{E}\left[\frac{\mathbf{1}_{\{D=d^{i}\}}}{\mathbb{E}\left[\mathbf{1}_{\{D=d^{i}\}}\mid \mathbf{X},\mathbf{Z}\right]}[g(D,\mathbf{U},\mathbf{Z})+\xi-g(d^{i},\mathbf{U},\mathbf{Z})]\right]\\
	&=-\mathbb{E}\left[\frac{\mathbb{E}\left[\mathbf{1}_{\{D=d^{i}\}}\mid \mathbf{X},\mathbf{U},\mathbf{Z}\right]}{\mathbb{E}\left[\mathbf{1}_{\{D=d^{i}\}}\mid \mathbf{X},\mathbf{Z}\right]}\mathbb{E}\left[\xi\mid \mathbf{X},\mathbf{U},\mathbf{Z}\right]\right]=0.
	\end{aligned}
	\end{equation*}
	Moreover, we compute $\partial_{\mathcal{g}}\psi(W,\theta,\varrho)$ and $\partial_{a_{i}}\psi(W,\theta,\varrho)$, giving
	\begin{equation*}
	\begin{aligned}
	\partial_{\mathcal{g}}\psi(W,\theta,\varrho)&=-1+\frac{\mathbf{1}_{\{D=d^{i}\}}}{a_{i}(\mathbf{X},\mathbf{Z})}\quad \textit{ and }\quad \partial_{a_{i}}\psi(W,\theta,\varrho)=\frac{\mathbf{1}_{\{D=d^{i}\}}}{a_{i}(\mathbf{X},\mathbf{Z})^{2}}[Y-g(d^{i},\mathbf{U},\mathbf{Z})]
	\end{aligned}
	\end{equation*}
	respectively. Denoting $\mathcal{g}(d^{i},\mathbf{U},\mathbf{Z})$ and $g(d^{i},\mathbf{U},\mathbf{Z})$ as $\mathcal{g}^{i}$ and $g^{i}$ respectively, we compute $\mathbb{E}\left[\partial_{\mathcal{g}}\psi(W,\theta,\varrho)\mid_{\varrho=\rho}(\mathcal{g}^{i}-g^{i})\right]$ and $\mathbb{E}\left[\partial_{a_{i}}\psi(W,\theta,\varrho)\mid_{\varrho=\rho}(a_{i}-\mathbb{E}\left[\mathbf{1}_{\{D=d^{i}\}}\mid\mathbf{X},\mathbf{Z}\right])\right]$, which are all $0$. Indeed, we have
	\begin{equation*}
	\begin{aligned}
	\mathbb{E}\left[\partial_{\mathcal{g}}\psi(W,\theta,\varrho)\mid_{\varrho=\rho}(\mathcal{g}^{i}-g^{i})\right]&=\mathbb{E}\left[(\mathcal{g}^{i}-g^{i})\left(-1+\frac{\mathbb{E}\left[\mathbf{1}_{\{D=d^{i}\}}\mid \mathbf{U},\mathbf{X},\mathbf{Z}\right]}{\mathbb{E}\left[\mathbf{1}_{\{D=d^{i}\}}\mid \mathbf{X},\mathbf{Z}\right]}\right)\right]=0
	\end{aligned}
	\end{equation*}
	%
	and
	\begin{equation*}
	\begin{aligned}
	&\mathbb{E}\left[\partial_{a_{i}}\psi(W,\theta,\varrho)\mid_{\varrho=\rho}(a_{i}-\mathbb{E}\left[\mathbf{1}_{\{D=d^{i}\}}\mid\mathbf{X},\mathbf{Z}\right])\right]\\
	=&\mathbb{E}\left[\frac{\mathbf{1}_{\{D=d^{i}\}}\{Y-g(d^{i},\mathbf{U},\mathbf{Z})\}}{\mathbb{E}\left[\mathbf{1}_{\{D=d^{i}\}}\mid \mathbf{X},\mathbf{Z}\right]^{2}}(a_{i}-\mathbb{E}\left[\mathbf{1}_{\{D=d^{i}\}}\mid\mathbf{X},\mathbf{Z}\right])\right]\\
	=&\mathbb{E}\left[(a_{i}-\mathbb{E}\left[\mathbf{1}_{\{D=d^{i}\}}\mid\mathbf{X},\mathbf{Z}\right])\frac{\mathbb{E}\left[\mathbf{1}_{\{D=d^{i}\}}\mid\mathbf{X},\mathbf{U},\mathbf{Z}\right]\mathbb{E}\left[\xi\mid\mathbf{X},\mathbf{U},\mathbf{Z}\right]}{\mathbb{E}\left[\mathbf{1}_{\{D=d^{i}\}}\mid \mathbf{X},\mathbf{Z}\right]^{2}}\right]=0.\\
	\end{aligned}
	\end{equation*}
	Therefore, the orthogonal score functions which can be used to recover the $\theta^{i}$ is
	\begin{equation*}
	\begin{aligned}
	\vartheta-\mathcal{g}(d^{i},\mathbf{U},\mathbf{Z})-\frac{\mathbf{1}_{\{D=d^{i}\}}}{a_{i}(\mathbf{X},\mathbf{Z})}[Y-\mathcal{g}(d^{i},\mathbf{U},\mathbf{Z})]
	\end{aligned}.
	\end{equation*}
	
	\underline{\textbf{\textit{Proof of (b):}}}
	
	Next, we derive the orthogonal score function of $\theta^{i\mid j}:=\mathbb{E}\left[g(d^{i},\mathbf{U},\mathbf{Z})\mid D=d^{j}\right]=\frac{\mathbb{E}\left[\mathbf{1}_{\{D=d^{j}\}}g(d^{i},\mathbf{U},\mathbf{Z})\right]}{\mathbb{E}\left[\mathbf{1}_{\{D=d^{j}\}}\right]}$.
	We also let $m(X,\mathcal{g}):=\mathbf{1}_{\{D=d^{j}\}}\mathcal{g}(d^{i},\mathbf{U},\mathbf{Z})$, implying that $m(X,\mathcal{g}):=\mathbf{1}_{\{D=d^{j}\}}\mathcal{g}(d^{i},\mathbf{U},\mathbf{Z})$. We need to show that $\mathcal{g}\mapsto \mathbb{E}\left[m(X,\mathcal{g})\right]:=\mathbb{E}\left[\mathbf{1}_{\{D=d^{j}\}}\mathcal{g}(d^{i},\mathbf{U},\mathbf{Z})\right]$ is linear. In fact, for any $\mathcal{g}_{1}$, $\mathcal{g}_{2}$ and $a\in\mathbb{R}$, we have
	\begin{equation*}
	\begin{aligned}
	\mathbb{E}\left[m(X,\mathcal{g}_{1}+\mathcal{g}_{2})\right]:&=\mathbb{E}\left[\mathbf{1}_{\{D=d^{j}\}}(\mathcal{g}_{1}+\mathcal{g}_{2})(d^{i},\mathbf{U},\mathbf{Z})\right]\\
	&=\mathbb{E}\left[\mathbf{1}_{\{D=d^{j}\}}\mathcal{g}_{1}(d^{i},\mathbf{U},\mathbf{Z})\right]+\mathbb{E}\left[\mathbf{1}_{\{D=d^{j}\}}\mathcal{g}_{2}(d^{i},\mathbf{U},\mathbf{Z})\right]\\
	&=\mathbb{E}\left[m(X,\mathcal{g}_{1})\right]+\mathbb{E}\left[m(X,\mathcal{g}_{2})\right],\\
	\mathbb{E}\left[m(X, a \mathcal{g})\right]:&=\mathbb{E}\left[\mathbf{1}_{\{D=d^{j}\}}(a \mathcal{g})(d^{i},\mathbf{U},\mathbf{Z})\right]\\
	&=a\mathbb{E}\left[\mathbf{1}_{\{D=d^{j}\}}\mathcal{g}(d^{i},\mathbf{U},\mathbf{Z})\right]=a\mathbb{E}\left[m(X,\mathcal{g})\right].
	\end{aligned}
	\end{equation*}
	We then find $\alpha(X)$ such that $\mathbb{E}\left[m(X,\mathcal{g})\right]=\mathbb{E}\left[\alpha(X)\mathcal{g}(X)\right]$. Let $h_{D,\mathbf{U},\mathbf{X},\mathbf{Z}}(\cdot,\cdot,\cdot,\cdot)$ be the joint density function of $D$, $\mathbf{U}$, $\mathbf{X}$ and $\mathbf{Z}$, then
	\begin{equation*}
	\begin{aligned}
	\mathbb{E}[m(X,\mathcal{g})]=&\mathbb{E}[\mathbf{1}_{\{D=d^{j}\}}\mathcal{g}(d^{i},\mathbf{U},\mathbf{Z})]\\
	=&\sum_{k=1}^{n}\int\int\int\mathbf{1}_{\{d^{k}=d^{j}\}}\mathcal{g}(d^{i},\mathbf{u},\mathbf{z})h_{D,\mathbf{U},\mathbf{X},\mathbf{Z}}(d^{k},\mathbf{u},\mathbf{x},\mathbf{z})\;d\mathbf{u}\;d\mathbf{x}\;d\mathbf{\mathbf{z}},\\
	=&\int\int\int \mathcal{g}(d^{i},\mathbf{u},\mathbf{z})h_{D,\mathbf{U},\mathbf{X},\mathbf{Z}}(d^{j},\mathbf{u},\mathbf{x},\mathbf{z})\;d\mathbf{u}\;d\mathbf{x}\;d\mathbf{\mathbf{z}}\\
	\mathbb{E}[\alpha(X)\mathcal{g}(X)]=&\sum_{k=1}^{n}\int\int\int\alpha(d^{k},\mathbf{u},\mathbf{x},\mathbf{z})\mathcal{g}(d^{k},\mathbf{u},\mathbf{z})h_{D,\mathbf{U},\mathbf{X},\mathbf{Z}}(d^{k},\mathbf{u},\mathbf{x},\mathbf{z})\;d\mathbf{u}\;d\mathbf{x}\;d\mathbf{\mathbf{z}}.
	\end{aligned}
	\end{equation*}
	Equating the two equalities gives that
	\begin{equation*}
	\begin{aligned}
	\alpha(d,\mathbf{u},\mathbf{x},\mathbf{z})=\begin{cases}
	0,&\text{when $d\neq d^{i}$}\\
	\frac{h_{D,\mathbf{U},\mathbf{X},\mathbf{Z}}(d^{j},\mathbf{u},\mathbf{x},\mathbf{z})}{h_{D,\mathbf{U},\mathbf{X},\mathbf{Z}}(d^{i},\mathbf{u},\mathbf{x},\mathbf{z})},&\text{when $d=d^{i}$}
	\end{cases}
	\end{aligned}.
	\end{equation*}
	We derive another expression of $\frac{h_{D,\mathbf{U},\mathbf{X},\mathbf{Z}}(d^{j},\mathbf{u},\mathbf{x},\mathbf{z})}{h_{D,\mathbf{U},\mathbf{X},\mathbf{Z}}(d^{i},\mathbf{u},\mathbf{x},\mathbf{z})}$.
	Note that $\mathbb{E}[\mathbf{1}_{\{D=d^{i}\}}\mid\mathbf{X}=\mathbf{x},\mathbf{Z}=\mathbf{z}]=\overset{n}{\underset{k=1}{\sum}}\mathbf{1}_{\{d^{k}=d^{i}\}}\frac{h_{D,\mathbf{X},\mathbf{Z}}(d^{k},\mathbf{x},\mathbf{z})}{h_{\mathbf{X},\mathbf{Z}}(\mathbf{x},\mathbf{z})}=\frac{h_{D,\mathbf{x},\mathbf{z}}(d^{i},\mathbf{x},\mathbf{z})}{h_{\mathbf{X},\mathbf{Z}}(\mathbf{x},\mathbf{z})}$, where $h_{D,\mathbf{X},\mathbf{Z}}(\cdot,\cdot,\cdot)$ is the joint density function of $D$, $\mathbf{X}$ and $\mathbf{Z}$ and $h_{\mathbf{X},\mathbf{Z}}(\cdot,\cdot)$ is the joint density function of $\mathbf{X}$ and $\mathbf{Z}$. As a result, we have $\frac{\mathbb{E}[\mathbf{1}_{\{D=d^{j}\}}\mid\mathbf{X}=\mathbf{x},\mathbf{Z}=\mathbf{z}]}{\mathbb{E}[\mathbf{1}_{\{D=d^{i}\}}\mid\mathbf{X}=\mathbf{x},\mathbf{Z}=\mathbf{z}]}=\frac{h_{D,\mathbf{X},\mathbf{Z}}(d^{j},\mathbf{x},\mathbf{z})}{h_{D,\mathbf{X},\mathbf{Z}}(d^{i},\mathbf{x},\mathbf{z})}$.
	
	At the same time, we have
	\begin{equation*}
	\begin{aligned}
	\frac{h_{D,\mathbf{U},\mathbf{X},\mathbf{Z}}(d^{j},\mathbf{u},\mathbf{x},\mathbf{z})}{h_{D,\mathbf{U},\mathbf{X},\mathbf{Z}}(d^{i},\mathbf{u},\mathbf{x},\mathbf{z})}&=\frac{h_{\mathbf{U}\mid D,\mathbf{X},\mathbf{Z}}(\mathbf{u}\mid d^{j},\mathbf{x},\mathbf{z})}{h_{\mathbf{U}\mid D,\mathbf{X},\mathbf{Z}}(\mathbf{u}\mid d^{i},\mathbf{x},\mathbf{z})}\times \frac{\mathbb{E}[\mathbf{1}_{\{D=d^{j}\}}\mid\mathbf{X}=\mathbf{x},\mathbf{Z}=\mathbf{z}]}{\mathbb{E}[\mathbf{1}_{\{D=d^{i}\}}\mid\mathbf{X}=\mathbf{x},\mathbf{Z}=\mathbf{z}]},
	\end{aligned}
	\end{equation*}
	where $h_{\mathbf{U}\mid D,\mathbf{X},\mathbf{Z}}(\cdot\mid \cdot,\cdot,\cdot)$ is the conditional density function of $\mathbf{U}$ given $D$, $\mathbf{X}$ and $\mathbf{Z}$. According to the assumption that $\mathbf{U}$ is independent of $(\mathbf{X},\mathbf{Z})$, we know that $\frac{h_{\mathbf{U}\mid D,\mathbf{X},\mathbf{Z}}(\mathbf{u}\mid d^{j},\mathbf{x},\mathbf{z})}{h_{\mathbf{U}\mid D,\mathbf{X},\mathbf{Z}}(\mathbf{u}\mid d^{i},\mathbf{x},\mathbf{z})}=1$. Hence, we have
	\begin{equation*}
	\begin{aligned}
	\alpha(d,\mathbf{u},\mathbf{x},\mathbf{z})=\begin{cases}
	0,&\text{when $d\neq d^{i}$}\\
	\frac{\mathbb{E}\left[\mathbf{1}_{\{D=d^{j}\}}\mid\mathbf{X}=\mathbf{x},\mathbf{Z}=\mathbf{z}\right]}{\mathbb{E}\left[\mathbf{1}_{\{D=d^{i}\}}\mid\mathbf{X}=\mathbf{x},\mathbf{Z}=\mathbf{z}\right]},&\text{when $d=d^{i}$}
	\end{cases}
	\end{aligned},
	\end{equation*}
	implying that $\alpha(X):=\alpha(D,\mathbf{U},\mathbf{X},\mathbf{Z})=\mathbf{1}_{\{D=d^{i}\}}\frac{\mathbb{E}[\mathbf{1}_{\{D=d^{j}\}}\mid\mathbf{X},\mathbf{Z}]}{\mathbb{E}[\mathbf{1}_{\{D=d^{i}\}}\mid\mathbf{X},\mathbf{Z}]}$. 
	As a result, making use of Result \ref{result:riesz representer} gives the score function $\psi(W,\vartheta,\varrho)$ of $\theta^{i\mid j}$ which equals
	\begin{equation*}
	\begin{aligned}
	\frac{1}{m_{j}}\left\{\vartheta\mathbf{1}_{\{D=d^{j}\}}-\mathbf{1}_{\{D=d^{j}\}}\mathcal{g}(d^{i},\mathbf{U},\mathbf{Z})-\mathbf{1}_{\{D=d^{i}\}}\frac{a_{j}(\mathbf{X},\mathbf{Z})}{a_{i}(\mathbf{X},\mathbf{Z})}[Y-\mathcal{g}(d^{i},\mathbf{U},\mathbf{Z})]\right\},
	\end{aligned}
	\end{equation*}
	where
	\begin{equation*}
	\begin{aligned}
	\varrho&:=(\mathcal{g}(d^{i},\mathbf{U},\mathbf{Z}),m_{j},a_{j}(\mathbf{X},\mathbf{Z}),a_{i}(\mathbf{X},\mathbf{Z})),\\
	\rho&:=(g(d^{i},\mathbf{U},\mathbf{Z}),\mathbb{E}\left[\mathbf{1}_{\{D=d^{j}\}}\right],\mathbb{E}\left[\mathbf{1}_{\{D=d^{j}\}}\mid \mathbf{X},\mathbf{Z}\right],\mathbb{E}\left[\mathbf{1}_{\{D=d^{i}\}}\mid \mathbf{X},\mathbf{Z}\right]).
	\end{aligned}
	\end{equation*}
	We check if the score function \eqref{eqt:orthogonal score function conditional expectation restated in appendix} satisfies the Definition \ref{def:moment condition} and the Definition \ref{def:Neyman orthogonal score restated in appendix}. Now,
	\begin{equation*}
	\begin{aligned}
	&\mathbb{E}\left[\psi(W,\theta^{i\mid j},\rho)\right]\\
	=&\theta^{i\mid j}-\frac{\mathbb{E}\left[\mathbf{1}_{\{D=d^{j}\}}g(d^{i},\mathbf{U},\mathbf{Z})\right]}{\mathbb{E}\left[\mathbf{1}_{\{D=d^{j}\}}\right]}-\frac{1}{\mathbb{E}\left[\mathbf{1}_{\{D=d^{j}\}}\right]}\mathbb{E}\left[\frac{\mathbf{1}_{\{D=d^{i}\}}\mathbb{E}\left[\mathbf{1}_{\{D=d^{j}\}}\mid \mathbf{X},\mathbf{Z}\right]}{\mathbb{E}\left[\mathbf{1}_{\{D=d^{i}\}}\mid \mathbf{X},\mathbf{Z}\right]}[Y-g(d^{i},\mathbf{U},\mathbf{Z})]\right].
	\end{aligned}
	\end{equation*}
	We consider how to simplify $\mathbb{E}\left[\frac{\mathbf{1}_{\{D=d^{i}\}}\mathbb{E}\left[\mathbf{1}_{\{D=d^{j}\}}\mid \mathbf{X},\mathbf{Z}\right]}{\mathbb{E}\left[\mathbf{1}_{\{D=d^{i}\}}\mid \mathbf{X},\mathbf{Z}\right]}[Y-g(d^{i},\mathbf{U},\mathbf{Z})]\right]$. Note that it equals
	\begin{equation*}
	\begin{aligned}
	&\mathbb{E}\left[\frac{\mathbf{1}_{\{D=d^{i}\}}\mathbb{E}\left[\mathbf{1}_{\{D=d^{j}\}}\mid \mathbf{X},\mathbf{Z}\right]}{\mathbb{E}\left[\mathbf{1}_{\{D=d^{i}\}}\mid \mathbf{X},\mathbf{Z}\right]}[g(D,\mathbf{U},\mathbf{Z})+\xi-g(d^{i},\mathbf{U},\mathbf{Z})]\right]\\
	=&\mathbb{E}\left[\frac{\mathbb{E}\left[\mathbf{1}_{\{D=d^{j}\}}\mid \mathbf{X},\mathbf{Z}\right]}{\mathbb{E}\left[\mathbf{1}_{\{D=d^{i}\}}\mid \mathbf{X},\mathbf{Z}\right]}\mathbb{E}\left[\mathbf{1}_{\{D=d^{i}\}}\mid \mathbf{X},\mathbf{U},\mathbf{Z}\right]\mathbb{E}\left[\xi\mid \mathbf{X},\mathbf{U},\mathbf{Z}\right]\right]=0.
	\end{aligned}
	\end{equation*}
	Moreover, we compute $\partial_{\mathcal{g}}\psi(W,\theta^{i\mid j},\varrho)$, $\partial_{a_{i}}\psi(W,\theta^{i\mid j},\varrho)$, $\partial_{a_{j}}\psi(W,\theta^{i\mid j},\varrho)$ and $\partial_{m_{j}}\psi(W,\theta^{i\mid j},\varrho)$, giving
	\begin{equation*}
	\begin{aligned}
	\partial_{\mathcal{g}}\psi(W,\theta^{i\mid j},\varrho)&=-
	\frac{1}{m_{j}}\mathbf{1}_{\{D=d^{j}\}}+\frac{1}{m_{j}}\mathbf{1}_{\{D=d^{i}\}}\frac{a_{j}(\mathbf{X},\mathbf{Z})}{a_{i}(\mathbf{X},\mathbf{Z})},\\
	\partial_{a_{i}}\psi(W,\theta^{i\mid j},\varrho)&=\frac{1}{m_{j}}\frac{\mathbf{1}_{\{D=d^{i}\}}a_{j}(\mathbf{X},\mathbf{Z})}{a_{i}(\mathbf{X},\mathbf{Z})^{2}}[Y-g(d^{i},\mathbf{U},\mathbf{Z})],\\
	\partial_{a_{j}}\psi(W,\theta^{i\mid j},\varrho)&=-\frac{1}{m_{j}}\frac{\mathbf{1}_{\{D=d^{i}\}}}{a_{i}(\mathbf{X},\mathbf{Z})}[Y-g(d^{i},\mathbf{U},\mathbf{Z})],\\
	\partial_{m_{j}}\psi(W,\theta^{i\mid j},\varrho)&=-\frac{1}{m_{j}^{2}}\psi(W,\theta^{i\mid j},\varrho).
	\end{aligned}
	\end{equation*}
	Now, we have that $\mathbb{E}\left[\partial_{\mathcal{g}}\psi(W,\theta^{i\mid j},\varrho)\mid_{\varrho=\rho}(\mathcal{g}^{i}-g^{i})\right]$, $\mathbb{E}\left[\partial_{a_{i}}\psi(W,\theta^{i\mid j},\varrho)\mid_{\varrho=\rho}(a_{i}-\mathbb{E}\left[\mathbf{1}_{\{D=d^{i}\}}\mid\mathbf{X},\mathbf{Z}\right])\right]$, $\mathbb{E}\left[\partial_{a_{j}}\psi(W,\theta^{i\mid j},\varrho)\mid_{\varrho=\rho}(a_{j}-\mathbb{E}\left[\mathbf{1}_{\{D=d^{j}\}}\mid\mathbf{X},\mathbf{Z}\right])\right]$ and $\mathbb{E}\left[\partial_{m_{j}}\psi(W,\theta^{i\mid j},\varrho)\mid_{\varrho=\rho}(m_{j}-\mathbb{E}\left[\mathbf{1}_{\{D=d^{j}\}}\right])\right]$ are all $0$ where we denote $\mathcal{g}(d^{i},\mathbf{U},\mathbf{Z})$ and $g(d^{i},\mathbf{U},\mathbf{Z})$ as $\mathcal{g}^{i}$ and $g^{i}$ respectively.
	Therefore, the orthogonal score functions which can be used to recover $\theta^{i\mid j}$ is
	\begin{equation*}
	\begin{aligned}
	\frac{1}{m_{j}}\left\{\vartheta\mathbf{1}_{\{D=d^{j}\}}-\mathbf{1}_{\{D=d^{j}\}}\mathcal{g}(d^{i},\mathbf{U},\mathbf{Z})-\mathbf{1}_{\{D=d^{i}\}}\frac{a_{j}(\mathbf{X},\mathbf{Z})}{a_{i}(\mathbf{X},\mathbf{Z})}[Y-\mathcal{g}(d^{i},\mathbf{U},\mathbf{Z})]\right\}.
	\end{aligned}
	\end{equation*}

\end{proof}

\subsection{Score functions of IoC and Doubly Robust (DR) estimators}\label{sec:Score other than IwC}
In the section ``Proofs of Theorem", we present the score functions which can recover the IwC estimators of $\theta^{i}$ and $\theta^{i\mid j}$. In this section, we give out the score functions which are used to construct the IoC estimators and the doubly robust estimators (DREs) of $\theta^{i}$ and $\theta^{i\mid j}$, and check whether the score functions violate the orthogonal condition. As a result, we explain the differences between IoC estimators, DREs and IwC estimators from the theoretical standpoint. The results are summarized in the Proposition \ref{prop:score function of IoC} and the Proposition \ref{prop:score function of DR}.

To start with, we state the score functions that can recover the IoC estimators of $\theta^{i}$ and $\theta^{i\mid j}$. Let  $W=(Y,D,\mathbf{U},\mathbf{X},\mathbf{Z})$. Besides, we let $\Theta_{i}$ and $T_{i}$ be convex set such that
\begin{equation*}
\begin{aligned}
\Theta_{i}&:=\{\vartheta^i=\mathbb{E}\left[\mathcal{g}(d^{i},\mathbf{U},\mathbf{Z})\right]\mid \textit{$\mathcal{g}$ is $\mathbb{P}$-integrable}\},\quad\theta^{i}:=\mathbb{E}\left[g(d^{i},\mathbf{U},\mathbf{Z})\right]\in\Theta_{i},\\
T_{i}&:=\{\varrho^{i}=\mathcal{g}(d^{i},\mathbf{U},\mathbf{Z})\mid \textit{$\mathcal{g}$ is $\mathbb{P}$-integrable}\},\quad \rho^{i}:=g(d^{i},\mathbf{U},\mathbf{Z})\in T_{i}.
\end{aligned}
\end{equation*}
The score function which can recover the IoC estimator of $\theta^{i}$ is
\begin{equation}
\begin{aligned}\label{eqt:IoC score function expectation}
\vartheta-\mathcal{g}(d^{i},\mathbf{U},\mathbf{Z}).
\end{aligned}
\end{equation}
Simultaneously, let $\Theta_{i\mid j}$ and $T_{i\mid j}$ be convex set such that 
\begin{equation*}
\begin{aligned}
\Theta_{i\mid j}&:=\{\vartheta^{i|j}=\mathbb{E}\left[\mathcal{g}(d^{i},\mathbf{U},\mathbf{Z})\mid D=d^{j}\right]\mid \textit{$\mathcal{g}$ is $\mathbb{P}$-integrable}\},\quad\theta^{i\mid j}:=\mathbb{E}\left[g(d^{i},\mathbf{U},\mathbf{Z})\mid D=d^{j}\right]\in\Theta_{i\mid j},\\
T_{i\mid j}&:=\{\varrho^{i|j}=(\mathcal{g}(d^{i},\mathbf{U},\mathbf{Z}),m_{j})\mid \textit{$\mathcal{g}$ is $\mathbb{P}$-integrable}\},\quad \rho^{i|j}:=(g(d^{i},\mathbf{U},\mathbf{Z}),\mathbb{E}\left[\mathbf{1}_{\{D=d^{j}\}}\right])\in T_{i\mid j}.
\end{aligned}
\end{equation*}
Then the score function which can recover the IoC estimator of $\theta^{i\mid j}$ is
\begin{equation}
\begin{aligned}\label{eqt:IoC score function conditional expectation}
\frac{\vartheta\mathbf{1}_{\{D=d^{j}\}}}{m_{j}}-\frac{\mathbf{1}_{\{D=d^{j}\}}\mathcal{g}(d^{i},\mathbf{U},\mathbf{Z})}{m_{j}}.
\end{aligned}
\end{equation}
We check that whether \eqref{eqt:IoC score function expectation} and \eqref{eqt:IoC score function conditional expectation} satisfy the moment condition (the Definition \ref{def:moment condition}) and the orthogonal condition (the Definition \ref{def:Neyman orthogonal score restated in appendix}), which are summarized in Proposition \ref{prop:score function of IoC}.
\begin{proposition}\label{prop:score function of IoC} Both \eqref{eqt:IoC score function expectation} and \eqref{eqt:IoC score function conditional expectation} satisfy the moment condition but violate the orthogonal condition.
\end{proposition}
\begin{proof} 
	We divide the proofs into two parts. First, we check if the moment condition is satisfied for \eqref{eqt:IoC score function expectation} and \eqref{eqt:IoC score function conditional expectation}. Second, we check if the orthogonal condition is satisfied for \eqref{eqt:IoC score function expectation} and \eqref{eqt:IoC score function conditional expectation}.
	
	\underline{\textbf{Moment condition check for \eqref{eqt:IoC score function expectation} and \eqref{eqt:IoC score function conditional expectation}}}.
	
	We check if $\mathbb{E}\left[\theta^{i}-g(d^{i},\mathbf{U},\mathbf{Z})\right]$ and $\mathbb{E}\left[\frac{\theta^{i|j}\mathbf{1}_{\{D=d^{j}\}}}{\mathbb{E}\left[\mathbf{1}_{\{D=d^{j}\}}\right]}-\frac{\mathbf{1}_{\{D=d^{j}\}}g(d^{i},\mathbf{U},\mathbf{Z})}{\mathbb{E}\left[\mathbf{1}_{\{D=d^{j}\}}\right]}\right]$ are $0$. Indeed, we have
	\begin{equation*}
	\begin{aligned}
	\mathbb{E}\left[\theta^{i}-g(d^{i},\mathbf{U},\mathbf{Z})\right]&=\theta^{i}-\mathbb{E}\left[g(d^{i},\mathbf{U},\mathbf{Z})\right]=0,\\
	\mathbb{E}\left[\frac{\theta^{i|j}\mathbf{1}_{\{D=d^{j}\}}}{\mathbb{E}\left[\mathbf{1}_{\{D=d^{j}\}}\right]}-\frac{\mathbf{1}_{\{D=d^{j}\}}g(d^{i},\mathbf{U},\mathbf{Z})}{\mathbb{E}\left[\mathbf{1}_{\{D=d^{j}\}}\right]}\right]&=\mathbb{E}\left[\frac{\theta^{i|j}\mathbf{1}_{\{D=d^{j}\}}}{\mathbb{E}\left[\mathbf{1}_{\{D=d^{j}\}}\right]}\right]-\mathbb{E}\left[\frac{\mathbf{1}_{\{D=d^{j}\}}g(d^{i},\mathbf{U},\mathbf{Z})}{\mathbb{E}\left[\mathbf{1}_{\{D=d^{j}\}}\right]}\right]\\
	&=\frac{\theta^{i|j}\mathbb{E}\left[\mathbf{1}_{\{D=d^{j}\}}\right]}{\mathbb{E}\left[\mathbf{1}_{\{D=d^{j}\}}\right]}-\frac{\mathbb{E}\left[\mathbf{1}_{\{D=d^{j}\}}g(d^{i},\mathbf{U},\mathbf{Z})\right]}{\mathbb{E}\left[\mathbf{1}_{\{D=d^{j}\}}\right]}\\
	&=\theta^{i|j}-\mathbb{E}\left[g(d^{i},\mathbf{U},\mathbf{Z})\mid D=d^{j}\right]=0.
	\end{aligned}
	\end{equation*}
	\underline{\textbf{Orthogonal condition check for \eqref{eqt:IoC score function expectation} and \eqref{eqt:IoC score function conditional expectation}}}.
	
	Let
	\begin{equation*}
	\begin{aligned}
	\psi^{i}(W,\vartheta,\varrho)&=\vartheta-\mathcal{g}(d^{i},\mathbf{U},\mathbf{Z})\quad \text{and}\quad \psi^{i\mid j}(W,\vartheta,\varrho)=\frac{\vartheta\mathbf{1}_{\{D=d^{j}\}}}{m_{j}}-\frac{\mathbf{1}_{\{D=d^{j}\}}\mathcal{g}(d^{i},\mathbf{U},\mathbf{Z})}{m_{j}}.
	\end{aligned}
	\end{equation*}
	We compute the derivatives of $\psi^{i}(W,\vartheta,\varrho)$ and $\psi^{i\mid j}(W,\vartheta,\varrho)$ w.r.t. nuisance parameters $\mathcal{g}$ and $m_{j}$, which give
	\begin{equation*}
	\begin{aligned}
	\partial_{\mathcal{g}}\psi^{i}(W,\vartheta,\varrho)&=-1,\quad\partial_{\mathcal{g}}\psi^{i\mid j}(W,\vartheta,\varrho)=-\frac{\mathbf{1}_{\{D=d^{j}\}}}{m_{j}}\\
	\partial_{m_{j}}\psi^{i\mid j}(W,\vartheta,\varrho)&=-\frac{\vartheta\mathbf{1}_{\{D=d^{j}\}}-\mathbf{1}_{\{D=d^{j}\}}\mathcal{g}(d^{i},\mathbf{U},\mathbf{Z})}{m_{j}^{2}}.
	\end{aligned}
	\end{equation*}
	Denoting $\mathcal{g}(d^{i},\mathbf{U},\mathbf{Z})$ and $g(d^{i},\mathbf{U},\mathbf{Z})$ as $\mathcal{g}^{i}$ and $g^{i}$ respectively, we have
	\begin{equation*}
	\begin{aligned}
	\mathbb{E}\left[\partial_{\mathcal{g}}\psi^{i}(W,\vartheta,\;\varrho)\mid_{\vartheta=\theta^{i},\varrho=\rho}(\mathcal{g}^{i}-g^{i})\right]&\neq 0,\quad \mathbb{E}\left[\partial_{\mathcal{g}}\psi^{i\mid j}(W,\vartheta,\varrho)\mid_{\vartheta=\theta^{i|j},\;\varrho=\rho}(\mathcal{g}^{i}-g^{i})\right]\neq 0
	\end{aligned}
	\end{equation*}
	and
	\begin{equation*}
	\begin{aligned}
	&\mathbb{E}\left[\partial_{m_{j}}\psi^{i\mid j}(W,\vartheta,\varrho)\mid_{\vartheta=\theta^{i|j},\;\varrho=\rho}(m_{j}-\mathbb{E}\left[\mathbf{1}_{\{D=d^{j}\}}\right])\right]\\
	=&-\theta^{i|j}\mathbb{E}\left[\frac{\mathbf{1}_{\{D=d^{j}\}}(m_{j}-\mathbb{E}\left[\mathbf{1}_{\{D=d^{j}\}}\right])}{\left(\mathbb{E}\left[\mathbf{1}_{\{D=d^{j}\}}\right]\right)^{2}}\right]+\mathbb{E}\left[\frac{\mathbf{1}_{\{D=d^{j}\}}g(d^{i},\mathbf{U},\mathbf{Z})(m_{j}-\mathbb{E}\left[\mathbf{1}_{\{D=d^{j}\}}\right])}{\left(\mathbb{E}\left[\mathbf{1}_{\{D=d^{j}\}}\right]\right)^{2}}\right]\\
	=&-\theta^{i|j}\frac{(m_{j}-\mathbb{E}\left[\mathbf{1}_{\{D=d^{j}\}}\right])}{\mathbb{E}\left[\mathbf{1}_{\{D=d^{j}\}}\right]}+\frac{\frac{\mathbb{E}\left[\mathbf{1}_{\{D=d^{j}\}}g(d^{i},\mathbf{U},\mathbf{Z})\right]}{\mathbb{E}\left[\mathbf{1}_{\{D=d^{j}\}}\right]}(m_{j}-\mathbb{E}\left[\mathbf{1}_{\{D=d^{j}\}}\right])}{\mathbb{E}\left[\mathbf{1}_{\{D=d^{j}\}}\right]}=0.
	\end{aligned}
	\end{equation*}
	Consequently, the orthogonal condition does not satisfied for \eqref{eqt:IoC score function expectation} and \eqref{eqt:IoC score function conditional expectation}.
\end{proof}
We can obtain the IoC estimators (the corresponding IoC estimations are computed using \eqref{eqt:IoC estimator for IwC formulation}) in the main paper. Clearly, the locally biasedness comes from the estimation of $g(d^{i},\cdot,\cdot)$ and we need to have a good estimation of $g(d^{i},\cdot,\cdot)$. Otherwise, the IoC estimators can be biased and lead to wrong conclusion in the causal inference analysis.

Next, we consider the score functions of $\theta^{i}$ and $\theta^{i\mid j}$ that we can use to recover the corresponding DREs respectively. 
According to \cite{farrell2015robust}, the score function which can be used to recover the DRE of $\theta^{i}$ is
\begin{equation}
\begin{aligned}\label{eqt:DR score function expectation}
\vartheta-\mathcal{g}(d^{i},\mathbf{U},\mathbf{Z})-\frac{\mathbf{1}_{\{D=d^{i}\}}}{a_{i}(\mathbf{X},\mathbf{Z})}(Y-\mathcal{g}(d^{i},\mathbf{U},\mathbf{Z}))
\end{aligned}
\end{equation}
while the score function which can be used to recover the DRE of $\theta^{i\mid j}$ is
\begin{equation}
\begin{aligned}\label{eqt:DR score function conditional expectation}
\vartheta-\frac{\mathbf{1}_{\{D=d^{j}\}}\mathcal{g}(d^{i},\mathbf{U},\mathbf{Z})}{m_{j}}-\frac{\mathbf{1}_{\{D=d^{i}\}}a_{j}(\mathbf{X},\mathbf{Z})}{m_{j}a_{i}(\mathbf{X},\mathbf{Z})}(Y-\mathcal{g}(d^{i},\mathbf{U},\mathbf{Z})).
\end{aligned}
\end{equation}
In the Proposition \ref{prop:score function of DR}, we check that if the score functions satisfy both the moment condition and the orthogonal condition. The results are summarized in the Proposition \ref{prop:score function of DR}.
\begin{proposition}\label{prop:score function of DR}\eqref{eqt:DR score function expectation} satisfies both the moment condition and the orthogonal condition. On the other hand, \eqref{eqt:DR score function conditional expectation} satisfies the moment condition but violates the orthogonal condition.
\end{proposition}
\begin{proof}
	From \eqref{eqt:DR score function expectation}, we notice that the score function which is used to recover the DRE of $\theta^{i}$ is the same as the score function which is used to recover the IwC estimator of $\theta^{i}$. The moment condition and the orthogonal condition have been checked (see the proofs of \textit{\textbf{(a)}} of the Theorem 3.2).
	
	It remains to check the moment condition and the orthogonal condition of \eqref{eqt:DR score function conditional expectation}. Let
	\begin{equation*}
	\begin{aligned}
	\psi^{i\mid j}:=\vartheta-\frac{\mathbf{1}_{\{D=d^{j}\}}\mathcal{g}(d^{i},\mathbf{U},\mathbf{Z})}{m_{j}}-\frac{\mathbf{1}_{\{D=d^{i}\}}a_{j}(\mathbf{X},\mathbf{Z})}{m_{j}a_{i}(\mathbf{X},\mathbf{Z})}(Y-\mathcal{g}(d^{i},\mathbf{U},\mathbf{Z})).
	\end{aligned}
	\end{equation*}
	Again, we divide into two parts.
	
	\underline{\textbf{Moment condition check}}. Taking expectation on \eqref{eqt:DR score function conditional expectation}, we have
	\begin{equation*}
	\begin{aligned}
	&\mathbb{E}\left[\theta^{i|j}-\frac{\mathbf{1}_{\{D=d^{j}\}}g(d^{i},\mathbf{U},\mathbf{Z})}{\mathbb{E}\left[\mathbf{1}_{\{D=d^{j}\}}\right]}-\frac{\mathbf{1}_{\{D=d^{i}\}}\mathbb{E}\left[\mathbf{1}_{\{D=d^{j}\}}\mid \mathbf{X},\mathbf{Z}\right]}{\mathbb{E}\left[\mathbf{1}_{\{D=d^{j}\}}\right]\mathbb{E}\left[\mathbf{1}_{\{D=d^{i}\}}\mid \mathbf{X},\mathbf{Z}\right]}(Y-g(d^{i},\mathbf{U},\mathbf{Z}))\right]\\
	=&\theta^{i|j}-\mathbb{E}\left[\frac{\mathbf{1}_{\{D=d^{j}\}}g(d^{i},\mathbf{U},\mathbf{Z})}{\mathbb{E}\left[\mathbf{1}_{\{D=d^{j}\}}\right]}\right]-\mathbb{E}\left[\frac{\mathbf{1}_{\{D=d^{i}\}}\mathbb{E}\left[\mathbf{1}_{\{D=d^{j}\}}\mid \mathbf{X},\mathbf{Z}\right]}{\mathbb{E}\left[\mathbf{1}_{\{D=d^{j}\}}\right]\mathbb{E}\left[\mathbf{1}_{\{D=d^{i}\}}\mid \mathbf{X},\mathbf{Z}\right]}(Y-g(d^{i},\mathbf{U},\mathbf{Z}))\right]\\
	\\
	=&\theta^{i|j}-\frac{\mathbb{E}\left[\mathbf{1}_{\{D=d^{j}\}}g(d^{i},\mathbf{U},\mathbf{Z})\right]}{\mathbb{E}\left[\mathbf{1}_{\{D=d^{j}\}}\right]}-\mathbb{E}\left[\frac{\mathbf{1}_{\{D=d^{i}\}}\mathbb{E}\left[\mathbf{1}_{\{D=d^{j}\}}\mid \mathbf{X},\mathbf{Z}\right]}{\mathbb{E}\left[\mathbf{1}_{\{D=d^{j}\}}\right]\mathbb{E}\left[\mathbf{1}_{\{D=d^{i}\}}\mid \mathbf{X},\mathbf{Z}\right]}(Y-g(d^{i},\mathbf{U},\mathbf{Z}))\right]=0.
	\end{aligned}
	\end{equation*}
	The last term in the last equality equals $0$, which can be found in the proof of (b) of the Theorem 3.2.
	
	\underline{\textbf{Orthogonal condition check}}. We first compute $\partial_{\mathcal{g}}\psi^{i\mid j}$, $\partial_{m_{j}}\psi^{i\mid j}$, $\partial_{a_{i}}\psi^{i\mid j}$ and $\partial_{a_{j}}\psi^{i\mid j}$:
	\begin{equation*}
	\begin{aligned}
	\partial_{\mathcal{g}}\psi^{i\mid j}&=-\frac{\mathbf{1}_{\{D=d^{j}\}}}{m_{j}}+\frac{\mathbf{1}_{\{D=d^{i}\}}a_{j}(\mathbf{X},\mathbf{Z})}{m_{j}a_{i}(\mathbf{X},\mathbf{Z})}\\
	\partial_{m_{j}}\psi^{i\mid j}&=\frac{\mathbf{1}_{\{D=d^{j}\}}\mathcal{g}(d^{i},\mathbf{U},\mathbf{Z})}{m_{j}^{2}}+\frac{\mathbf{1}_{\{D=d^{i}\}}a_{j}(\mathbf{X},\mathbf{Z})}{m_{j}^{2}a_{i}(\mathbf{X},\mathbf{Z})}(Y-\mathcal{g}(d^{i},\mathbf{U},\mathbf{Z}))\\
	\partial_{a_{i}}\psi^{i\mid j}&=\frac{\mathbf{1}_{\{D=d^{i}\}}a_{j}(\mathbf{X},\mathbf{Z})}{m_{j}(a_{i}(\mathbf{X},\mathbf{Z}))^{2}}(Y-\mathcal{g}(d^{i},\mathbf{U},\mathbf{Z}))\\
	\partial_{a_{j}}\psi^{i\mid j}&=-\frac{\mathbf{1}_{\{D=d^{i}\}}}{m_{j}a_{i}(\mathbf{X},\mathbf{Z})}(Y-\mathcal{g}(d^{i},\mathbf{U},\mathbf{Z})).
	\end{aligned}
	\end{equation*}
	Denoting $\mathcal{g}(d^{i},\mathbf{U},\mathbf{Z})$ and $g^{i}(d^{i},\mathbf{U},\mathbf{Z})$ as $\mathcal{g}^{i}$ and $g^{i}$ respectively, we then compute
	\begin{equation*}
	\begin{aligned}
	&\mathbb{E}\left[\partial_{\mathcal{g}}\psi^{i\mid j}\mid_{\vartheta^{i\mid j}=\theta^{i\mid j}\;\varrho=\rho}(\mathcal{g}^{i}-g^{i})\right]\\
	&\mathbb{E}\left[\partial_{m_{j}}\psi^{i\mid j}\mid_{\vartheta^{i\mid j}=\theta^{i\mid j}\;\varrho=\rho}(m_{j}-\mathbb{E}\left[\mathbf{1}_{\{D=d^{j}\}}\right])\right]\\
	&\mathbb{E}\left[\partial_{a_{i}}\psi^{i\mid j}\mid_{\vartheta^{i\mid j}=\theta^{i\mid j}\;\varrho=\rho}(a_{i}-\mathbb{E}\left[\mathbf{1}_{\{D=d^{i}\}}\mid\mathbf{X},\mathbf{Z}\right])\right]\\
	&\mathbb{E}\left[\partial_{a_{j}}\psi^{i\mid j}\mid_{\vartheta^{i\mid j}=\theta^{i\mid j}\;\varrho=\rho}(a_{j}-\mathbb{E}\left[\mathbf{1}_{\{D=d^{j}\}}\mathbf{X},\mathbf{Z}\right])\right].
	\end{aligned}
	\end{equation*}
	Indeed, we have
	\begin{equation*}
	\begin{aligned}
	&\mathbb{E}\left[\partial_{\mathcal{g}}\psi^{i\mid j}\mid_{\vartheta^{i\mid j}=\theta^{i\mid j}\;\varrho=\rho}(\mathcal{g}^{i}-g^{i})\right]\\
	=&\mathbb{E}\left[\frac{(\mathcal{g}^{i}-g^{i})}{\mathbb{E}\left[\mathbf{1}_{\{D=d^{j}\}}\right]}\mathbb{E}\left[-\mathbf{1}_{\{D=d^{j}\}}+\frac{\mathbf{1}_{\{D=d^{i}\}}\mathbb{E}\left[\mathbf{1}_{\{D=d^{j}\}}\mid\mathbf{X},\mathbf{Z}\right]}{\mathbb{E}\left[\mathbf{1}_{\{D=d^{i}\}}\mid\mathbf{X},\mathbf{Z}\right]}\mid\mathbf{U},\mathbf{X},\mathbf{Z}\right]\right]=0,
	\end{aligned}
	\end{equation*}
	
	\begin{equation*}
	\begin{aligned}
	&\mathbb{E}\left[\partial_{m_{j}}\psi^{i\mid j}\mid_{\vartheta^{i\mid j}=\theta^{i\mid j}\;\varrho=\rho}(m_{j}-\mathbb{E}\left[\mathbf{1}_{\{D=d^{j}\}}\right])\right]\\
	=&\mathbb{E}\left[\frac{\mathbf{1}_{\{D=d^{j}\}}g(d^{i},\mathbf{U},\mathbf{Z})}{\mathbb{E}\left[\mathbf{1}_{\{D=d^{j}\}}\right]^{2}}(m_{j}-\mathbb{E}\left[\mathbf{1}_{\{D=d^{j}\}}\right])\right]\\
	&\qquad\qquad+\mathbb{E}\left[\frac{\mathbf{1}_{\{D=d^{i}\}}\mathbb{E}\left[\mathbf{1}_{\{D=d^{j}\}}\mid\mathbf{X},\mathbf{Z}\right]}{\mathbb{E}\left[\mathbf{1}_{\{D=d^{j}\}}\right]^{2}\mathbb{E}\left[\mathbf{1}_{\{D=d^{i}\}}\mid\mathbf{X},\mathbf{Z}\right]}(Y-g(d^{i},\mathbf{U},\mathbf{Z}))(m_{j}-\mathbb{E}\left[\mathbf{1}_{\{D=d^{j}\}}\right])\right]\\
	=&\mathbb{E}\left[\frac{\mathbf{1}_{\{D=d^{j}\}}g(d^{i},\mathbf{U},\mathbf{Z})}{\mathbb{E}\left[\mathbf{1}_{\{D=d^{j}\}}\right]^{2}}(m_{j}-\mathbb{E}\left[\mathbf{1}_{\{D=d^{j}\}}\right])\right]\\
	&\qquad\qquad+\mathbb{E}\left[\frac{\mathbb{E}\left[\mathbf{1}_{\{D=d^{i}\}}\mid\mathbf{U},\mathbf{X},\mathbf{Z}\right]\mathbb{E}\left[\mathbf{1}_{\{D=d^{j}\}}\mid\mathbf{X},\mathbf{Z}\right]}{\mathbb{E}\left[\mathbf{1}_{\{D=d^{j}\}}\right]^{2}\mathbb{E}\left[\mathbf{1}_{\{D=d^{i}\}}\mid\mathbf{X},\mathbf{Z}\right]}\mathbb{E}\left[\xi\mid\mathbf{U},\mathbf{X},\mathbf{Z})\right](m_{j}-\mathbb{E}\left[\mathbf{1}_{\{D=d^{j}\}}\right])\right]\\
	=&\mathbb{E}\left[\frac{\mathbf{1}_{\{D=d^{j}\}}g(d^{i},\mathbf{U},\mathbf{Z})}{\mathbb{E}\left[\mathbf{1}_{\{D=d^{j}\}}\right]^{2}}(m_{j}-\mathbb{E}\left[\mathbf{1}_{\{D=d^{j}\}}\right])\right]\neq 0,
	\end{aligned}
	\end{equation*}
	
	\begin{equation*}
	\begin{aligned}
	&\mathbb{E}\left[\partial_{a_{i}}\psi^{i\mid j}\mid_{\vartheta^{i\mid j}=\theta^{i\mid j}\;\varrho=\rho}(a_{i}-\mathbb{E}\left[\mathbf{1}_{\{D=d^{i}\}}\mid\mathbf{X},\mathbf{Z}\right])\right])\\
	=&\mathbb{E}\left[\frac{\mathbf{1}_{\{D=d^{i}\}}\mathbb{E}\left[\mathbf{1}_{\{D=d^{j}\}}\mid\mathbf{X},\mathbf{Z}\right]}{\mathbb{E}\left[\mathbf{1}_{\{D=d^{j}\}}\right]\mathbb{E}\left[\mathbf{1}_{\{D=d^{i}\}}\mid\mathbf{X},\mathbf{Z}\right]^{2}}(Y-g(d^{i},\mathbf{U},\mathbf{Z}))(a_{i}-\mathbb{E}\left[\mathbf{1}_{\{D=d^{i}\}}\mid\mathbf{X},\mathbf{Z}\right])\right]\\
	=&\mathbb{E}\left[\frac{\mathbb{E}\left[\mathbf{1}_{\{D=d^{i}\}}\mid\mathbf{U},\mathbf{X},\mathbf{Z}\right]\mathbb{E}\left[\mathbf{1}_{\{D=d^{j}\}}\mid\mathbf{X},\mathbf{Z}\right]}{\mathbb{E}\left[\mathbf{1}_{\{D=d^{j}\}}\right]\mathbb{E}\left[\mathbf{1}_{\{D=d^{i}\}}\mid\mathbf{X},\mathbf{Z}\right]^{2}}\mathbb{E}\left[\xi\mid \mathbf{U},\mathbf{X},\mathbf{Z}\right](a_{i}-\mathbb{E}\left[\mathbf{1}_{\{D=d^{i}\}}\mid\mathbf{X},\mathbf{Z}\right])\right]=0
	\end{aligned}
	\end{equation*}
	and
	\begin{equation*}
	\begin{aligned}
	&\mathbb{E}\left[\partial_{a_{j}}\psi^{i\mid j}\mid_{\vartheta^{i\mid j}=\theta^{i\mid j}\;\varrho=\rho}(a_{j}-\mathbb{E}\left[\mathbf{1}_{\{D=d^{j}\}}\mid\mathbf{X},\mathbf{Z}\right])\right]\\
	=&-\mathbb{E}\left[\frac{\mathbf{1}_{\{D=d^{i}\}}}{\mathbb{E}\left[\mathbf{1}_{\{D=d^{j}\}}\right]\mathbb{E}\left[\mathbf{1}_{\{D=d^{i}\}}\mid \mathbf{X},\mathbf{Z}\right]}(Y-g(d^{i},\mathbf{U},\mathbf{Z}))(a_{j}-\mathbb{E}\left[\mathbf{1}_{\{D=d^{j}\}}\mid\mathbf{X},\mathbf{Z}\right])\right]\\
	=&-\mathbb{E}\left[\frac{\mathbb{E}\left[\mathbf{1}_{\{D=d^{i}\}}\mid\mathbf{U},\mathbf{X},\mathbf{Z}\right]}{\mathbb{E}\left[\mathbf{1}_{\{D=d^{j}\}}\right]\mathbb{E}\left[\mathbf{1}_{\{D=d^{i}\}}\mid \mathbf{X},\mathbf{Z}\right]}\mathbb{E}\left[\xi\mid\mathbf{U},\mathbf{X},\mathbf{Z}\right](a_{j}-\mathbb{E}\left[\mathbf{1}_{\{D=d^{j}\}}\mid\mathbf{X},\mathbf{Z}\right])\right]=0.
	\end{aligned}
	\end{equation*}
	Hence, \eqref{eqt:DR score function conditional expectation} violates the orthogonal condition.
\end{proof}
When using the DRE of $\theta^{i\mid j}$, it requires the estimation of the quantity $\mathbb{E}\left[\mathbf{1}_{\{D=d^{j}\}}\right]$, which cause the locally biasedness of the DRE of $\theta^{i\mid j}$. This is different from our IwC estimator of $\theta^{i\mid j}$, which do not require the estimation of $\mathbb{E}\left[\mathbf{1}_{\{D=d^{j}\}}\right]$.
\clearpage
\subsection{Consistency}\label{sec:Consistency}
In this section we study the consistency of our IwC estimators of $\theta^{i}$ and $\theta^{i\mid j}$ by decomposing the error in estimating the `true' quantity using the IwC estimators. There are several advantages of using the error decomposition: 1. We can understand which terms contribute the most in estimating the true quantity $\mathbb{E}\left[g(d^{i},\mathbf{U},\mathbf{Z})\right]$ or $\mathbb{E}\left[g(d^{i},\mathbf{U},\mathbf{Z})\mid D=d^{j}\right]$; 2. We can understand if our estimators are consistent estimators.
Lemma \ref{lemma:unbiasedness and consistency} would help us determine if an estimator is consistent.
\begin{lemma}\label{lemma:unbiasedness and consistency}
	Given a probability space $(\Omega,\mathcal{F},\mathbb{P})$ and $\mathbf{V}$ is a random variable such that $\mathbf{V}_{1},\cdots,\mathbf{V}_{n}$ are independent and identically distributed with $\mathbf{V}$. Let $f$ be a scalar-valued function such that $f(\mathbf{V}_{i})$ have finite mean and finite variance. Defining $S:=\frac{\overset{n}{\underset{i=1}{\sum}}f(\mathbf{V}_{i})}{n}$. Then we have
	\begin{enumerate}
		\item $S$ is an unbiased estimator of $f(\mathbf{V})$\label{lemma:unbiasedness};
		\item $S$ converges to $f(\mathbf{V})$ in probability\label{lemma:consistency}.
	\end{enumerate}
\end{lemma}

The proof of the Lemma \ref{lemma:unbiasedness and consistency} is omitted. Besides, there are several Lemmas that are useful when we want to prove the consistency of an estimator.
\begin{lemma}\label{lemma:stochastic_bound}
Given an unknown quantity $\theta$. For each $n\in\mathbb{N}$, let $T_{n}$ and $S_{n}$ be the estimators based on the first $n$ data points. The following claims are correct.
\begin{enumerate}
\item If $T_{n}-\theta=O_{P}(n^{-\frac{1}{2}})$, then $T_{n}$ is a consistent estimator of $\theta$.\label{lemma:stochastic_bound_true}
\item If $T_{n}-S_{n}=O_{P}(n^{-\frac{1}{2}})$, then $T_{n}-S_{n}$ converges to $0$ in probability.\label{lemma:stochastic_bound_sequence}
\end{enumerate}
\end{lemma}
\begin{proof}\ \\
\underline{\textit{\textbf{Proof of \ref{lemma:stochastic_bound_true}:}}} Since $T_{n}-\theta=O_{P}(n^{-\frac{1}{2}})$, we have\\
$\forall \epsilon>0$, $\exists$ $M>0$ and $K>0$ such that
\begin{equation}
\begin{aligned}
Pr\left\{\left|\frac{T_{n}-\theta}{n^{-\frac{1}{2}}}\right|\geq M\right\}\leq\epsilon,\qquad \forall n\geq K.
\end{aligned}
\end{equation}
To prove that $T_{n}$ is a consistent estimator of $\theta$, we need to show that\\
$\forall \delta>0$ and $\forall \epsilon>0$, there exists $K^{'}>0$ such that
\begin{equation}
\begin{aligned}
Pr\left\{\left|T_{n}-\theta\right|\geq \delta\right\}\leq\epsilon,\qquad \forall n\geq K^{'}.
\end{aligned}
\end{equation}
Considering $Pr\left\{\left|T_{n}-\theta\right|\geq \delta\right\}$. We have
\begin{equation*}
\begin{aligned}
Pr\left\{\left|T_{n}-\theta\right|\geq \delta\right\}&= Pr\left\{\left|\frac{T_{n}-\theta}{n^{-\frac{1}{2}}}\right|\geq \frac{\delta}{n^{-\frac{1}{2}}}\right\}.
\end{aligned}
\end{equation*}
Now, take $K^{'}=\max\left\{K,\;\left(\frac{M}{\delta}\right)^{2} \right\}$. Then for any $n\geq K^{'}$, we have $n\geq\left(\frac{M}{\delta}\right)^{2}\Rightarrow \frac{\delta}{n^{-\frac{1}{2}}}\geq M$. Hence, $Pr\left\{\left|T_{n}-\theta\right|\geq \delta\right\}=Pr\left\{\left|\frac{T_{n}-\theta}{n^{-\frac{1}{2}}}\right|\geq \frac{\delta}{n^{-\frac{1}{2}}}\right\}\leq Pr\left\{\left|\frac{T_{n}-\theta}{n^{-\frac{1}{2}}}\right|\geq M\right\}\leq\epsilon$ $\forall n\geq K^{'}$.\\\\
The proof of \ref{lemma:stochastic_bound_sequence} in the Lemma \ref{lemma:stochastic_bound} is similar to the proof of \ref{lemma:stochastic_bound_true} in the Lemma \ref{lemma:stochastic_bound} and hence is omitted.
\end{proof}

\begin{lemma}\label{lemma:stochastic bound fraction}
Given a random variable $X>0\;a.s.$. Suppose $X_{n}>0\;a.s.$ is a sequence of random variables. Simultaneously, let $A_{n}$ be a sequence of positive constant which converges to $0$. If $X_{n}-X=O_{p}(A_{n})$, then $\frac{1}{X_{n}}-\frac{1}{X}=O_{p}(A_{n})$. Here, we say that $X_{n}=O_{p}(Y_{n})$ if for any $\epsilon>0$, there exists $M>0$ and $N>0$ such that
\begin{equation}
\begin{aligned}\label{eqt:stochastic bound fraction}
Pr\left\{\left|\frac{X_{n}}{Y_{n}}\right|\geq M\right\}\leq\epsilon.
\end{aligned}
\end{equation}
\end{lemma}
\begin{proof}
From $X_{n}-X=O_{p}(A_{n})$, there exists $M>0$ and $N>0$ such that we have
\begin{equation}
\begin{aligned}\label{eqt:stochastic bound fraction 1}
Pr\left\{\left|X_{n}-X\right|\geq M A_{n}\right\}\leq\epsilon\qquad\forall n\geq N\\
\Rightarrow Pr\left\{X_{n}\geq X+M A_{n}\right\}+Pr\left\{X_{n}\leq X-M A_{n}\right\}\leq\epsilon\qquad\forall n\geq N.
\end{aligned}
\end{equation}
Now, for any $n$ and $\bar{M}>0$, we have
\begin{equation}
\begin{aligned}\label{eqt:stochastic bound fraction 2}
&Pr\left\{\left|\frac{1}{X_{n}}-\frac{1}{X}\right|\geq \bar{M} A_{n}\right\}\qquad\forall n\\
=& Pr\left\{\left|X_{n}-X\right|\geq \bar{M} A_{n} X_{n}X\right\}\qquad\forall n\\
=& Pr\left\{\left|X_{n}-X\right|\geq \bar{M} A_{n} X_{n}X,\;X_{n}\geq X+\bar{M} A_{n}\right\}+Pr\left\{\left|X_{n}-X\right|\geq \bar{M} A_{n} X_{n}X,\;X_{n}\leq X+\bar{M} A_{n}\right\}\qquad\forall n\\
\leq& Pr\left\{\left|X_{n}-X\right|\geq \bar{M} A_{n} X(X+\bar{M} A_{n})\right\}+Pr\left\{\left|X_{n}-X\right|\geq \bar{M} A_{n} X_{n}X,\;X_{n}\leq X+\bar{M} A_{n}\right\}\qquad\forall n\\
=& Pr\left\{\left|X_{n}-X\right|\geq \bar{M} A_{n} X(X+\bar{M} A_{n})\right\}+ Pr\left\{\left|X_{n}-X\right|\geq \bar{M} A_{n} X_{n}X,\;X_{n}\leq X+\bar{M} A_{n},\;X_{n}\leq X-\bar{M} A_{n}\right\}\\
+& Pr\left\{\left|X_{n}-X\right|\geq \bar{M} A_{n} X_{n}X,\;X_{n}\leq X+\bar{M} A_{n},\;X_{n}\geq X-\bar{M} A_{n}\right\}\qquad\forall n\\
\leq& Pr\left\{\left|X_{n}-X\right|\geq \bar{M} A_{n} X(X+\bar{M} A_{n})\right\}+Pr\left\{X_{n}\leq X-\bar{M} A_{n}\right\}\\
+& Pr\left\{\left|X_{n}-X\right|\geq \bar{M} A_{n} X_{n}X,\;X-\bar{M} A_{n}\leq X_{n}\leq X+\bar{M} A_{n}\right\}\qquad\forall n.
\end{aligned}
\end{equation}
Note that 
\begin{equation}
\begin{aligned}\label{eqt:stochastic bound fraction 3}
&Pr\left\{\left|X_{n}-X\right|\geq \bar{M} A_{n}X(X+\bar{M} A_{n})\right\}=Pr\left\{\frac{\left|X_{n}-X\right|}{A_{n}}\geq \bar{M} X^{2}+\bar{M}^{2}XA_{n}\right\}\leq Pr\left\{\frac{\left|X_{n}-X\right|}{A_{n}}\geq \bar{M}X^{2}\right\}.
\end{aligned}
\end{equation}
Now, since $X>0\;a.s.$, $X^{2}>0\;a.s.$. Furthermore, there exists $\delta>0$ and $\bar{\delta}>0$ with $\bar{\delta}>\delta$ such that $Pr\{X^{2}\leq\delta\}\leq\epsilon$ and $Pr\{X^{2}\geq\bar{\delta}\}\leq\epsilon$ respectively. Besides, from \eqref{eqt:stochastic bound fraction 3}, we have
\begin{equation}
\begin{aligned}\label{eqt:stochastic bound fraction 4}
&Pr\left\{\frac{\left|X_{n}-X\right|}{A_{n}}\geq \bar{M} X^{2}\right\}\\
=&Pr\left\{\frac{\left|X_{n}-X\right|}{A_{n}}\geq \bar{M} X^{2},\;X^{2}\leq\delta\right\}+Pr\left\{\frac{\left|X_{n}-X\right|}{A_{n}}\geq \bar{M} X^{2},\;X^{2}>\delta\right\}\\
\leq &Pr\left\{X^{2}\leq\delta\right\}+Pr\left\{\frac{\left|X_{n}-X\right|}{A_{n}}\geq \bar{M}\delta\right\}.
\end{aligned}
\end{equation}
Simultaneously, we consider $Pr\left\{\left|X_{n}-X\right|\geq \bar{M} A_{n} X_{n}X,\;X-\bar{M} A_{n}\leq X_{n}\leq X+\bar{M} A_{n}\right\}$ $\forall n$. Indeed, we have
\begin{equation}
\begin{aligned}\label{eqt:stochastic bound fraction 3b}
&Pr\left\{\left|X_{n}-X\right|\geq \bar{M} A_{n} X_{n}X,\;X-\bar{M} A_{n}\leq X_{n}\leq X+\bar{M} A_{n}\right\}\\
=&Pr\left\{\left|X_{n}-X\right|\geq \bar{M} A_{n} X_{n}X,\;X-\bar{M} A_{n}\leq X_{n}\leq X+\bar{M} A_{n},\;X^{2}\leq\delta\right\}\\
+&Pr\left\{\left|X_{n}-X\right|\geq \bar{M} A_{n} X_{n}X,\;X-\bar{M} A_{n}\leq X_{n}\leq X+\bar{M} A_{n},\;X^{2}>\delta\right\}\\
\leq&Pr\left\{X^{2}\leq\delta\right\}+Pr\left\{\left|X_{n}-X\right|\geq \bar{M} A_{n} X_{n}X,\;X-\bar{M} A_{n}\leq X_{n}\leq X+\bar{M} A_{n},\;X^{2}>\delta,\;X^{2}\geq\bar{\delta}\right\}\\
+&Pr\left\{\left|X_{n}-X\right|\geq \bar{M} A_{n} X_{n}X,\;X-\bar{M} A_{n}\leq X_{n}\leq X+\bar{M} A_{n},\;X^{2}>\delta,\;X^{2}<\bar{\delta}\right\}\\
\leq&Pr\left\{X^{2}\leq\delta\right\}+Pr\left\{X^{2}\geq\bar{\delta}\right\}+Pr\left\{\left|X_{n}-X\right|\geq \bar{M} A_{n} X_{n}X,\;X-\bar{M} A_{n}\leq X_{n}\leq X+\bar{M} A_{n},\;\delta\leq X^{2}\leq \bar{\delta}\right\}\\
\leq&Pr\left\{X^{2}\leq\delta\right\}+Pr\left\{X^{2}\geq\bar{\delta}\right\}+Pr\left\{\left|X_{n}-X\right|\geq \bar{M} A_{n} (X-\bar{M}A_{n})X,\;\delta\leq X^{2}\leq\bar{\delta}\right\}\\
\leq&Pr\left\{X^{2}\leq\delta\right\}+Pr\left\{X^{2}\geq\bar{\delta}\right\}+Pr\left\{\frac{\left|X_{n}-X\right|}{A_{n}}\geq \bar{M}\delta- \bar{M}^{2} A_{n}\sqrt{\bar{\delta}}\right\}.
\end{aligned}
\end{equation}
Since $A_{n}$ is a sequence which converges to $0$, together with 
Combining \eqref{eqt:stochastic bound fraction 1}, \eqref{eqt:stochastic bound fraction 2}, \eqref{eqt:stochastic bound fraction 4} and \eqref{eqt:stochastic bound fraction 3b}, we can conclude that there exists $M^{'}$ and $N^{'}$ such that 
\begin{equation*}
\begin{aligned}
Pr\left\{\left|\frac{1}{X_{n}}-\frac{1}{X}\right|\geq M^{'}A_{n}\right\}\leq 5\epsilon \qquad \forall n\geq N^{'}
\end{aligned}
\end{equation*}
hence the claim is proved.
\end{proof}

Before presenting the error decomposition, we redefine some of the notations which are used in the sequel to reduce the computation complexity. They are
\begin{equation*}
\begin{aligned}
g_{i}&:=g(d^{i},\mathbf{U},\mathbf{Z}),\quad\hat{g}_{i}:=\hat{g}(d^{i},\mathbf{U},\mathbf{Z})\\
E_{i}&:=\mathbb{E}\left[\mathbf{1}_{\{D=d^{i}\}}\mid\mathbf{X},\mathbf{Z}\right],\quad\hat{E}_{i}:=\hat{\mathbb{E}}\left[\mathbf{1}_{\{D=d^{i}\}}\mid\mathbf{X},\mathbf{Z}\right]\\
P_{j}&:=Pr\{D=d^{j}\},\quad\hat{P}_{j}:=\hat{Pr}\{D=d^{j}\}.
\end{aligned}
\end{equation*}
Here, the estimated functions $\hat{g}_{i}$ and $\hat{E}_{i}$ are estimated based on an observational dataset. We omit the size of the observational dataset in defining the notations for the estimated functions. When we want to emphasize the computations of the $m^{\mathrm{th}}$ individual, we add either a subscript or a superscript $m$ to the notations in the first two rows of the above equations. As a consequence, the notations become
\begin{equation*}
\begin{aligned}
g^{m}_{i}&:=g(d^{i},\mathbf{U}_{m},\mathbf{Z}_{m}),\quad\hat{g}^{m}_{i}:=\hat{g}(d^{i},\mathbf{U}_{m},\mathbf{Z}_{m})\\
E_{i}^{m}&:=\mathbb{E}\left[\mathbf{1}_{\{D_{m}=d^{i}\}}\mid\mathbf{X}_{m},\mathbf{Z}_{m}\right],\quad\hat{E}_{i}^{m}:=\hat{\mathbb{E}}\left[\mathbf{1}_{\{D_{m}=d^{i}\}}\mid\mathbf{X}_{m},\mathbf{Z}_{m}\right].
\end{aligned}
\end{equation*}
%

Next, we do the error decomposition on the difference between $\mathbb{E}\left[g(d^{i},\mathbf{U},\mathbf{Z})\right]$ and the IwC estimator of $\theta^{i}$, which gives
\begin{align*}
&\mathbb{E}\left[g(d^{i},\mathbf{U},\mathbf{Z})\right]-\frac{1}{N}\sum_{m=1}^{N}g^{m}_{i}+\frac{1}{N}\sum_{m=1}^{N}\big[g_{i}^{m}-\hat{g}_{i}^{m}\big]+\frac{1}{N}\sum_{m=1}^{N}\mathbf{1}_{\{D_{m}=d^{i}\}}\left[\frac{(Y_{m}-g_{i}^{m})}{E_{i}^{m}}-\frac{(Y_{m}-\hat{g}_{i}^{m})}{\hat{E}_{i}^{m}}\right]\\
+&\left\{\mathbb{E}\left[\frac{\mathbf{1}_{\{D=d^{i}\}}(Y-g(d^{i},\mathbf{U},\mathbf{Z}))}{\mathbb{E}\left[\mathbf{1}_{\{D=d^{i}\}}\mid \mathbf{X},\mathbf{Z}\right]}\right]-\frac{1}{N}\sum_{m=1}^{N}\frac{\mathbf{1}_{\{D_{m}=d^{i}\}}(Y_{m}-g_{i}^{m})}{E_{i}^{m}}\right\}-\mathbb{E}\left[\frac{\mathbf{1}_{\{D=d^{i}\}}(Y-g(d^{i},\mathbf{U},\mathbf{Z}))}{\mathbb{E}\left[\mathbf{1}_{\{D=d^{i}\}}\mid \mathbf{X},\mathbf{Z}\right]}\right].
\end{align*}
Simplifying $\frac{1}{N}\overset{N}{\underset{m=1}{\sum}}\mathbf{1}_{\{D_{m}=d^{i}\}}\left[\frac{(Y_{m}-g_{i}^{m})}{E_{i}^{m}}-\frac{(Y_{m}-\hat{g}_{i}^{m})}{\hat{E}_{i}^{m}}\right]$ gives
\begin{equation*}
\begin{aligned}
\frac{1}{N}\sum_{m=1}^{N}\mathbf{1}_{\{D_{m}=d^{i}\}}\frac{(\hat{E}_{i}^{m}-E_{i}^{m})(Y_{m}-g_{i}^{m})+E_{i}^{m}(\hat{g}_{i}^{m}-g_{i}^{m})}{E_{i}^{m}(\hat{E}_{i}^{m}-E_{i}^{m})+(E_{i}^{m})^{2}}.
\end{aligned}
\end{equation*}
As a result, we have
\begin{subequations}
	\begin{align}
	&\mathbb{E}\left[g(d^{i},\mathbf{U},\mathbf{Z})\right]-\frac{1}{N}\sum_{m=1}^{N}g^{m}_{i}\label{eqt:consistency expectation 1}\\
	+&\frac{1}{N}\sum_{m=1}^{N}\big[g_{i}^{m}-\hat{g}_{i}^{m}\big]\label{eqt:consistency expectation 2}\\
	+&\mathbb{E}\left[\frac{\mathbf{1}_{\{D=d^{i}\}}(Y-g(d^{i},\mathbf{U},\mathbf{Z}))}{\mathbb{E}\left[\mathbf{1}_{\{D=d^{i}\}}\mid \mathbf{X},\mathbf{Z}\right]}\right]-\frac{1}{N}\sum_{m=1}^{N}\frac{\mathbf{1}_{\{D_{m}=d^{i}\}}(Y_{m}-g_{i}^{m})}{E_{i}^{m}}\label{eqt:consistency expectation 3}\\
	+&\frac{1}{N}\sum_{m=1}^{N}\mathbf{1}_{\{D_{m}=d^{i}\}}\frac{(\hat{E}_{i}^{m}-E_{i}^{m})(Y_{m}-g_{i}^{m})+E_{i}^{m}(\hat{g}_{i}^{m}-g_{i}^{m})}{E_{i}^{m}(\hat{E}_{i}^{m}-E_{i}^{m})+(E_{i}^{m})^{2}}\label{eqt:consistency expectation 4}\\
	-&\mathbb{E}\left[\frac{\mathbf{1}_{\{D=d^{i}\}}(Y-g(d^{i},\mathbf{U},\mathbf{Z}))}{\mathbb{E}\left[\mathbf{1}_{\{D=d^{i}\}}\mid \mathbf{X},\mathbf{Z}\right]}\right].\label{eqt:consistency expectation 5}
	\end{align}
\end{subequations}
Clearly, \eqref{eqt:consistency expectation 5} equals $0$. Expectations of \eqref{eqt:consistency expectation 1} and \eqref{eqt:consistency expectation 3} are $0$ using \ref{lemma:unbiasedness} of Lemma \ref{lemma:unbiasedness and consistency}. Moreover, \eqref{eqt:consistency expectation 1} and \eqref{eqt:consistency expectation 3} converge to $0$ in probability using \ref{lemma:consistency} of Lemma \ref{lemma:unbiasedness and consistency}. Consequently, errors mainly come from \eqref{eqt:consistency expectation 2} and \eqref{eqt:consistency expectation 4}. When $\hat{E}_{i}$ and $\hat{g}_{i}$ should converge to $E_{i}$ and $g_{i}$ well,
we know that \eqref{eqt:consistency expectation 2} and \eqref{eqt:consistency expectation 4} converge to $0$ in probability based on the Lemma \ref{lemma:stochastic_bound} and the Lemma \ref{lemma:stochastic bound fraction}. Consequently, $\hat{\theta}^{i}_{w}$ is a consistent estimator.

Here, $\hat{g}^{i}$, $\hat{E}^{i}$ and $\hat{P}_{j}$ are the estimated functions of $g^{i}$, $E^{i}$ and $P_{j}$ accordingly. In addition, we omit the dependence of $N$ when defining the notations of the estimated functions. Next, we do the error decomposition on the difference between $\mathbb{E}\left[g(d^{i},\mathbf{U},\mathbf{Z})\mid D=d^{j}\right]$ and the IwC estimator of $\theta^{i\mid j}$. To start with, we rewrite the IwC estimator of $\theta^{i\mid j}$. Indeed, the IwC estimator of $\theta^{i\mid j}$ can be rewritten as
\begin{align*}
&\frac{1}{N_{j}}\left\{\overset{N_{j}}{\underset{m=1}{\sum}} \hat{g}(d^{i},\mathbf{U}_{m}^{j},\mathbf{Z}_{m}^{j})+\overset{N_{i}}{\underset{m=1}{\sum}}\frac{\hat{\mathbb{E}}\left[\mathbf{1}_{\{D_{m}=d^{j}\}}\mid\mathbf{X}_{m}^{i},\mathbf{Z}_{m}^{i}\right]}{\hat{\mathbb{E}}\left[\mathbf{1}_{\{D_{m}=d^{i}\}}\mid\mathbf{X}_{m}^{i},\mathbf{Z}_{m}^{i}\right]}\left[Y_{m}^{i}-\hat{g}(d^{i},\mathbf{U}_{m}^{i},\mathbf{Z}_{m}^{i})\right]\right\}\nonumber\\
=&\frac{\frac{1}{N}\overset{N}{\underset{m=1}{\sum}}\mathbf{1}_{\{D_{m}=d^{j}\}}\hat{g}(d^{i},\mathbf{U}_{m},\mathbf{Z}_{m})}{\frac{N_{j}}{N}}+\frac{\frac{1}{N}\overset{N}{\underset{m=1}{\sum}}\frac{\mathbf{1}_{\{D_{m}=d^{i}\}}\hat{\mathbb{E}}\big[\mathbf{1}_{\{D_{m}=d^{j}\}}\mid\mathbf{X}_{m},\mathbf{Z}_{m}\big]}{\hat{\mathbb{E}}\big[\mathbf{1}_{\{D_{m}=d^{i}\}}\mid\mathbf{X}_{m},\mathbf{Z}_{m}\big]}(Y_{m}-\hat{g}(d^{i},\mathbf{U}_{m},\mathbf{Z}_{m}))}{\frac{N_{j}}{N}}\nonumber\\
=&\frac{\frac{1}{N}\overset{N}{\underset{m=1}{\sum}}\mathbf{1}_{\{D_{m}=d^{j}\}}\hat{g}_{i}^{m}}{\frac{N_{j}}{N}}+\frac{\frac{1}{N}\overset{N}{\underset{m=1}{\sum}}\frac{\mathbf{1}_{\{D_{m}=d^{i}\}}\hat{E}_{j}^{m}}{\hat{E}_{i}^{m}}(Y_{m}-\hat{g}_{i}^{m})}{\frac{N_{j}}{N}}.
\end{align*}
Instead of undergoing the error decomposition on the difference between $\mathbb{E}\left[g(d^{i},\mathbf{U},\mathbf{Z})\mid D=d^{j}\right]$ and the IwC estimator of $\theta^{i\mid j}$, we consider the error decomposition on the difference between $\mathbb{E}\left[g(d^{i},\mathbf{U},\mathbf{Z})\mid D=d^{j}\right]$ and $\left\{\frac{\frac{1}{N}\overset{N}{\underset{m=1}{\sum}}\mathbf{1}_{\{D_{m}=d^{j}\}}\hat{g}_{i}^{m}}{\hat{P}_{j}}+\frac{\frac{1}{N}\overset{N}{\underset{m=1}{\sum}}\frac{\mathbf{1}_{\{D_{m}=d^{i}\}}\hat{E}_{j}^{m}}{\hat{E}_{i}^{m}}(Y_{m}-\hat{g}_{i}^{m})}{\hat{P}_{j}}\right\}$, which gives
\begin{align*}
&\frac{\mathbb{E}\left[g(d^{i},\mathbf{U},\mathbf{Z})\mathbf{1}_{\{D=d^{j}\}}\right]}{P_{j}}-\frac{\frac{1}{N}\overset{N}{\underset{m=1}{\sum}}\mathbf{1}_{\{D_{m}=d^{j}\}}g_{i}^{m}}{P_{j}}\\
+&\frac{\frac{1}{N}\overset{N}{\underset{m=1}{\sum}}\mathbf{1}_{\{D_{m}=d^{j}\}}g_{i}^{m}}{P_{j}}-\frac{\frac{1}{N}\overset{N}{\underset{m=1}{\sum}}\mathbf{1}_{\{D_{m}=d^{j}\}}\hat{g}_{i}^{m}}{\hat{P}_{j}}\\
+&\frac{\mathbb{E}\left[\frac{\mathbf{1}_{\{D=d^{i}\}}\mathbb{E}\left[\mathbf{1}_{\{D=d^{j}\}}\mid \mathbf{X},\mathbf{Z}\right](Y-g(d^{i},\mathbf{U},\mathbf{Z}))}{\mathbb{E}\left[\mathbf{1}_{\{D=d^{i}\}}\mid \mathbf{X},\mathbf{Z}\right]}\right]}{P_{j}}-\frac{\frac{1}{N}\overset{N}{\underset{m=1}{\sum}}\frac{\mathbf{1}_{\{D_{m}=d^{i}\}}E_{j}^{m}}{E_{i}^{m}}(Y_{m}-g_{i}^{m})}{P_{j}}
\\
+&\frac{\frac{1}{N}\overset{N}{\underset{m=1}{\sum}}\frac{\mathbf{1}_{\{D_{m}=d^{i}\}}E_{j}^{m}}{E_{i}^{m}}(Y_{m}-g_{i}^{m})}{P_{j}}-\frac{\frac{1}{N}\overset{N}{\underset{m=1}{\sum}}\frac{\mathbf{1}_{\{D_{m}=d^{i}\}}\hat{E}_{j}^{m}}{\hat{E}_{i}^{m}}(Y_{m}-\hat{g}_{i}^{m})}{\hat{P}_{j}}\\
-&\frac{\mathbb{E}\left[\frac{\mathbf{1}_{\{D=d^{i}\}}\mathbb{E}\left[\mathbf{1}_{\{D=d^{j}\}}\mid \mathbf{X},\mathbf{Z}\right](Y-g(d^{i},\mathbf{U},\mathbf{Z}))}{\mathbb{E}\left[\mathbf{1}_{\{D=d^{i}\}}\mid \mathbf{X},\mathbf{Z}\right]}\right]}{\mathbb{E}\left[\mathbf{1}_{\{D=d^{j}\}}\right]}
\end{align*}
We simplify the quantities $\frac{\frac{1}{N}\overset{N}{\underset{m=1}{\sum}}\mathbf{1}_{\{D_{m}=d^{j}\}}g_{i}^{m}}{P_{j}}-\frac{\frac{1}{N}\overset{N}{\underset{m=1}{\sum}}\mathbf{1}_{\{D_{m}=d^{j}\}}\hat{g}_{i}^{m}}{\hat{P}_{j}}$ and $\frac{\frac{1}{N}\overset{N}{\underset{m=1}{\sum}}\frac{\mathbf{1}_{\{D_{m}=d^{i}\}}E_{j}^{m}}{E_{i}^{m}}(Y_{m}-g_{i}^{m})}{P_{j}}-\frac{\frac{1}{N}\overset{N}{\underset{m=1}{\sum}}\frac{\mathbf{1}_{\{D_{m}=d^{i}\}}\hat{E}_{j}^{m}}{\hat{E}_{i}^{m}}(Y_{m}-\hat{g}_{i}^{m})}{\hat{P}_{j}}$. Indeed, we have
\begin{equation*}
\begin{aligned}
\frac{\frac{1}{N}\overset{N}{\underset{m=1}{\sum}}\mathbf{1}_{\{D_{m}=d^{j}\}}g_{i}^{m}}{P_{j}}-\frac{\frac{1}{N}\overset{N}{\underset{m=1}{\sum}}\mathbf{1}_{\{D_{m}=d^{j}\}}\hat{g}_{i}^{m}}{\hat{P}_{j}}=\frac{1}{N}\frac{\overset{N}{\underset{m=1}{\sum}}\mathbf{1}_{\{D_{m}=d^{j}\}}\hat{P}_{j}(g_{i}^{m}-\hat{g}_{i}^{m})+(\hat{P}_{j}-P_{j})\hat{g}_{i}^{m}}{(\hat{P}_{j}-P_{j})P_{j}+P_{j}^{2}}
\end{aligned}
\end{equation*}
and
\begin{equation*}
\begin{aligned}
&\frac{\frac{1}{N}\overset{N}{\underset{m=1}{\sum}}\frac{\mathbf{1}_{\{D_{m}=d^{i}\}}E_{j}^{m}}{E_{i}^{m}}(Y_{m}-g_{i}^{m})}{P_{j}}-\frac{\frac{1}{N}\overset{N}{\underset{m=1}{\sum}}\frac{\mathbf{1}_{\{D_{m}=d^{i}\}}\hat{E}_{j}^{m}}{\hat{E}_{i}^{m}}(Y_{m}-\hat{g}_{i}^{m})}{\hat{P}_{j}}\\
=&\frac{1}{N}\overset{N}{\underset{m=1}{\sum}}\mathbf{1}_{\{D_{m}=d^{i}\}}\frac{(\hat{P}_{j}-P_{j})(\hat{E}_{i}^{m}-E_{i}^{m})E_{j}^{m}Y_{m}+(\hat{P}_{j}-P_{j})E_{i}^{m}E_{j}^{m}Y_{m}}{[(\hat{P}_{j}-P_{j})E_{i}^{m}-P_{j}(\hat{E}_{i}^{m}-E_{i}^{m})]P_{j}E_{i}^{m}+(P_{j}E_{i}^{m})^{2}}\\
+&\frac{1}{N}\overset{N}{\underset{m=1}{\sum}}\mathbf{1}_{\{D_{m}=d^{i}\}}\frac{P_{j}[(\hat{E}_{i}^{m}-E_{i}^{m})E_{j}^{m}+E_{i}^{m}(E_{j}^{m}-\hat{E}_{j}^{m})]Y_{m}}{[(\hat{P}_{j}-P_{j})E_{i}^{m}-P_{j}(\hat{E}_{i}^{m}-E_{i}^{m})]P_{j}E_{i}^{m}+(P_{j}E_{i}^{m})^{2}}\\
+&\frac{1}{N}\overset{N}{\underset{m=1}{\sum}}\mathbf{1}_{\{D_{m}=d^{i}\}}\frac{P_{j}E_{i}^{m}(\hat{g}_{i}^{m}-g_{i}^{m})+[P_{j}^{m}(E_{i}^{m}-\hat{E}_{i}^{m})+(P_{j}-\hat{P}_{j})\hat{E}_{i}^{m}]g_{i}^{m}}{[(\hat{P}_{j}-P_{j})E_{i}^{m}-P_{j}(\hat{E}_{i}^{m}-E_{i}^{m})]P_{j}E_{i}^{m}+(P_{j}E_{i}^{m})^{2}}
\end{aligned}
\end{equation*}
As a result, $\mathbb{E}\left[g(d^{i},\mathbf{U},\mathbf{Z})\mid D=d^{j}\right]-\left\{\frac{\frac{1}{N}\overset{N}{\underset{m=1}{\sum}}\mathbf{1}_{\{D_{m}=d^{j}\}}\hat{g}_{i}^{m}}{\hat{P}_{j}}+\frac{\frac{1}{N}\overset{N}{\underset{m=1}{\sum}}\frac{\mathbf{1}_{\{D_{m}=d^{i}\}}\hat{E}_{j}^{m}}{\hat{E}_{i}^{m}}(Y_{m}-\hat{g}_{i}^{m})}{\hat{P}_{j}}\right\}$ can be decomposed as
\begin{subequations}
	\begin{align}
	&\frac{\mathbb{E}\left[g(d^{i},\mathbf{U},\mathbf{Z})\mathbf{1}_{\{D=d^{j}\}}\right]}{P_{j}}-\frac{\frac{1}{N}\overset{N}{\underset{m=1}{\sum}}\mathbf{1}_{\{D_{m}=d^{j}\}}g_{i}^{m}}{P_{j}}\label{eqt:consistency conditional expectation 1}\\
	+&\frac{1}{N}\frac{\overset{N}{\underset{m=1}{\sum}}\mathbf{1}_{\{D_{m}=d^{j}\}}\hat{P}_{j}(g_{i}^{m}-\hat{g}_{i}^{m})+(\hat{P}_{j}-P_{j})\hat{g}_{i}^{m}}{(\hat{P}_{j}-P_{j})P_{j}+P_{j}^{2}}\label{eqt:consistency conditional expectation 2}\\
	+&\frac{\mathbb{E}\left[\frac{\mathbf{1}_{\{D=d^{i}\}}\mathbb{E}\left[\mathbf{1}_{\{D=d^{j}\}}\mid \mathbf{X},\mathbf{Z}\right](Y-g(d^{i},\mathbf{U},\mathbf{Z}))}{\mathbb{E}\left[\mathbf{1}_{\{D=d^{i}\}}\mid \mathbf{X},\mathbf{Z}\right]}\right]}{P_{j}}-\frac{\frac{1}{N}\overset{N}{\underset{m=1}{\sum}}\frac{\mathbf{1}_{\{D_{m}=d^{i}\}}E_{j}^{m}}{E_{i}^{m}}(Y_{m}-g_{i}^{m})}{P_{j}}
	\label{eqt:consistency conditional expectation 3}\\
	+&\frac{1}{N}\overset{N}{\underset{m=1}{\sum}}\mathbf{1}_{\{D_{m}=d^{i}\}}\frac{(\hat{P}_{j}-P_{j})(\hat{E}_{i}^{m}-E_{i}^{m})E_{j}^{m}Y_{m}+(\hat{P}_{j}-P_{j})E_{i}^{m}E_{j}^{m}Y_{m}}{[(\hat{P}_{j}-P_{j})E_{i}^{m}-P_{j}(\hat{E}_{i}^{m}-E_{i}^{m})]P_{j}E_{i}^{m}+(P_{j}E_{i}^{m})^{2}}\nonumber\\
	+&\frac{1}{N}\overset{N}{\underset{m=1}{\sum}}\mathbf{1}_{\{D_{m}=d^{i}\}}\frac{P_{j}[(\hat{E}_{i}^{m}-E_{i}^{m})E_{j}^{m}+E_{i}^{m}(E_{j}^{m}-\hat{E}_{j}^{m})]Y_{m}}{[(\hat{P}_{j}-P_{j})E_{i}^{m}-P_{j}(\hat{E}_{i}^{m}-E_{i}^{m})]P_{j}E_{i}^{m}+(P_{j}E_{i}^{m})^{2}}\nonumber\\
	+&\frac{1}{N}\overset{N}{\underset{m=1}{\sum}}\mathbf{1}_{\{D_{m}=d^{i}\}}\frac{P_{j}E_{i}^{m}(\hat{g}_{i}^{m}-g_{i}^{m})+[P_{j}^{m}(E_{i}^{m}-\hat{E}_{i}^{m})+(P_{j}-\hat{P}_{j})\hat{E}_{i}^{m}]g_{i}^{m}}{[(\hat{P}_{j}-P_{j})E_{i}^{m}-P_{j}(\hat{E}_{i}^{m}-E_{i}^{m})]P_{j}E_{i}^{m}+(P_{j}E_{i}^{m})^{2}}\label{eqt:consistency conditional expectation 4}\\
	-&\frac{\mathbb{E}\left[\frac{\mathbf{1}_{\{D=d^{i}\}}\mathbb{E}\left[\mathbf{1}_{\{D=d^{j}\}}\mid \mathbf{X},\mathbf{Z}\right](Y-g(d^{i},\mathbf{U},\mathbf{Z}))}{\mathbb{E}\left[\mathbf{1}_{\{D=d^{i}\}}\mid \mathbf{X},\mathbf{Z}\right]}\right]}{\mathbb{E}\left[\mathbf{1}_{\{D=d^{j}\}}\right]}.\label{eqt:consistency conditional expectation 5}
	\end{align}
\end{subequations}
Using the similar assumptions and arguments in proving the consistency of the IwC estimator of $\theta^{i}$, we can conclude that the IwC estimator of $\theta^{i\mid j}$ is also a consistent estimator.

\end{document}